\documentclass[a4paper,twocolumn,10pt,accepted=2023-08-02]{quantumarticle}
\pdfoutput=1
\usepackage[utf8]{inputenc}
\usepackage[english]{babel}
\usepackage[T1]{fontenc}
\usepackage{amsmath}
\usepackage{hyperref}

\usepackage{tikz}
\usepackage{lipsum}

\usepackage{graphicx}
\usepackage{indentfirst}
\usepackage{physics}
\usepackage{braket}
\usepackage{float}
\usepackage{amsmath}
\usepackage{amsthm}
\usepackage{wasysym}
\usepackage{mathrsfs}
\usepackage{bm}
\usepackage{epstopdf}
\usepackage{mathtools}
\usepackage{CJK}
\usepackage{esint}
\usepackage{color}
\usepackage{subfigure}
\usepackage{amsfonts}
\usepackage{footmisc}
\usepackage{scrextend}
\usepackage{multirow}
\usepackage{amssymb}
\usepackage{multirow}
\usepackage{diagbox}

\usepackage{url}
\usepackage{bm}

\usepackage{xcolor}
\definecolor{darkblue}{rgb}{0,0,0.5}
\hypersetup{
colorlinks=true,
linkcolor=black,
filecolor=blue,
citecolor=darkblue,  
urlcolor=black,
}
\newcommand{\bk}{\bm k}
\newcommand{\bs}{\bm s}

\newcommand{\bx}{\boldsymbol x}

\newcommand{\bI}{\boldsymbol{I}}

\newcommand*\diff{\mathop{}\!\mathrm{d}}

\newcommand{\QZ}[1]{{{\textcolor{black}{#1}}}}

\newtheorem{theorem}{Theorem}
\newtheorem*{conjecture*}{Conjecture}

\newtheorem{lemma}[theorem]{Lemma}
\newtheorem{corollary}{Corollary}[theorem]

\begin{document}

\title{Optimal encoding of oscillators into more oscillators}

\author{Jing Wu}
\thanks{These two authors contributed equally.}
\affiliation{James C. Wyant College of Optical Sciences, University of Arizona, Tucson, AZ 85721, USA}
\author{Anthony J. Brady}
\thanks{These two authors contributed equally.}
\affiliation{Department of Electrical and Computer Engineering, University of Arizona, Tucson, Arizona 85721, USA}
\author{Quntao Zhuang}
\email{qzhuang@usc.edu}
\orcid{0000-0002-9554-3846}
\affiliation{
Ming Hsieh Department of Electrical and Computer Engineering \& Department of Physics and Astronomy, University of Southern California, Los
Angeles, California 90089, USA
}
\affiliation{James C. Wyant College of Optical Sciences, University of Arizona, Tucson, AZ 85721, USA}
\affiliation{Department of Electrical and Computer Engineering, University of Arizona, Tucson, Arizona 85721, USA}

\begin{abstract}
Bosonic encoding of quantum information into harmonic oscillators is a hardware efficient approach to battle noise. In this regard, oscillator-to-oscillator codes not only provide an additional opportunity in bosonic encoding, but also extend the applicability of error correction to continuous-variable states ubiquitous in quantum sensing and communication. In this work, we derive the optimal oscillator-to-oscillator codes among the general family of Gottesman-Kitaev-Preskill (GKP)-stablizer codes for homogeneous noise. We prove that an arbitrary GKP-stabilizer code can be reduced to a generalized GKP two-mode-squeezing (TMS) code. The optimal encoding to minimize the geometric mean error can be constructed from GKP-TMS codes with an optimized GKP lattice and TMS gains. For single-mode data and ancilla, this optimal code design problem can be efficiently solved, and we further provide numerical evidence that a hexagonal GKP lattice is optimal and strictly better than the previously adopted square lattice. For the multimode case, general GKP lattice optimization is challenging. In the two-mode data and ancilla case, we identify the D4 lattice---a 4-dimensional dense-packing lattice---to be superior to a product of lower dimensional lattices. As a by-product, the code reduction allows us to prove a universal no-threshold-theorem for arbitrary oscillators-to-oscillators codes based on Gaussian encoding, even when the ancilla are not GKP states.
\end{abstract}
\maketitle

\section{Introduction}

The power of quantum information processing comes from delicate quantum effects such as coherence, squeezing and entanglement. As noise is ubiquitous, maintaining such power relies on quantum error correction. Since the early works~\cite{calderbank1996,steane1996multiple}, various codes and the corresponding error correction systems have been proposed, mostly focusing on discrete-variable physical realizations. The Gottesman-Kitaev-Preskill (GKP) code~\cite{gkp2000} provides an alternative approach to utilize the bosonic degree of freedom of an oscillator to encode a qubit. Such a qubit-to-oscillator encoding is hardware efficient, as the entire degree of freedom of a (cavity) mode is utilized and at the same time, microwave cavities have a long lifetime~\cite{romanenko2020}. By this approach, the lifetime of a logical qubit has been extended beyond the error correction `break-even' point for cat codes~\cite{ofek2016extending} and more recently for GKP codes~\cite{sivak2022breakeven}.

The concatenation of bosonic codes and qubit codes has also shown great promise~\cite{raveendran2022finite}, with applications in quantum repeaters~\cite{rozpkedek2021quantum} and quantum computer architecture~\cite{chamberland2022building}. It then becomes a natural question whether an oscillator can be encoded into multiple oscillators to gain additional advantage in protecting logical data. Noh, Girvin and Jiang provide an affirmative answer by designing a GKP-stabilizer code~\cite{noh2019encoding,noh2020o2o} to protect an oscillator by entangling it with an additional GKP ancilla. Despite that these oscillator-to-oscillator codes alone do not have a threshold theorem due to their analog nature and finite squeezing~\cite{hanggli2021oscillator}, an oscillator-to-oscillator encoding at the bottom-layer can substantially suppress the error in a multi-level qubit encoding~\cite{xu2022qubit}. Moreover, the GKP-stabilizer code can be utilized in general to protect an oscillator in an arbitrary quantum state, including a squeezed vacuum state and continuous-variable multi-partite entangled states that are widely applicable to quantum sensing~\cite{zhuang2020distributed,zhou2022enhancing} and communication~\cite{wu2022continuous}.

To \QZ{maximally} benefit from the oscillator-to-oscillator encoding for various tasks mentioned above, it is important to find the optimal GKP-stabilizer code design. In this paper, we make progress towards this endeavour for the commonly considered homogeneous noise case~\cite{noh2020o2o}. We prove that an arbitrary GKP-stabilizer code can be reduced to a generalized GKP two-mode-squeezing (TMS) code. Therefore, GKP-TMS codes with optimized GKP ancilla and gains \QZ{can achieve the minimum geometric mean error}. For decoding, we derive the minimum mean square error (MMSE) estimator, which is superior to the linear estimator~\cite{noh2020o2o} \QZ{in terms of minimizing the residue noise on the data modes}. While linear estimation leads to a break-even point (\QZ{the point above which noise can no longer be reduced by error correction}) of additive noise at $\sigma_{\rm lin}^\star\sim 0.558$, the MMSE estimation pushes it to $\sigma_{\rm MMSE}^\star= 0.605(5)$, which is just \QZ{at the edge of the best} break-even region for arbitrary GKP codes $[0.607,0.707]$ derived from our quantum capacity analyses.

For single-mode data and ancilla, we further show that the optimal code design problem \QZ{of minimizing the geometric mean error} can be efficiently solved and provide numerical evidence that a hexagonal GKP lattice is optimal and strictly better than the previously adopted square lattice. 
For the multimode case, solving the optimal lattice is more challenging. For two-mode data and ancilla, we identify the D4 lattice~\cite{baptiste2022multiGKP, eisert2022lattice}---a 4-dimensional dense-packing lattice---to be superior to direct products of 2-dimensional lattices. \QZ{For a single data mode with two ancilla modes, we compare a few examples of lattices and find the performance to be dependent on the noise levels.} Our results indicate that high-dimensional lattices have the potential of outperforming low-dimensional ones for GKP-stablizer codes.

\QZ{
Besides the results on the optimal code design, the code reduction also allows us to prove a more general version of the no-threshold-theorem of Ref.~\cite{hanggli2021oscillator}, where the original proof relies on explicit maximum likelihood error decoder based on GKP-type syndrome information. Our proof is based on a simple, classical information-theoretical argument; moreover, our no-threshold-theorem applies to all GKP-stabilizer codes, and more generally, even when the GKP ancilla are replaced by general non-Gaussian states.
}

\QZ{This paper is organized as the following. Section~\ref{sec:general_code} introduces the general GKP-stabilizer code for encoding an oscillator into many oscillators. Our main theorem of code reduction and optimality is presented in Section~\ref{sec:general_reduction}, which then leads to the no-threshold theorem in Section~\ref{sec:code_reduction_no_threshold}. Section~\ref{sec:gkp_lattices_mmse} introduces general GKP lattices and the MMSE estimation. The final part of the paper addresses code optimization and comparison, with single-mode case in Section~\ref{sec:single_mode} and multi-mode case in Section~\ref{sec:multi_mode}. We end the paper with discussions on the heterogeneous case in Section~\ref{sec:heterogeneous} and the imperfect GKP states in Section~\ref{sec:approx_gkp}.}

\begin{figure}[t]
    \centering
    \includegraphics[width=\linewidth]{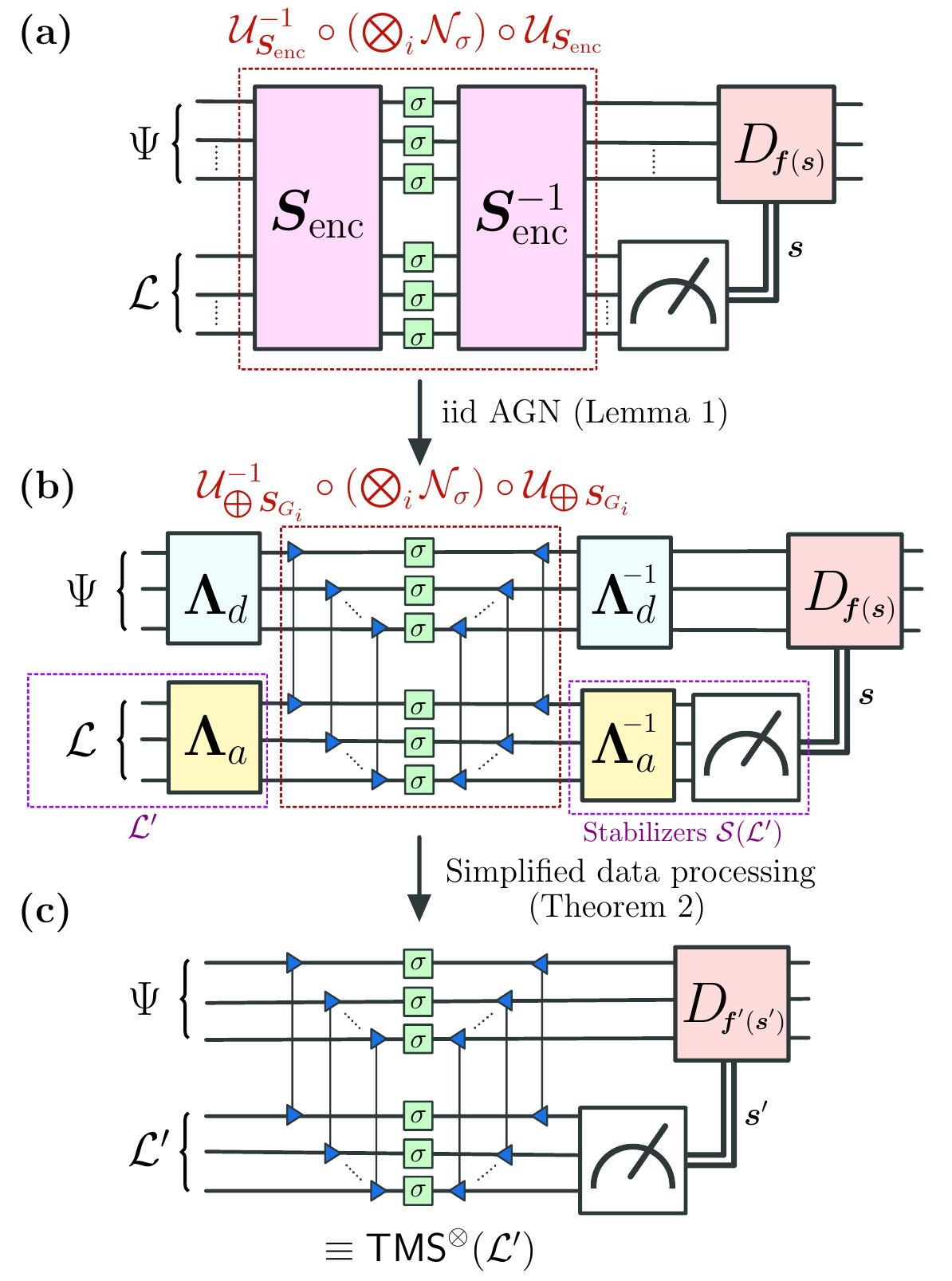}
    \caption{Illustration of GKP-stabilizer code reduction (here, $M=N$). (a) A general GKP-stabilizer code with encoding $\bm S_{\mathsf{enc}}$ and an ancillary GKP lattice $\mathcal{L}$. The syndromes $\bm s$ are extracted from stabilizer measurements on the ancillary lattice $\mathcal{L}$ and inform the corrective displacements on the data $D_{\bm f(\bm s)}$. (b) Reduction of the GKP stabilizer code to a set of TMS operations between the data and ancilla modes, up to local symplectic transformations $\bm\Lambda_d$ and $\bm \Lambda_a$. (c) Coherent data processing via $\bm\Lambda_d$ does not change the performance of the code. Acting on the initial ancillary lattice $\mathcal{L}$ by a symplectic transformation $\bm\Lambda_a$ produces another lattice $\mathcal{L}^\prime$. Thus, a general GKP-stabilizer code reduces to a direct product of TMS codes with a GKP ancillary lattice $\mathcal{L}^\prime$, $\mathsf{TMS}^{\otimes}(\mathcal{L}^\prime)$.}
    \label{fig:code_redux}
\end{figure}

\section{General GKP-stabilizer code}
\label{sec:general_code}

As shown in Fig.~\ref{fig:code_redux} (a), a general GKP-stabilizer code operates as follows. An arbitrary $N$-mode data system in quantum state $\Psi$ is encoded into $K= M+N$ modes by applying a Gaussian unitary to entangle the data system with an $M$-mode ancilla system in a non-Gaussian state $\cal L$. In general, non-Gaussian states are required, due to the no-go theorem of Gaussian error correction~\cite{cerf2009nogo}. Typically, we are interested in \QZ{(canonical)} GKP lattice states, hence $\cal L$ for lattice. The Gaussian unitary $U_{\bm S_{\mathsf{enc}}}$ can be described by the symplectic transform $\bm S_{\mathsf{enc}}$ (see Appendix~\ref{app:gauss_evol} for a brief introduction of Gaussian unitaries). We denote the corresponding unitary channel as $\mathcal{U}_{\bm S_{\mathsf{enc}}}$.

\QZ{At this juncture, we want to emphasize that the GKP lattice states appearing in this paper are \emph{canonical} lattice states as opposed to \emph{computational} lattice states. Computational lattice states allow one to encode a qubit (or qudit) into an oscillator per the original GKP approach~\cite{gkp2000}. Canonical lattices have larger spacing than their computational counterparts and, consequently, cannot support digital information. However, canonical lattices can be quite useful in multimode codes and, particularly, for oscillators-to-oscillators codes, as demonstrated in Refs~\cite{noh2020o2o,wu2021continuous}. Therefore, when considering GKP states in this paper, we refer solely to canonical lattices. We describe the mathematical and technical details of canonical lattice states in Section~\ref{sec:gkp_lattices_mmse}.}

As explained in the original GKP paper~\cite{gkp2000}, a natural noise model for bosonic systems is random displacements. In practice, such noise arises from amplifying a lossy signal, leading to additive Gaussian noise (AGN)---i.e. the random displacements have a Gaussian distribution. As general correlated Gaussian noises can be reduced to independent noises~\cite{wu2021continuous}, it suffices to consider a product of single-mode AGNs. Formally, we denote the single-mode AGN channel as $\cal N_\sigma$, where $\sigma^2$ is the variance of the displacement noise on a single mode.

On the decoding side, an inverse Gaussian unitary described by $\bm S_{\mathsf{enc}}^{-1}$ is first applied to disentangle the data and the ancilla. Then one measures the ancilla system to perform error correcting displacement operations ${D}_{\bm f(\bm s)}$ based on the measurement results (or \textit{syndromes}) $\bm s$~\footnote{Note that the measurement and displacements can be equivalently pushed before the inverse unitary, as shown in Ref.~\cite{xu2022qubit}.}. Here, $\bm f$ is a vector function (i.e., $\bm f:\mathbb{R}^{2M}\rightarrow\mathbb{R}^{2N}$) that takes the syndromes $\bm s$ as input and provides an estimate for the error displacements on the data. The corrective displacements aim to cancel the additive noise on the data. Due to the analog nature of the errors, such a cancellation is never perfect, and there will be residual, random displacements $\bm \xi_d$ on the output data modes.

To quantify the error correction performance, we evaluate the covariance matrix $\bm V_{\rm out}$ of the residue displacements $\bm \xi_d$, despite the distribution of $\bm \xi_d$ being non-Gaussian after decoding. Note that, prior to encoding, one may also apply a Gaussian unitary $U_{\bm S_d}$ and then apply the inverse $U_{\bm S_d}^{-1}$ after the final decoding step. Pre- and post-processing transform the residue noise covariance matrix as $\bm S_d^{-1} \bm V_{\rm out} \bm S_d^{-\top}$. Such an operation is considered `free' and should not improve the performance of the error correction. Therefore, we consider the geometric mean (GM) error
\begin{equation}
    \bar\sigma_{ \rm GM}^2\equiv\sqrt[2N]{\det\bm V_{\rm out}},
    \label{eq:sigma_GM}
\end{equation}
as the figure of merit to benchmark code performance,
which is necessarily invariant under symplectic operations on the data. Moreover, the GM error has information theoretic roots, as it relates to a lower bound on the quantum capacity for the additive non-Gaussian noise channel of a multimode GKP code. \QZ{Given the geometric mean error of general additive noise $\bar\sigma_{ \rm GM}^2$, the quantum capacity of the $N$-mode additive noise channel $C_{\mathcal{Q}}\geq \max\left[0,N\log_2\left(\frac{1}{e\bar{\sigma}^{2}_{\rm GM}}\right)\right] $ (see Lemma~\ref{lemma:multimode_lb} of Appendix~\ref{app:bounds})}.

As an alternative simpler metric, we also consider the root-mean-square (RMS) error, 
\begin{equation}
   \bar\sigma_{ \rm RMS}^2\equiv\frac{\Tr{\bm V_{\rm out}}}{2N}.
   \label{eq:sigma_RMS}
\end{equation}
The RMS error is only invariant under orthogonal symplectic (e.g., linear optical) transformations on the data. However, it is easier to evaluate and provides an upper bound for the GM error, $\bar\sigma_{ \rm RMS}^2\ge  \bar\sigma_{ \rm GM}^2$.

Before moving to our main results, we specify a few code examples. In Ref.~\cite{noh2020o2o}, two codes are proposed based on the \QZ{canonical} GKP square lattice, the GKP-TMS code described by the TMS symplectic matrix
\begin{equation}\label{eq:tms_matrix_main}
    \bm S_{G}=
    \begin{pmatrix}
        \sqrt{G}\bm I & \sqrt{G-1}\bm Z\\
        \sqrt{G-1}\bm Z & \sqrt{G}\bm I
    \end{pmatrix},
\end{equation}
with the gain $G$ tunable, and the GKP-squeezed-repetition code (also see Ref.~\cite{zhuang2020distributed}), which has the following encoding matrix for $N=2$ modes,
\begin{align}
    \label{eq:rep12}
    &
    \bm{S}^{[2]}_{\rm Sq-Rep} \equiv
    \begin{pmatrix}
        1/\lambda && 0 && 0 && 0\\
        0 && \lambda && 0 && -\lambda\\
        \lambda && 0 && \lambda && 0\\
        0 && 0 && 0 && 1/\lambda
      \end{pmatrix},
\end{align}
with $\lambda$ being a tunable parameter.

\section{General reduction of encoding}
\label{sec:general_reduction}

We focus on the homogeneous noise model, where the displacement error on all modes are independent and identically distributed (iid), where the overall $K$-mode noise channel $\bigotimes_{i=1}^K\mathcal{N}_{\sigma}$.

Towards proving the optimal code design, we begin with the following lemma to simplify a general code.

\begin{lemma}\label{lemma:lemma_redux}
For an iid AGN channel $\bigotimes_{i=1}^K\mathcal{N}_\sigma$ and up to local Gaussian unitaries acting on all data or all ancilla modes, a GKP-stabilizer code with \QZ{any} symplectic encoding matrix $\bm S_{\mathsf{enc}}$ reduces to a direct product of $N$ TMS operations (between the $N$ data modes and $N$ ancilla modes) together with an identity operation on the remaining $M-N$ ancilla modes---i.e., $\bm S_{\rm enc}\rightarrow\bigoplus_{i=1}^N\bm S_{G_i}\oplus\bm I_{2(M-N)}$. 
\end{lemma}

In other words, we can decompose a general GKP-stabilizer code into TMS operations and local symplectic operations; see Fig.~\ref{fig:code_redux} (a-b) for a visual aid of the lemma for the $M=N$ case. The local Gaussian unitary $U_{\bm\Lambda_d}$ and its inverse are applied on the data modes, while the local Gaussian unitary $U_{\bm\Lambda_a}$ and its inverse are applied on the ancilla modes. Consequently, the encoding and decoding can be taken as simple product of TMS operations. The proof of this result is based on Gaussian channel synthesis and the modewise entanglement theorem~\cite{reznik2003modewise} (Theorem~\ref{thm:synthesis} and  Theorem~\ref{thm:modewise} of Appendix~\ref{app:gauss_evol}, respectively), which we present in Appendix~\ref{app:proof_lemma_redux} in full detail.

The above lemma reduces the number of parameters in code description from $2(M+N)^2$ to $2M^2+2N^2$ in the leading order. In fact, as we will explain later, coherent data processing (via $\bm \Lambda_d$) only reshapes the residual noise; it does not aid in error correction. Hence a further reduction to $2M^2$ is permissible. Moreover, the generally multimode entangling operations between data and ancilla are now given by standard TMS operations with $N$ to-be-determined gain parameters $G_i$.

We now consider the GM error as a performance metric and arrive at the following main result of the work (see Fig.~\ref{fig:code_redux} (c)).

\begin{theorem}[Sufficiency of $\mathsf{TMS}^\otimes(\mathcal{L})$]\label{thm:code_redux}
For an iid AGN channel, in terms of the geometric mean error on the data modes, the most general GKP-stabilizer code can be completely reduced to a direct product of TMS codes with a general (potentially multimode) ancillary lattice state $\mathcal{L}$, defined as $\mathsf{TMS}^\otimes(\mathcal{L})$.
\end{theorem}

\begin{proof}
Theorem~\ref{thm:code_redux} is a consequence of Lemma~\ref{lemma:lemma_redux}, together with the fact that (1) a local symplectic transformation $\bm\Lambda_a$ on the initial ancillary lattice $\mathcal{L}$ defines a new (symplectically integral) lattice $\mathcal{L}^\prime$ (see, e.g., Appendix~\ref{app:gkp_lattice}) and (2) coherent pre- and post-processing of the data modes do not increase code performance. To show (2), using the geometric mean $\bar\sigma_{ \rm GM}^2$ as a performance metric, first observe that we can push data noise processing via $\bm\Lambda_d^{-1}$ after the corrective displacement ${D}_{\bm f}$ by simply redefining our error estimator $\bm f^\prime\equiv\bm\Lambda_d\bm f$, i.e., ${D}_{\bm f}\circ \mathcal{U}_{\bm\Lambda_d^{-1}}= \mathcal{U}_{\bm\Lambda_d^{-1}}\circ {D}_{\bm\Lambda_d\bm f}$. Just after corrective displacements (but prior to data processing), the residual error covariance matrix on the data modes is $\bm V_{\rm out}$ and transforms to $\bm V_{\rm out}^\prime=\bm\Lambda_d^{-1}\bm V_{\rm out}\bm\Lambda_d^{-\top}$ after data processing. However, the geometric mean is invariant under symplectic transformations; in other words, $\det\bm V_{\rm out}^\prime=\det\bm V_{\rm out}$. Thus, the performance of the code (as quantified by the geometric mean error) only depends on the gain parameters of the TMS operations, the ancillary lattice state, and the error estimator for corrective displacements---not on $\bm\Lambda_d$.
\end{proof}

Theorem~\ref{thm:code_redux} reduces code design to choosing $N$ gain parameters for the TMS operations and finding a good ancillary lattice $\mathcal{L}$ for best code performance. [\QZ{We approach this problem numerically in later sections.}] \QZ{Therefore, the following corollary directly follows from Theorem~\ref{thm:code_redux} as an observation.} 
\begin{corollary}[Optimality of $\mathsf{TMS}^\otimes(\mathcal{L})$]
The optimal GKP-stabilizer code \QZ{with $N$ data modes} can be constructed from a GKP-TMS code $\mathsf{TMS}^\otimes(\mathcal{L})$ with an optimized (potentially multimode) GKP lattice $\mathcal{L}$ and optimized TMS gain parameters \QZ{$\{G_i\}_{i=1}^N$}.
\label{coro:optimality}
\end{corollary}

As a side note, in deriving Lemma~\ref{lemma:lemma_redux} and Theorem~\ref{thm:code_redux}, observe that we do not explicitly utilize properties of the ancilla state $\mathcal{L}$. This indicates that our results hold for any Gaussian encoding with general non-Gaussian ancilla, not only for the GKP-stabilizer codes. In particular, one can likewise consider a general $M$-mode non-Gaussian resource state $\mathcal{L}$ and a family of such states $\mathscr{F}(\mathcal{L})$ that are related to $\mathcal{L}$ via Gaussian transformations, i.e. $\mathscr{F}(\mathcal{L})\equiv\{\mathcal{L}^\prime=\mathcal{U}_{\bm\Lambda_a}(\mathcal{L})|\bm\Lambda_a\in{\rm Sp}(2M,\mathbb{R})\}$, where ${\rm Sp}(2M,\mathbb{R})$ is the set of $2M\times2M$ real symplectic matrices. Such a class of codes represent a fairly general encoding of oscillators-to-oscillators, especially considering that Gaussian operations supplied with non-Gaussian GKP ancillae are universal and sufficient for fault-tolerant quantum computation~\cite{baragiola2019all}.

\begin{figure}
    \centering
    \includegraphics[width=\linewidth]{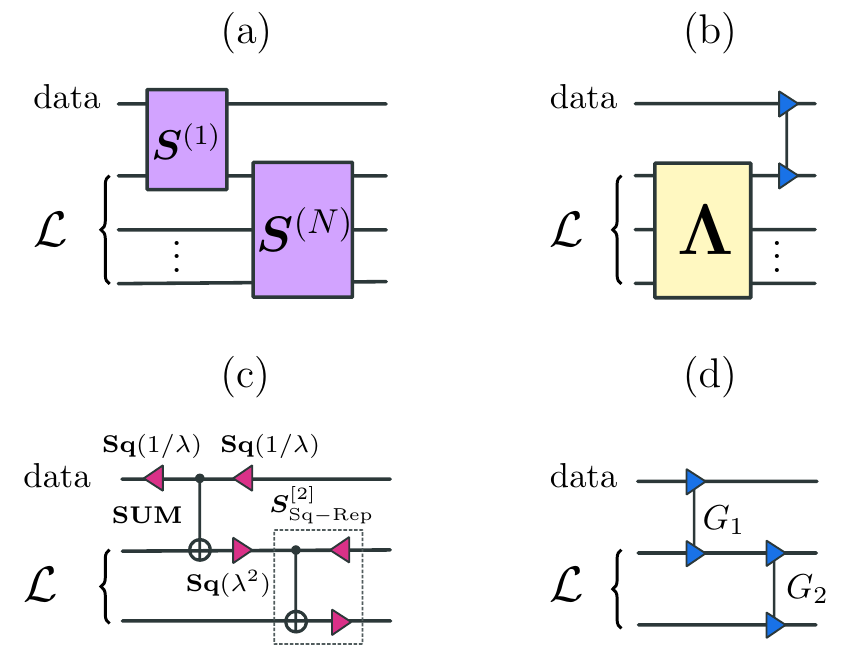}
    \caption{\QZ{Circuit diagrams for concatenated codes. (a) General definition of a concatenated code. Data couples to an ancilla via $\bm S^{(1)}$. Adding more layers leads to further error suppression. (b) Equivalent representation using TMS code reduction. Most operations in a concatenated code can be done offline on the ancillae and the ancillae can be coupled to the data at the end via TMS. (c) Example of a three mode concatenated squeezed repetition code (cf. Ref.~\cite{noh2019encoding}). (d) Example of a three mode concatenated TMS code.}}
    \label{fig:concatenated_codes}
\end{figure}

\subsection{\QZ{Concatenated codes}} 
\label{sec:concatenation_reduction}

\QZ{In general, a concatenated code applies another layer of encoding on the ancilla mode to suppress the noise level on the ancilla modes, prior to utilizing the ancilla modes to protect the data modes, as shown in Fig.~\ref{fig:concatenated_codes}(a). Examples of concatenated codes are the GKP-squeezed-repetition code detailed in Fig.~\ref{fig:concatenated_codes}(c), where repeated sum-gate and single-mode squeezing is applied in each layer of encoding, and the three-mode TMS code shown in Fig.~\ref{fig:concatenated_codes}(d). The purpose of concatenation is to suppress the logical output variance to higher powers of the input AGN $\sigma$, analogous to DV qubit codes. In particular, for $N$ concatenation layers, we expect that $\sigma^{\rm out}\sim\sigma^{N+1}$. Evidence of higher-order error suppression has been shown for squeezed repetition codes~\cite{noh2019encoding,zhuang2020distributed} and TMS codes~\cite{wu2021continuous,zhou2022enhancing}.}

\QZ{We now make an interesting observation on Theorem~\ref{thm:code_redux} as it pertains to concatenated codes}. To protect $N$ data modes with a concatenated code, Theorem~\ref{thm:code_redux} states that we can prepare a multimode ancilla ($M>N$) ``offline'' and then couple the $N$ data modes to the $M$ ancilla modes with only $N$ TMS gates. \QZ{This is shown schematically in Fig.~\ref{fig:concatenated_codes}(b) for $N=1$ data modes.} In other words, most of the operations are pushed to entangling ancilla modes (which generally requires $\sim M^2$ elementary gates). As an example, to protect a single data mode ($N=1$) by a concatenated code that leverages $M$ ancilla modes, the ancilla modes themselves need to interact by some $M$ mode Gaussian unitary, which can be done offline, but the data only needs to couple to \textit{one} of the ancilla modes by a single TMS operation \QZ{in the last step of encoding}; see Fig.~\ref{fig:concatenated_codes}(b) for an illustration. 

\subsection{Example of code reduction} 
\label{sec:example_reduction}
We give a few example code reductions (consequences of Lemma~\ref{lemma:lemma_redux}) by specifying the local symplectic matrices $\bm\Lambda_d$ and $\bm\Lambda_a$ discussed in Lemma~\ref{lemma:lemma_redux}.

In the two-mode case, we reduce the GKP-squeezed-repetition code, $\bm S^{[2]}_{\rm{Sq-Rep}}$ of Eq.~\eqref{eq:rep12}, to a GKP-TMS code. The local symplectic transform $\bm\Lambda_d$ is a single-mode squeezer $\bm\Lambda_d={\rm diag}(\sqrt[4]{2}\lambda, 1/\sqrt[4]{2}\lambda)$ and $\bm \Lambda_a=\bm \Lambda_d^{-1}$. Applying $\bm \Lambda_a$ and $\bm \Lambda_d$, we obtain TMS with gain ${G=(\sqrt{2}+1)/2}$. Hence, the two-mode squeezed repetition code is equivalent to a TMS code of fixed gain $G=(\sqrt{2}+1)/2$ and an ancilla rectangular GKP (with dimensions specified by the squeezing parameter $\lambda$).

Consider one data mode and two ancilla modes. We show that a downward staircase concatenation of the TMS code \QZ{(as sketched in Fig.~\ref{fig:concatenated_codes}(d))}~\cite{wu2021continuous} $\bm S_{G_1,G_2}^{\rm dwn}= (\bI_2 \oplus \bm S_{G_2})(\bm S_{G_1} \oplus \bI_2)$ is equivalent to an upward staircase concatenation $\bm S_{G_a,G_{12}}^{\rm up}=(\bm S_{G_{12}} \oplus \bI_2 ) (\bI_2 \oplus \bm S_{G_a})$ with to be determined gains $G_{a}$ and $G_{12}$~\footnote{The phrase downward staircase refers to the fact that we couple the first and second modes, then the second and third modes etc., whereas an upward staircase starts from the bottom mode and goes to the top.}. The data mode is already in normal form, hence $\bm \Lambda_d=\bm I_2$. The local symplectic transform on the ancilla modes is given by an inverse TMS transformation, $\bm \Lambda_a= \bm S_{G_a}^{-1}$, with gain ${G_a={G_2 G_1}/({1+G_2(G_1-1)})}$. Applying $\bm \Lambda_d$ and $\bm \Lambda_a$ then reduces the staircase encoding to TMS between the first and second modes with gain $G_{12}=1+G_2(G_1-1)$.

As a final example, we reduce the three-mode squeezed-repetition code introduced in Ref.~\cite{noh2019encoding} \QZ{(see sketch in Fig.~\ref{fig:concatenated_codes}(c))}, which leads to cubic noise suppression for the output standard deviation. The squeezed-repetition code has encoding matrix
\begin{equation}
   \bm S^{[3]}_{\rm{Sq-Rep}}=
   \begin{pmatrix}
    1/\lambda^2 & 0 & 0 & 0 & 0 & 0 \\
    0 & \lambda^2 & 0 & -\lambda & 0 & 0 \\
    1 & 0 & \lambda & 0 & 0 & 0 \\
    0 & 0 & 0 & 1/\lambda & 0 & -\lambda \\
    \lambda^2 & 0 & \lambda^3 & 0 & \lambda & 0 \\
    0 & 0 & 0 & 0 & 0 & 1/\lambda
    \end{pmatrix};
\end{equation}
see Fig.~\ref{fig:concatenated_codes}(c) for a circuit diagram. By squeezing the data mode via $\bm \Lambda_d = \textbf{Sq}(\lambda (\lambda^2+\lambda^4)^{\frac{1}{4}})$ and applying a two-mode ancilla transformation $\bm \Lambda_a$, which is lengthy to show, we end up with TMS between the first and the second mode with gain $G_{12}=(\sqrt{1/\lambda^2+1}+1)/2$.

\QZ{In our analyses of this section, we have considered geometric mean as a performance metric. While such a choice is motivated by information-theoretical roots of the metric, some generalization of the results is possible. First, while the sufficiency and optimality of $\mathsf{TMS}^\otimes(\mathcal{L})$ in Theorem~\ref{thm:code_redux} and Lemma~\ref{lemma:lemma_redux} are derived with geometric mean of error as the metric, they hold for any metric that is invariant under local symplectic transformations. On the other hand, it is generally an open question to what extend some of the results here hold or approximately hold for non-Gaussian additive noise channels. 
}

\QZ{
\section{Code reduction implies no threshold for finite squeezing}
\label{sec:code_reduction_no_threshold}
We prove a no threshold theorem for finite squeezing that applies to general oscillator-to-oscillator codes based on Gaussian encoding. Interestingly, the no threshold theorem is a direct consequence of our general code reduction to TMS codes.
}

\QZ{
Consider iid AGN noise $\bm\xi=(\bm\xi_d,\bm\xi_a)$ for the displacement errors on $N$ data modes and $M$ ancilla modes. After the encoding (and decoding) transformations $\bm S_{\rm enc}$ ($\bm S_{\rm enc}^{-1}$), the correlated noises are $(\bm x_d, \bm x_a)= {\bm S}_{\rm enc}^{-1} \bm \xi\sim\mathcal{N}(0,\bm V_{\bx})$, where $\bm V_{\bx}=\sigma^2\bm S_{\rm enc}^{-1}\bm S_{\rm enc}^{-\top}$. Since all encoding can be reduced to TMS up to local transformations via Lemma~\ref{lemma:lemma_redux}, we need only consider $\bm V_{\bx}$ as a direct sum of $N$ correlated two-mode blocks (plus an identity block on the remaining $M-N$ ancilla modes), which has $qq$ and $pp$ correlations between data and ancilla arising from two-mode squeezing (but zero $qp$ correlations). Therefore, to analyze general properties of the code, we can focus our attention on one data-ancilla mode pair (say, the $i$th mode pair) and one quadrature (say, the $q$ quadrature) at a time; see Appendix~\ref{app:classical_it} for further details. 
}

\QZ{
Let $q_{a_i}$ be the $(2i-1)$th element of $\bm x_a$ (where $i=1,2,\dots,N$) that is correlated with $q_{d_i}$ of $\bm x_d$ via TMS with gain $G_i$, and let $\Tilde{q}_{d_i}\equiv\Tilde{q}_{d_i}(q_{a_i})$ be the estimation of the data noise given knowledge of the ancilla noise, which we can extract from, e.g., syndrome measurements. Although perfect knowledge of the ancilla noise is not generally available, we assume that it is in order to place an ultimate lower bound. Now from the corollary of Theorem 8.6.6 in Ref.~\cite{cover2006elements}, the estimation variance of a generic random variable $X$, given side information $Y$, is lower bounded by a function of the conditional differential entropy $S(X|Y)$ via $\mathbb{E}[(X-\Tilde{X}(Y))^2]\geq\exp\left[2 S(X|Y)\right]/2\pi e$. In our current setting, $S(q_{d_i}|q_{a_i}) =\ln(\frac{2\pi e\sigma^2}{2G_i-1})/2$, which is limited by the finite squeezing to correlate the noises (see Appendix~\ref{app:classical_it} for a derivation). We point out that an equivalent relation holds for the momenta, $p_{a_i}$ and $p_{d_i}$. Therefore, 
\begin{equation}\label{eq:qestimation_variance}
    \mathbb{E}\left[\big(q_{d_i}-\tilde{q}_{d_i}\big)^2\right]\geq\frac{\sigma^2}{2G_i-1}.
\end{equation}
If the TMS has a finite squeezing level $G_i$, then having a larger number of ancilla modes (or an arbitrary ancilla state) will not further help error correction. This implies a universal no-threshold theorem for a wide variety of codes based on Gaussian encoding---including but not limited to GKP-stabilizer codes.
}

\QZ{
\begin{theorem}[No threshold for finite squeezing]\label{thm:noThreshold} For $N$ data modes and arbitrary number of ancilla modes, the residual error for any oscillators-to-oscillators code using Gaussian encoding is lower bounded by
\begin{align}
\Tr(\bm V_{\rm out})\geq\sum_{i=1}^{N}\frac{ 2\sigma^2}{2G_i-1}
\end{align}
where $\bm V_{\rm out}$ is the covariance matrix for the residual output error of the code, $G_i$ are the two-mode squeezing gains of the encoding after code reduction [see Lemma~\ref{lemma:lemma_redux} and Theorem~\ref{thm:code_redux}], and $\sigma^2$ is the variance of the iid AGN channels.
\end{theorem}
\begin{proof}
    The proof follows by summing over the individual variances (left hand side) of Eq.~\eqref{eq:qestimation_variance}, which is less than or equal to the trace of the residual output covariance matrix, $\bm V_{\rm out}$. The factor of two is due to the fact that the $q$ and $p$ quadratures of the $i$th mode contribute equally to the sum due to the structure of two-mode squeezing.
\end{proof}
}

\QZ{
If we place a tolerance on the output error $\varepsilon\geq\Tr(\bm V_{\rm out})$, Theorem~\ref{thm:noThreshold} implies that the (average) gain $G$ must scale as $G\sim N\sigma^2/\varepsilon$ for the error to stay below tolerance, which is independent of the number of ancilla modes $M$ used in the code. Thus, we cannot make $\varepsilon$ arbitrarily small with a finite amount of squeezing even if we increase the number of ancilla modes; this is the essence of the no-threshold theorem for oscillator-to-oscillator codes. 
}

\QZ{
Our proof follows from the general code reduction Theorem~\ref{thm:code_redux} (see also Lemma~\ref{lemma:lemma_redux}) and a simple, classical data-processing argument. Furthermore, our result has a broader scope than a similar no-threshold result of Ref.~\cite{hanggli2021oscillator} based on GKP-stabilizer codes and maximum-likelihood decoding, as we do not require the ancilla modes to be prepared in GKP states and the decoding strategy does not enter into our proof. The only caveat here is the assumption of iid noise across all modes (likewise in~\cite{hanggli2021oscillator}).
}

\section{General GKP lattices and minimum mean square estimation}\label{sec:gkp_lattices_mmse}

Before exploring code designs, we review general GKP lattices and derive the minimum mean square estimation (MMSE) that minimizes $\bar\sigma_{ \rm RMS}^2$ in Eq.~\eqref{eq:sigma_RMS}. 

\subsection{GKP lattices} 
Consider a set of vectors $\{\bm\lambda_K^{\mathcal{L}}\}_{k=1}^{2N}$, where $\bm\lambda_K^{\mathcal{L}}\in\mathbb{R}^{2N}$, that generate a rigid phase-space lattice $\mathcal{L}\subset\mathbb{R}^{2N}$. We define the `stabilizers' (formed by displacement operators) of the lattice as ${S}_K^{\mathcal{L}}\equiv{D}_{\bm\lambda_K^{\mathcal{L}}}$, such that
\begin{equation}\label{eq:stabilizer_comm_general}
    [{S}_K^{\mathcal{L}}, {S}_J^{\mathcal{L}}]=0,\quad \forall~ K,J.
\end{equation}
The commutator relation is equivalent to the condition $\bm\lambda_K^{\mathcal{L}\,\top}\bm\Omega\bm\lambda_J^{\mathcal{L}}=2\pi n_{KJ}$, where $\bm\Omega=
\bm{I}_N\otimes\big(\begin{smallmatrix}
0 & 1\\
-1 & 0
\end{smallmatrix}\big)$ is the $N$-mode symplectic form and $n_{KJ}\in\mathbb{Z}$ is an integer (see Appendix~\ref{app:gauss_evol} and Appendix~\ref{app:gkp_lattice}). A lattice $\mathcal{L}$ with basis vectors $\{\bm\lambda_K^{\mathcal{L}}\}$ which satisfy this condition is known as a \textit{symplectically integral} lattice~\cite{baptiste2022multiGKP}. By virture of the commutator~\eqref{eq:stabilizer_comm_general}, the stabilizers form a group $\mathcal{S}_\mathcal{L}\equiv\langle{S}_1^{\mathcal{L}},\dots,{S}_{2N}^{\mathcal{L}}\rangle$. We define the lattice state $\ket{\mathcal{L}}$ as the $+1$ eigenstate of all elements in $\mathcal{S}_\mathcal{L}$, i.e. `$\mathcal{S}_\mathcal{L}\ket{\mathcal{L}}=\ket{\mathcal{L}}$'. This state is periodic in the $2N$ dimensional phase space of the modes. See Fig.~\ref{fig:lattice_examples} for an illustration of two-dimensional (symplectically integral) lattices.

From the lattice (column) vectors, one can construct a $2N\times2N$ generator matrix,
\begin{equation}
    \bm M=
    \begin{pmatrix}
        \bm\lambda_1^{\mathcal{L}} & \bm\lambda_2^{\mathcal{L}} & \dots & \bm\lambda_{2N}^{\mathcal{L}}
    \end{pmatrix}.
    \label{eq:lattice-generator-matrix}
\end{equation}
Then, the set of conditions $\bm\lambda_K^{\mathcal{L}\,\top}\bm\Omega\bm\lambda_J^{\mathcal{L}}=2\pi n_{KJ}$ can be compactly written as
\begin{equation}\label{eq:integral_A}
    \bm M^\top\cdot\bm\Omega\cdot\bm M=2\pi\bm A,
\end{equation}
where $\bm A$ is an anti-symmetric matrix with only integer elements.
We can generate the same lattice with different choices of basis vectors. For instance, given a generator matrix $\bm M$ that generates a lattice $\mathcal{L}$, one can choose a unimodular matrix $\bm N$ (i.e., integer matrix with $\det\bm N=1$) such that $\bm M^\prime = \bm M\bm N$ also generates $\mathcal{L}$.

\QZ{We now make the distinction between computational GKP states and canonical GKP states precise. In general, if we want to encode a qudit with $d$-levels into a system of $M$ oscillators, the encoded Hilbert space stabilized by $\mathcal{S}_{\mathcal{L}}$ will be related to the generator matrix $\bm M$ via $d=\sqrt{\det \bm M/(2\pi)^{M}}$~\cite{gkp2000,baptiste2022multiGKP,eisert2022lattice}.\footnote{The factor $1/(2\pi)^M$ is due to our definition of the generator matrix $\bm M$, which differs from convention in Refs.~\cite{gkp2000,baptiste2022multiGKP,eisert2022lattice} by a constant factor $\sqrt{2\pi}$.} In this paper, we are focusing on the $d=1$ case which only supports a single code state. This is commonly referred to as canonical GKP state~\cite{noh2020o2o} (hence the moniker ``canonical GKP lattice state'' for generic lattices) or the sensor state in the square lattice case~\cite{duivenvoorden2017}.}

\begin{figure}
    \centering
    \includegraphics[width=\linewidth]{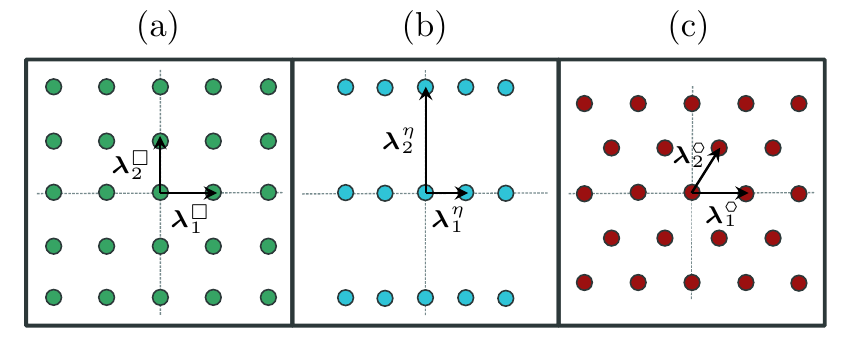}
    \caption{Illustration of two-dimensional lattices with basis vectors denoted: (a) Square, (b) Rectangular, (c) Hexagonal. Any two-dimensional symplectically integral lattice can be generated from the square lattice by rotation and squeezing.}
    \label{fig:lattice_examples}
\end{figure}

\subsection{Minimum mean square error (MMSE) estimation} 
Here we consider MMSE for corrective displacements, which is constructed to minimize the RMS error of Eq.~\eqref{eq:sigma_RMS} (as RMS error is the square root of mean square error). The encoding symplectic transform $\bm{S}_{\rm{enc}}$ correlates the additive noise between the data and the ancilla. Let the covariance matrix of the AGN be $\bm{V}_\xi$. The error correlations are described by
\begin{align}
    \bm V_{\bx} = \bm{S}_{\rm{enc}}^{-1}\bm{V}_\xi\bm{S}_{\rm{enc}}^{-\top}.
    \label{eq:Noise-CM-transform}
\end{align}
The additive noise on the ancilla $\bx_a$ can be extracted by measuring the stabilizers $S_J^{\mathcal{L}}$. This leads to an error syndrome $\bm s \in [-\sqrt{\pi/2},\sqrt{\pi/2}]^{2M}$, from which we can estimate the additive noise on the data $\bx_d$. For a general lattice with generator matrix $\bm M$, we have $\bm s = R_{\sqrt{2\pi}}(\bm M^\top \bm\Omega\bx_a)$, where $R_{\sqrt{2\pi}}$ is the element-wise modulus of $\sqrt{2\pi}$ that associates the vector $\bm M^\top \bm\Omega\bx_a$ to the nearest lattice point within a region $[-\sqrt{\pi/2},\sqrt{\pi/2}]^{2M}$ of that point~\cite{noh2020o2o}; see also Appendix~\ref{app:measurements}.

The joint probability density distribution (PDF) of the data and the error syndrome, $P(\bx_d,\bm s)$, is not a Gaussian distribution but a sum of Gaussian distributions. The MMSE minimizes
$\bar\sigma_{ \rm RMS}^2$ in Eq.~\eqref{eq:sigma_RMS}, and the estimator can be derived from the conditional distribution $P(\bx_d|\bm s)$ via
$
    \bm f_{\rm MMSE}(\bm s)= \int_{\mathbb{R}^{2N}}\diff{\bx_d}\;\;  \bx_d P(\bx_d|\bm s)
$
leading to the following theorem (see Appendix~\ref{app:error_estimation} for a derivation). 
\begin{theorem}\label{thm:fMMSE}
For a GKP-stabilizer code with GKP lattice state $\mathcal{L}$ described by generator matrix $\bm M$, the minimum mean square estimation (MMSE) for an error syndrome $\bm s$ is given by
\begin{align}
    &\bm f_{\rm MMSE}(\bm s) \nonumber
    \\
   & = -\frac{\sum_{\bm n} \bm V_d^{-1}\bm V_{da}(\bm s-\bm n\sqrt{2\pi})e^{(\sqrt{2\pi} \bm n^{\top} \bm V_{d|a}\bm s-\pi \bm n^{\top} \bm  V_{d|a}\bm n)} }
    {\sum_{\bm m} e^{(\sqrt{2\pi} {\bm m}^{\top} \bm V_{d|a}\bm s-\pi {\bm m}^{\top} \bm V_{d|a}\bm m)}},
    \label{eq:fMMSE_main}
\end{align}
where $\bm m, \bm n\in\mathbb{Z}^{2M}$ sum over all integer vectors, and the matrices $\bm V_{da}$, $\bm V_d$ and $\bm V_a$ are defined through the equation
\begin{align}
    \begin{pmatrix}
    \bm V_d & \bm V_{da}\\
    \bm V_{da}^T & \bm V_a
    \end{pmatrix}^{-1}\equiv 
    (\bI_{2N} \oplus {\bm M^\top \bm \Omega}) \bm V_{\bx} (\bI_{2N} \oplus {(\bm M^\top \bm \Omega)} ^{\top})
    .
    \label{eq:MMSE_submatrices}
\end{align}
and $\bm V_{d|a}=\bm V_a-\bm V_{da}^T \bm V_d^{-1}\bm V_{da}$.
\end{theorem}

\QZ{As an exmaple, we provide the explicit application of Theorem~\ref{thm:fMMSE} on GKP-TMS code in Appendix~\ref{app:explicit_applying_theorems_estimator}.}

Given the estimator $f_{\rm MMSE}(\bm s)$ above, the residual noise covariance matrix for the data, $\bm V_{\rm out}$, can be evaluated. 
In Appendix~\ref{app:lattice_transform}, we show that the PDF of joint distribution $P(\bx_d,\bm s)$ is invariant under a change of lattice basis ${\bm M\rightarrow\bm M\bm N}$ where $\bm N$ is a unimodular matrix. It follows that the error correction performance with MMSE is invariant under a change of lattice basis as well.

If the noise level is much smaller than $\sqrt{2\pi}$, then $\bm s = R_{\sqrt{2\pi}}(\bm M^\top \bm\Omega \bx_a) \approx \bm M^\top \bm\Omega\bx_a$. In this case, the PDF is close to the Gaussian distribution. \QZ{From Eq.~\eqref{eq:fMMSE_main}, keeping the leading terms $\bm n=0$ and $\bm m=0$ for numerator and denominator, the estimator is close to linear estimation, $\bm f_{\rm{Linear}}(\bs) = -\bm V_d^{-1} \bm V_{da} \bs$.} However, unlike MMSE, linear estimation is not invariant under lattice basis transform. See Appendix~\ref{app:error_estimation} for technical details.

\QZ{Both the linear and MMSE estimators assume knowledge of the covariance matrix of the Gaussian noise. While matrix multiplication and inversion are both required (only once) to derive the estimators, linear estimation involves no summations. On the contrary, MMSE estimation requires summation over two integer vectors of length $2M$. Fortunately, the convergence of the summation is exponential, which makes the evaluation possible, although the cost grows with the number of modes.}

\begin{table*}
    \centering
    \renewcommand{\arraystretch}{2}
    \begin{tabular}{c|c c c c}
    \hline\hline
        $\left(\theta,r\right)$ & $\big(\frac{\pi}{4},\sqrt[4]{3}\big)$ & $\left(0.16\pi,2.095\right)$ &$\left(0.11\pi,3.021\right)$
        &$\left(0.18\pi,3.385\right)$\\\hline
        $\bm N$ & $\normalsize{\big(\begin{smallmatrix}
        1 & 0\\
        0 & 1
        \end{smallmatrix}\big)}$  & $\normalsize{\big(\begin{smallmatrix}
        -2 & -1\\
        1 & 0
        \end{smallmatrix}\big)}$ &
        $\normalsize{\big(\begin{smallmatrix}
        2 & -1\\
        1 & 0
        \end{smallmatrix}\big)}$ 
        &${\normalsize\big(\begin{smallmatrix}
        1 & -2\\
        1 & -1
        \end{smallmatrix}\big)}$\\
        \hline\hline
    \end{tabular}
    \caption{Equivalent representations of a hexagonal lattice. $\bm N$ is a unimodular matrix that relates the lattice basis vectors.}
    \label{tab:hexagonal}
\end{table*}

\section{Optimal single-mode code}
\label{sec:single_mode}

\begin{figure}
    \centering
    \includegraphics[width=\linewidth]{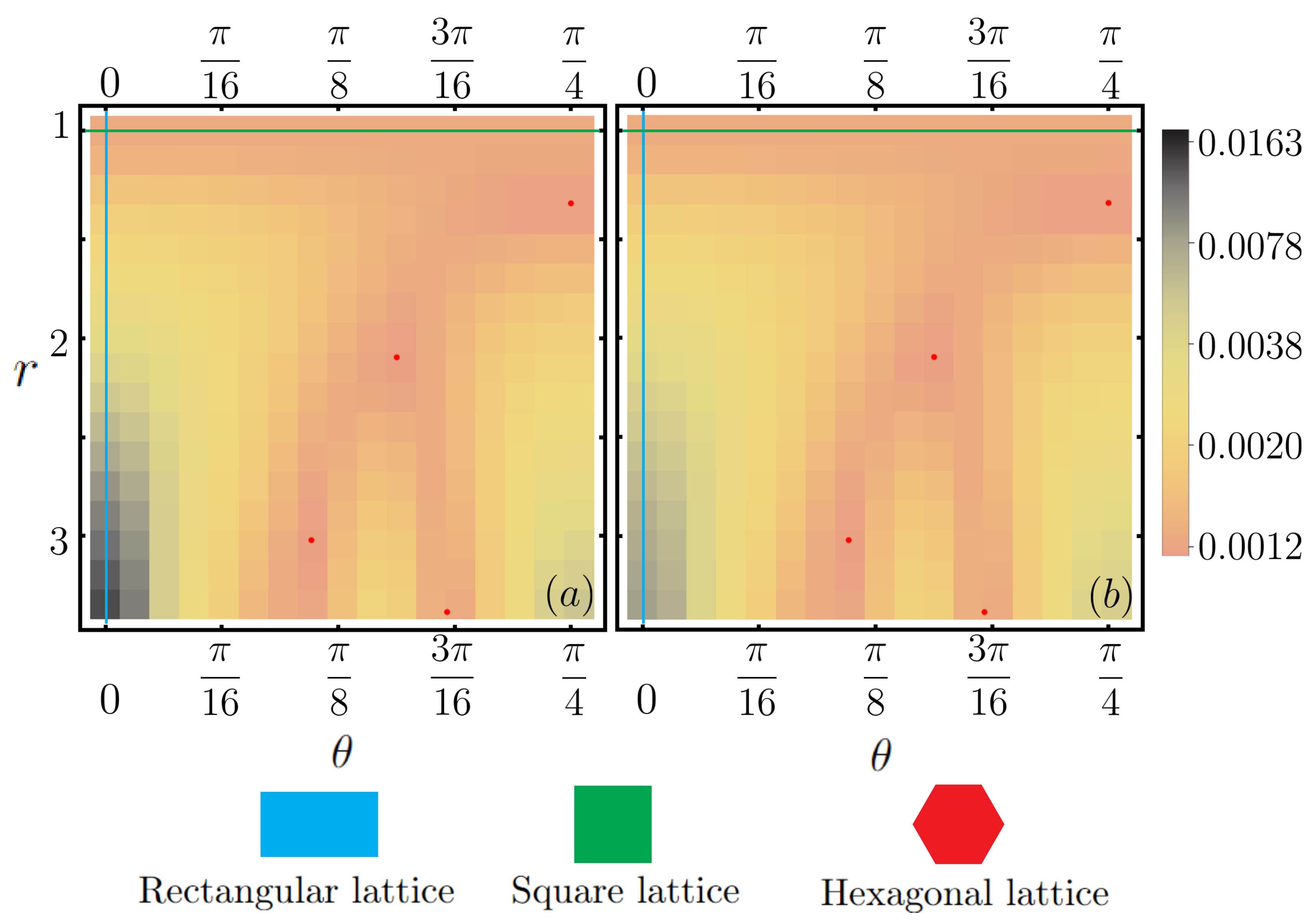}
    \caption{Output noise for a single-mode ($N=M=1$) GKP stabilizer code. Input noise variance is \QZ{$\sigma^2=10^{-2}$}. We optimize the TMS gain $G$ for each point $(r,\theta)$. (a) RMS error $\bar\sigma_{ \rm RMS}^2$, (b) GM error $\bar\sigma_{ \rm GM}^2$. For the square lattice (green line), $\bar\sigma_{ \rm RMS}^2=1.25129(5)\times10^{-3}$and $\bar\sigma_{ \rm GM}^2=1.25129(5)\times10^{-3}$. The four hexagonal lattice points (red dots) have the same output noises of $\bar\sigma_{ \rm RMS}^2=1.15575(5)\times10^{-3}$ for RMS error and $\bar\sigma_{ \rm GM}^2 =1.15575(5)\times10^{-3}$ for GM error. Only the range $\theta\in[0,\pi/4]$ is considered due to symmetry; see Appendix~\ref{app:Symmetry_single_mode}.}
    \label{fig:single_mode_lattice_contour}
\end{figure}

Since all single-mode canonical lattice states can be generated by local symplectic transformations on the square grid (see Appendix~\ref{app:gkp_lattice}), Theorem~\ref{thm:code_redux} and Corollary~\ref{coro:optimality} immediately imply the following.

\begin{theorem}\label{thm:tms_1mode}
For single-mode data and ancilla undergoing iid AGN, the TMS code $\mathsf{TMS}(\mathcal{L})$, with gain $G$ and an ancillary lattice $\mathcal{L}\subset\mathbb{R}^2$ generated from the square grid by a local symplectic transformation, is the optimal GKP-stabilizer code in terms of geometric mean error.
\end{theorem}

Therefore, to obtain the best two-mode GKP-stablizer code, one simply needs to optimize the local Gaussian unitary and TMS gain, alongside choosing an optimal estimator $\bm f$ for the classical decoding strategy. Deriving the best estimator $\bm f$ to minimize the GM error is nontrivial. On the other hand, we can obtain an upper bound on the GM error from the RMS error since $\bar\sigma_{ \rm RMS}^2\ge  \bar\sigma_{ \rm GM}^2$. Suppose that one can optimize the RMS error (which is possible with MMSE of Theorem~\ref{thm:fMMSE}) while showing that the GM error is very close to $\bar\sigma_{ \rm RMS}^2$; then, this is an indication that the code is optimal for GM error as well.

For the single-mode case, any ancillary lattice $\mathcal{L}$ can be generated by applying a single-mode Gaussian unitary $U_{\bm \Lambda_a}$ on the square GKP lattice (see Fig.~\ref{fig:lattice_examples} for an illustration). The canonical square GKP lattice can be written down in the $q$ or $p$ quadrature bases as
\begin{equation}
\ket{\square}\propto\sum_{n\in\mathbb{Z}}\ket{q=n\ell}=\sum_{n\in\mathbb{Z}}\ket{p=n\ell},
\end{equation}
which is a translation invariant square lattice in the single-mode phase space $\mathbb{R}^2$ with period $\ell=\sqrt{2\pi}$. Moreover, any single-mode Gaussian unitary can be described by a symplectic transform, with the decomposition $\bm R(\phi)\text{\bf Sq}(r)\bm R(\theta)$, where $\bm R(\varphi)$ is a $2\times2$ rotation matrix and $\text{\bf Sq}(r)\equiv {\rm diag}[r,1/r]$ is single-mode squeezing. Due to the symmetry of the AGN, the effect of the last phase rotation $\bm R(\phi)$ applied on the lattice will not change the performance, therefore we parameterize the transform as $\bm \Lambda_a =\text{\bf Sq}(r)\bm R(\theta)$ such that $\ket{\mathcal{L}}=U_{\bm \Lambda_a}\ket{\square}$. As examples, a rectangular GKP state (Fig.~\ref{fig:lattice_examples}b) is given by $\theta=0$ and $r>0$, and a hexagonal GKP state (Fig.~\ref{fig:lattice_examples}c) is given by $\theta=\pi/4$ and $r_{{\tiny\hexagon}}=\sqrt[4]{3}$.

In Fig.~\ref{fig:single_mode_lattice_contour}(a), we plot the contour of the RMS error $\bar\sigma_{ \rm RMS}^2$ for an MMSE decoder optimized over the TMS gain $G$ for each point $(r,\theta)$. We find four equal minimum points for $\bar\sigma_{ \rm RMS}^2$, which turn out to be equivalent lattice representations of the hexagonal lattice as listed in Table~\ref{tab:hexagonal}. Meanwhile, the square lattice has $r=1$ with $\theta$ arbitrary (represented by the green line); the rectangular lattice has $\theta=0$ and $r$ changes the shape of the rectangle (represented by the blue line). The square and rectangular lattices are strictly sub-optimal.

In Fig.~\ref{fig:single_mode_lattice_contour}(b), we plot the GM error $\bar\sigma_{ \rm GM}^2$ in $(r,\theta)$ parameter space for the same optimized gain values of Fig.~\ref{fig:single_mode_lattice_contour}(a). The two subplots are very similar, with some deviations at the left-bottom corner due to the large squeezing of a rectangular lattice which induces asymmetry between $q$ and $p$ quadratures. The hexagonal lattices again minimize the output noise. Moreover, for the hexagonal lattices, $\bar\sigma_{ \rm GM}^2\approx \bar\sigma_{ \rm RMS}^2$ up to our numerical precision, which is a strong indicator that---even if we minimize the GM error instead---the hexagonal lattice is still optimal.

\section{Multimode codes}
\label{sec:multi_mode}

It turns out that all canonical (or ``self-dual'') lattices can be generated from canonical square GKP (consequence of Corollary 1 in Ref.~\cite{eisert2022lattice}). Therefore, Theorem~\ref{thm:tms_1mode} can generalize to the multimode case. However, optimization is still challenging since $|{\rm Sp}(2M,\mathbb{R})|=2M^2 +M$ parameters need to be optimized in general. The search for an optimal GKP lattice suggested by Theorem~\ref{thm:code_redux} and Corollary~\ref{coro:optimality} is therefore difficult. Nevertheless, as we will show in this section with a few examples, going to higher-dimensional lattices may indeed improve the performance of oscillators-to-oscillators codes.

Below, we present our results on multi-mode codes.
We first give a lower bound on the output noise for a general multimode GKP code, then discuss \QZ{break-even points}. Finally, evaluate the performance of \QZ{$N=M=2$ and $N=1,M=2$} multimode GKP stabilizer codes for various lattice configurations (e.g., square, hexagonal, and D4) and estimation strategies (e.g., linear estimation versus MMSE). 

\subsection{Lower bound and AGN \QZ{break-even point}}

By information theoretic arguments (see Appendix~\ref{app:bounds}), we are able to find lower bounds for the RMS and GM errors, $\Bar{\sigma}_{\rm RMS}$ and $\Bar{\sigma}_{\rm GM}$, for a general multimode GKP code, with $M$ ancilla modes and $N$ data modes, in terms of the variances $\sigma_i^2$ of the AGN channels $\bigotimes_i\mathcal{N}_{\sigma_i}$ (Theorem~\ref{thm:lb_logicalnoise}). In particular, for iid AGN, we show that (Corollary~\ref{cor:iid_bound}) $\bar{\sigma}_{\rm RMS}\geq\bar{\sigma}_{\rm GM}\geq\sigma_{\rm LB}$, where
\begin{equation}\label{eq:lower_bound}
    \sigma_{\rm LB}\equiv\frac{1}{\sqrt{e}}\left(\frac{\sigma^{2}}{1-\sigma^{2}}\right)^{\frac{N+M}{2N}}.
\end{equation}
For single-layer codes ($N=M$), there is at best quadratic error suppression, exactly similar to the $N=M=1$ GKP codes discussed in Ref.~\cite{noh2020o2o}. Higher order error suppression can be obtained for codes with $M>N$---with the output standard deviation scaling as $\sim\sigma^{1+\frac{M}{N}}$ per Eq.~\eqref{eq:lower_bound}---in agreement with the results on concatenated codes (for $N=1$ and $M>1$) found in Refs.~\cite{noh2019encoding,noh2020o2o,wu2021continuous,zhou2022enhancing}.

\begin{figure}
    \centering
    \includegraphics[width=.9\linewidth]{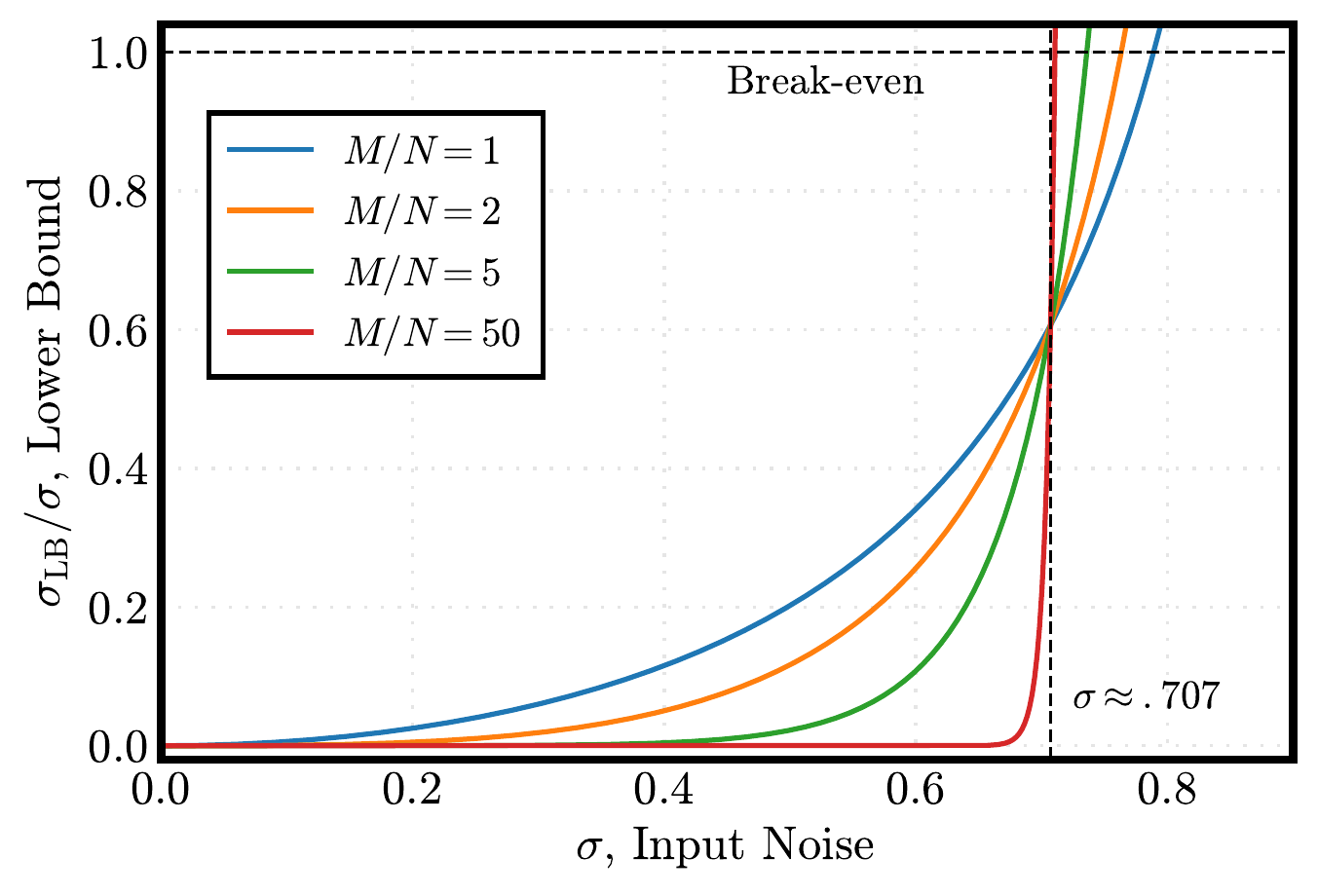}
    \caption{Lower bound $\sigma_{\rm LB}$ [Eq.~\eqref{eq:lower_bound}] on output noise versus initial noise $\sigma$ for multimode GKP codes with increasing $M/N$. A transition occurs at $\sigma=1/\sqrt{2}\approx.707$, which provides an upper bound for the maximal, correctable amount of AGN---i.e., $\sigma^\star\leq1/\sqrt{2}$, where $\sigma^\star$ is the AGN error \QZ{break-even point} for multimode GKP codes.}
    \label{fig:lowerbound}
\end{figure}

In Fig.~\ref{fig:lowerbound}, we plot the ratio $\sigma_{\rm LB}/\sigma$ versus the initial AGN noise $\sigma$, for increasing $M/N$; $\sigma_{\rm LB}/\sigma=1$ corresponds to the break-even point. We observe a sharp transition occurring near $\sigma=1/\sqrt{2}$ which defines an upper bound on the AGN error \QZ{break-even point} $\sigma^\star$ for general multimode GKP codes---i.e., $\sigma^\star\leq1/\sqrt{2}$. Thus, for $\sigma\gtrsim1/\sqrt{2}$, we expect no gain to be had from such codes. This is consistent with the upper bound on the energy-unconstrained quantum capacity of an AGN channel~\cite{albert2018GKPcapacity} (see also Ref.~\cite{wu2021continuous} and Lemma~\ref{lemma:capacity_ub} of Appendix~\ref{app:bounds}) $C_{\mathcal{Q}}(\mathcal{N}_\sigma)\leq\max[0,\log_2(\frac{1-\sigma^2}{\sigma^2})]$, which vanishes as $\sigma\rightarrow1/\sqrt{2}$. Furthermore, since a bosonic pure-loss channel with transmittance $\eta\in[0,1]$ can be converted via pre-amplification to an AGN channel with variance $\sigma^2=1-\eta$~\cite{albert2018GKPcapacity}, the AGN \QZ{break-even point} $\sigma^\star$ then corresponds to a pure-loss \QZ{transmissivity} $\eta^\star=1-\sigma^{\star\,2}\geq 1/2$.

The multimode GKP circuit of Fig.~\ref{fig:code_redux} corresponds to an ($N$ mode) additive non-Gaussian noise channel $\widetilde{\mathcal{N}}$ for the data, with output GM error $\bar{\sigma}_{\rm GM}$. We find a lower bound for the quantum capacity of the channel $\widetilde{\mathcal{N}}$ (see Lemma~\ref{lemma:multimode_lb} in Appendix~\ref{app:bounds}), ${C_{\mathcal{Q}}(\widetilde{\mathcal{N}})\geq\max[0,N\log_2(\frac{1}{e\bar{\sigma}_{\rm GM}^2})]}$. Assuming break-even $\bar{\sigma}_{\rm GM}=\sigma$, this in turn implies a lower bound for the AGN error \QZ{break-even point}, $\sigma^\star\geq1/\sqrt{e}\approx.607$ (thus, $\eta^\star\leq1-1/e\approx.632$). Hence, the \QZ{break-even point} $\sigma^\star$ ($\eta^\star$) for multimode GKP codes lies within $.607\leq\sigma^\star\leq.707$ ($.5\leq\eta^\star\leq.632$). In the next section, we numerically find break-even points for multimode ($N=M=2$) GKP stabilizer codes and MMSE estimation with AGN error \QZ{break-even points} near $1/\sqrt{e}\approx .607$. We remark that linear estimation strategies have a lower \QZ{break-even point} of $.558$~\cite{noh2020o2o}. See Fig.~\ref{fig:output_noise} and below.

\begin{figure*}
    \centering
    \includegraphics[width=.9\linewidth]{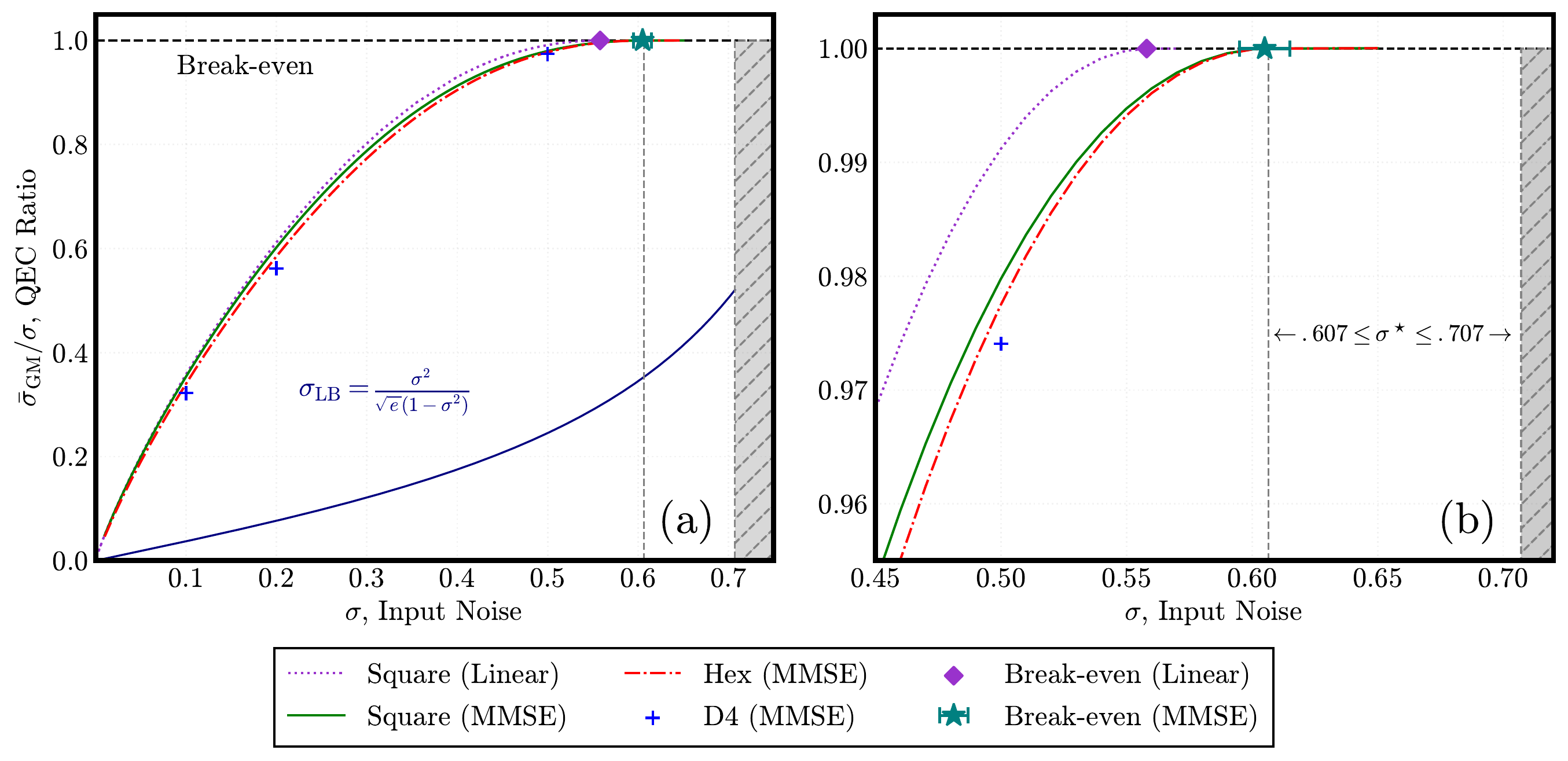}
    \caption{Quantum error correction (QEC) ratio between output noise and input noise of a single-layer, multimode ($N=M=2$) GKP TMS code $\mathsf{TMS}^{\otimes2}(\mathcal{L})$ for different lattices $\mathcal{L}$ (Square, Hexagonal, D4) and estimation strategies (linear, MMSE). The square code with linear estimation (dotted purple) agrees with the original GKP TMS code presented in Ref.~\cite{noh2020o2o}. Grey hatched region is forbidden by information theoretic arguments. \QZ{While the break-even point of the linear estimation with square lattice $\sigma\simeq 0.558$ (purple diamond), the break-even point of the MMSE estimation $\sigma=0.605(5)$ for both square and hexagonal lattices (green star).} \QZ{For $D4$ lattice, we narrowed the break-even point to $0.60-0.61$, which is consistent with square and hexagonal lattices.}}
    \label{fig:output_noise}
\end{figure*}

\subsection{\QZ{Two data modes and two ancilla modes ($N=M=2$)}}
\label{sec:N2M2}
We compare the performance of different initial lattices $\mathcal{L}$ for a single-layer, multimode ($N=M=2$) GKP-TMS code $\mathsf{TMS}^{\otimes2}(\mathcal{L})$. The encoding (decoding) is given by two TMS operations, with each TMS operation coupling one data mode to one ancilla mode; see Fig.~\ref{fig:code_redux}. The gain values of the TMS operations are (numerically) chosen to minimize the RMS error $\Bar{\sigma}_{\rm RMS}^2=\Tr{\bm V_{\rm out}}/4$. We consider MMSE estimation, which optimizes the RMS error and thus sets an upper bound on the optimal GM error. We also consider linear estimation with initial square GKP states, which was analyzed in Ref.~\cite{noh2020o2o}.

In our analysis, we choose three initial GKP lattice states: a direct product of square GKP states, a direct product of hexagonal GKP states, and a D4 lattice which can be generated from two-square GKP states by a two-mode symplectic transformation [see Eq.~\eqref{eq:Sd4} in Appendix~\ref{app:gkp_lattice}]. The D4 GKP state is necessarily entangled; hence, the D4 TMS code $\mathsf{TMS}^{\otimes2}(D_4)$ is a genuine multimode code. On the other hand, since the TMS encoding (decoding) operates on individual data modes and the additive noises are independent, the multimode square and hexagonal TMS codes, $\mathsf{TMS}^{\otimes2}(\square)$ and $\mathsf{TMS}^{\otimes2}(\hexagon)$, are simple extensions of their single-mode counterparts, i.e. $\mathsf{TMS}^{\otimes2}(\square)\simeq\mathsf{TMS}(\square)\otimes\mathsf{TMS}(\square)$ and $\mathsf{TMS}^{\otimes2}(\hexagon)\simeq\mathsf{TMS}(\hexagon)\otimes\mathsf{TMS}(\hexagon)$. 

\begin{table}
 \renewcommand{\arraystretch}{1.2}
 \setlength\tabcolsep{3ex}
    \centering
    \begin{tabular}{c | c c c}
        \hline\hline
         \diagbox[width=20ex,height=6ex, font=\normalsize]{\hspace{1ex}$\mathcal{L}$}{$\sigma$}& $0.1$ & $0.2$ & $0.5$ \\
         \hline
        Square & $0.0354$ &$0.120$ & $0.490$ \\
        Hexagonal & $0.0340$ & $0.117$ & $0.489$ \\
        D4 & $0.0322$ & $0.112$ & $0.487$ \\
        \hline\hline
    \end{tabular}
    \caption{Output geometric mean error $\Bar{\sigma}_{\rm GM}$ for a $N=M=2$ multimode GKP stabilizer code. Various lattices $\mathcal{L}$ (Square, Hexagonal, D4) and three representative values of AGN ($\sigma=0.1,\,0.2,\,0.5$) are considered. MMSE was used. See Fig.~\ref{fig:output_noise}.}
    \label{tab:relative_perform}
\end{table}

We report the output GM error $\Bar{\sigma}_{\rm GM}=\sqrt[4]{\det\bm V_{\rm out}}$ of the codes in Fig~\ref{fig:output_noise}. As shown in the figure, the D4 TMS code $\mathsf{TMS}^{\otimes2}(D_4)$ performs better than the square and hexagonal TMS codes $\mathsf{TMS}^{\otimes2}(\square)$ and $\mathsf{TMS}^{\otimes2}(\hexagon)$, at least for the few data that we generated (corresponding to $\sigma=.1,.2,.5$; see also Table~\ref{tab:relative_perform}), however we expect similar findings for $\sigma\ll1$. Indeed, using $\mathsf{TMS}^{\otimes2}(\square)$ as a benchmark~\cite{noh2020o2o}, we find the relative performance of other lattices to be better for smaller initial noise $\sigma$. For instance, at $\sigma=.01$, we find that $\mathsf{TMS}^{\otimes2}(\hexagon)$ outperforms $\mathsf{TMS}^{\otimes2}(\square)$ by a relative difference of about $5.25\%$. For $\sigma=.1$, $\mathsf{TMS}^{\otimes2}(\hexagon)$ achieves a relative difference of about $3.95\%$, whereas $\mathsf{TMS}^\otimes(D_4)$ achieves a relative difference of about $9.04\%$. For $\sigma=.2$, the relative difference is $2.5\%$ for $\mathsf{TMS}^{\otimes2}(\hexagon)$ and $6.67\%$ for $\mathsf{TMS}^\otimes(D_4)$. Finally, for $\sigma=.5$ (which is near the \QZ{break-even point} $\sigma\approx.607$), the relative difference is $.20\%$ for $\mathsf{TMS}^{\otimes2}(\hexagon)$ and $.61\%$ for $\mathsf{TMS}^\otimes(D_4)$. Numerical values of $\Bar{\sigma}_{\rm GM}$ for the different lattices and $\sigma=.1,.2,.5$ are displayed in Table~\ref{tab:relative_perform}.

We observe that single-layer ($N=M$) codes using MMSE---regardless of the initial lattice $\mathcal{L}$ \QZ{being hexagonal or square}---have a \QZ{break-even point} $\sigma^\star_{\rm MMSE}\approx.605(5)$ (teal star in Fig.~\ref{fig:output_noise}), whereas linear estimation \QZ{on square lattice} has a \QZ{break-even point} of $\sigma^\star_{\rm lin}\approx.558$~\cite{noh2020o2o} (purple diamond). \QZ{For $D4$ lattice, we narrowed the break-even point to between $0.60$ and $0.61$, which is consistent with $\sigma^\star_{\rm MMSE}$. Indeed, the break-even point is likely universal for all lattices with MMSE decoding.} The value $\sigma^\star_{\rm MMSE}\approx.605(5)$ for MMSE agrees with the lower bound on the \QZ{break-even point} ($.607\leq\sigma^\star\leq.707$) for general GKP codes discussed in the previous section. It is an open question whether this can be pushed further or not.

\QZ{
\subsection{One data mode and two ancilla modes ($N=1$ and $M=2$)}
\label{sec:N1M2}
In Sec.~\ref{sec:concatenation_reduction}, we showed that general concatenation can be reduced to ancilla preparation and then presented examples of this reduction with one data mode and two ancilla modes ($N=1,M=2$) in Sec.~\ref{sec:example_reduction}. Indeed, the reduction of encoding shows that a single TMS operation between the lone data mode and only one of the ancilla modes is required (see Fig.~\ref{fig:concatenated_codes}(b)). The other ancilla mode---which need not directly interact with the data---is in general entangled with the ancilla that interacts with the data. Code optimization then reduces to optimizing the ancillary GKP lattice state and the TMS strength. 
}

\QZ{Due to the numerical challenge in two-mode lattice optimization, we focus on the TMS concatenation code and compare a few lattices, similar to Sec.~\ref{sec:N2M2}. We examine the concatenated GKP-TMS code as it represents the best-known code for $N=1$ and $M=2$ (better than the squeezed-repetition code~\cite{wu2021continuous}). In what follows, we consider an ``upward staircase'' TMS concatenation, in which two ancillary modes (prepared in various lattices) are first coupled by a TMS operation of gain $G_2$ followed by a TMS operation of gain $G_1$ that couples the lone data mode to only one of the ancilla; see the inset of Fig.~\ref{fig:N1M2} for an illustration.\footnote{It is worth noting that this ``upward staircase'' concatenation is equivalent to a ``downward staircase'' concatenation [Fig.~\ref{fig:concatenated_codes}(d)], as demonstrated in Sec.~\ref{sec:example_reduction}.} Consequently, we can interpret the code as a scenario where a lone data mode is coupled to a correlated two-mode lattice state. To provide a more comprehensive analysis, we numerically optimize over the TMS gains $G_1,G_2$ for three different types of initial ancillary lattices: square, hexagonal and D4.}

We consider three different levels of initial noise $\sigma=0.1,0.2,0.3$ and numerically optimize the gains $(G_1,G_2)$ to minimize $\bar{\sigma}_{\rm GM}$. The quantum error correction (QEC) ratio $\bar{\sigma}_{\rm GM}/\sigma$ is plotted in Fig.~\ref{fig:N1M2}, and the corresponding optimal values of gain parameters are shown in Table~\ref{tab:N1M2}. With the additional TMS gain $G_2>1$, the product of two square or hexagonal lattices becomes a four-dimensional lattice entangled between the two ancillary modes. Moreover, the four dimensional lattice generated in such way can outperform D4 lattice, especially for the case $\sigma=0.1$ case. Although there are always caveats with numerical optimization, such results indicate that high-dimensional GKP lattice design is very much an open problem for GKP-stabilizer codes.

\begin{figure}
    \centering
    \includegraphics[width=\linewidth]{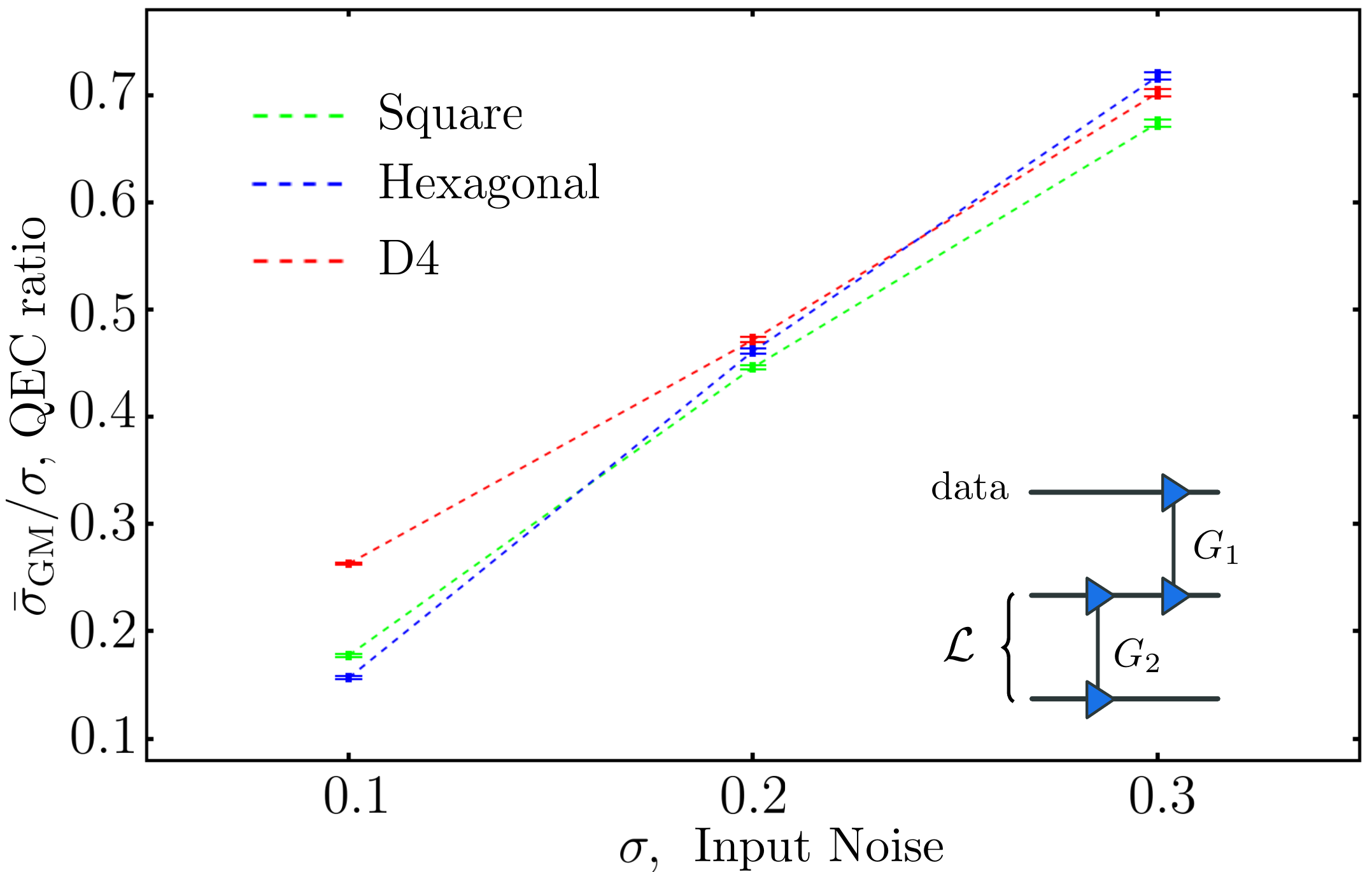}
    \caption{Quantum error correction (QEC) ratio of different lattices for a 1 data + 2 ancilla concatenated TMS code. The fences represent error bars of the local minima. Inset shows the encoding where $G_1$ and $G_2$ are the TMS gains.}
    \label{fig:N1M2}
\end{figure}

\begin{table}
    \centering
    \begin{tabular}{c|c c c}
    \hline\hline
        $\sigma$ & $0.1$ & $0.2$ & $0.3$ \\ \hline
        \rm{Square} & $(18.9,1.14)$ &  $(3.43,1.20)$ & $(1.92,1.20)$  \\ \hline
        \rm{Hexagonal} & $(22.5,2.04)$ & $(3.20,1.11)$ & $(1.75,1.11)$\\ \hline
        \rm{D4} & $(8.46,1.21)$ &$(3.21,1.04)$ & $(1.72,1.01)$\\
        \hline \hline
    \end{tabular}
    \caption{\QZ{Optimized two-mode squeezing gains $(G_1,G_2)$ for different input noise $\sigma$. $G_1$ characterizes the two-mode squeezing between the data and ancilla, while $G_2$ characterizes the squeezing that correlates the two-mode lattices. See Fig.~\ref{fig:N1M2} inset.}}
    \label{tab:N1M2}
\end{table}


\section{Heterogeneous noise case}
\label{sec:heterogeneous}

We attempt to make progress towards generalizing our findings on iid AGN to more generic sources of AGN. Consider a generic $N$ mode AGN channel $\mathcal{N}_{\bm Y}$, where $\bm Y$ is a $2N\times2N$ noise matrix that may contain correlations. From Ref.~\cite{wu2021continuous}, any correlated Gaussian noises can be reduced to a product of heterogeneous AGN channel by symplectic transformations, i.e. $\mathcal{N}_{\bm Y}\rightarrow\bigotimes_i\mathcal{N}_{\sigma_i}$, where $\sigma_i^2$ is the variance of the $i$th mode. It is thus sufficient to examine independent AGNs with different variances, for which we posit the following conjecture.

\begin{conjecture*}\label{conject:heterogenous}
Consider $N$ data modes and $N$ ancilla modes undergoing heterogeneous AGN ($\bigotimes_i\mathcal{N}_{\sigma_i}$). Then, up to unitary data processing and ancilla preparation, two- (one data and one ancilla) and four-mode interactions (two data and two ancilla) are sufficient for encoding. 
\end{conjecture*}

To promote support for this claim, consider a generic symplectic encoding matrix $\bm S_{\rm enc}$. Observe that the Gaussian channel $\mathcal{N}_{\bm Y^\prime}=\mathcal{U}_{\bm S_{\rm enc}}^{-1}\circ(\bigotimes_i\mathcal{N}_{\sigma_i})\circ\mathcal{U}_{\bm S_{\rm enc}}$ is an AGN channel with noise matrix $\bm Y^\prime=\bm S_{\rm enc}^{-1}(\bigoplus_i\sigma_i\bm I_2)\bm S_{\rm enc}^{-\top}$, which has diagonal blocks corresponding to local noises and off diagonal blocks corresponding to (generally multimode) correlated noises. By the results of Ref.~\cite{wolf2008not} (see also Ref.~\cite{caruso2008multi}), there exists local symplectic transformations $\bm\Lambda_d,\bm\Lambda_a\in{\rm Sp}(2N,\mathbb{R})$ that condense the correlation blocks to elementary two (one data and one ancilla) and four (two data and two ancilla) mode units.  Since the initial noise channels are independent, the correlation units originate from the encoding, and it certainly seems plausible that two- and four-mode interactions are sufficient to generate them. However, we do not find this argument strong enough to unequivocally validate the conjecture. 

Though we have inferred that a generic $2N$ mode encoding (encoded in $\bm S_{\rm enc}$) reduces to ``local'' two- and four-mode interactions between data and ancillae, our reduction is not constructive, in the sense that it does not give us which particular interactions to use (such as the TMS of Theorem~\ref{thm:code_redux}). Moreover, as stated, our reduction \QZ{in the heterogeneous case is only expected to} hold when the number of data modes equals the number of ancilla modes ($N=M$) and thus does not straightforwardly apply to \QZ{codes where the number of ancilla modes is greater than the number of data modes (such as concatenated codes).}

\section{Analyses on finite-squeezing}\label{sec:approx_gkp}

\QZ{
A GKP state $\ket{\mathcal{L}}$ has support on the entire phase space of the modes and thus has infinite energy. We can regularize the state by confining the energy to a ball of radius $\sim\Delta^{-1}$ in phase space via
\begin{equation}
    \ket{\mathcal{L}(\Delta)}\propto\exp\left(-\frac{\Delta^2\hat{\bm r}^\top\hat{\bm r}}{2}\right)\ket{\mathcal{L}};
\end{equation}
see also Refs.~\cite{albert2018GKPcapacity,PhysRevA.101.012316,noh2020o2o,royer2020stabilization}. Note that $\hat{\bm r}^\top\hat{\bm r}/2=\sum_{i=1}^N\hat{n}_i+N/2$. The regularizer (or envelope operator) has a nice form in the displacement operator basis,
\begin{equation}
    \exp\left(-\frac{\Delta^2\hat{\bm r}^\top\hat{\bm r}}{2}\right)\propto\int\dd{\bm \mu}\exp\left(-\frac{\bm\mu^2}{4\tanh(\Delta^2/2)}\right)D_{\bm \mu}.
\end{equation}
We can thus think of the regularization procedure as coherently applying displacements (with a Gaussian envelope) on the GKP state. To simplify the analysis, we can twirl the regularized state~\cite{PhysRevA.101.012316} to generate an incoherent GKP state, $\bigotimes_{i=1}^N\mathcal{N}_{\sigma_{\rm GKP}}(\dyad{\mathcal{L}})$, where we define the finite GKP noise per mode $\sigma_{\rm GKP}^2=\tanh(\Delta^2/2)$; the relation between finite squeezing and the variance of the GKP noise is $s_{\rm{GKP}}=-10\log_{10}(2\sigma_{\rm{GKP}}^2)$. This twirling is not physical as it requires an infinite amount of energy and is only used for computational convenience. We use this noise model to approximate finite GKP noise in all of our calculations. Specifically, we assume that the GKP ancilla modes as well as the GKP states used for measurements are all noisy. We thus have two layers of finite GKP squeezing noise which enter into the analysis in slightly different ways as discussed below.
}

\begin{figure}[t]
    \centering
    \includegraphics[width=\linewidth]{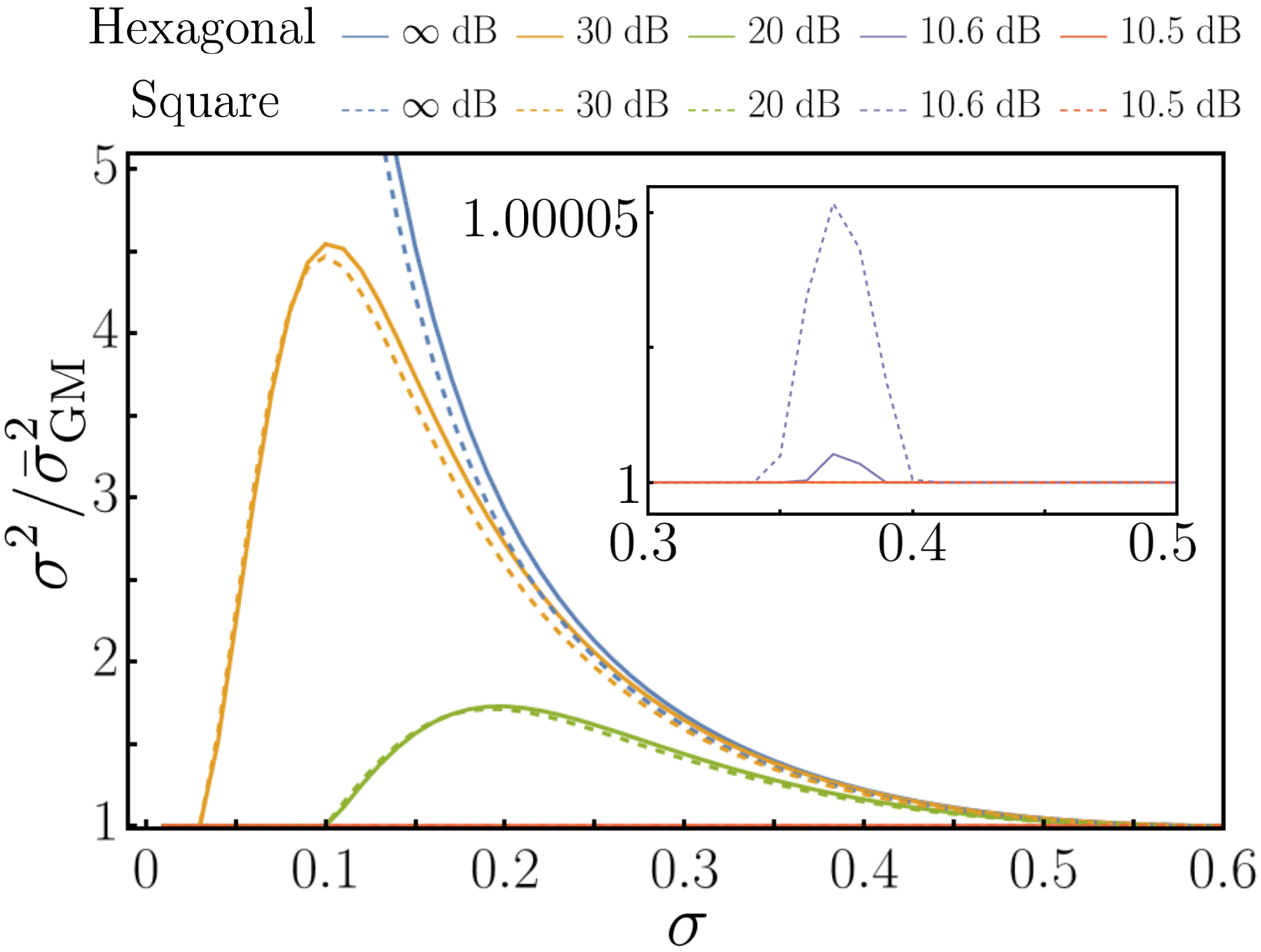}
    \caption{Quantum error correction (QEC) gain $\sigma^2/\sigma_{\rm{GM}}^2 $ versus input noise $\sigma$ for finite GKP squeezing in dB $s_{\rm GKP}=(10.5, 10.6,20,30,\infty)$. Solid lines correspond to hexagonal lattice. Dashed lines correspond to square lattice. }\label{fig:finiteGKP}
\end{figure}

\QZ{
With the finite-squeezed GKP states as encoding ancilla and measurement ancilla, the final covariance matrix defined in Eq.~\eqref{eq:MMSE_submatrices} changes to
\begin{align}
    \begin{pmatrix}
    \bm V_d & \bm V_{da}\\
    \bm V_{da}^T & \bm V_a
    \end{pmatrix}^{-1}\equiv &
    (\bI_{2N} \oplus {\bm M^\top \bm \Omega}) \bm V_{\bx} (\bI_{2N} \oplus {(\bm M^\top \bm \Omega)} ^{\top}) \nonumber
    \\
    & + \bm{0}_{2N}\oplus \sigma_{\rm{GKP}}^2(\bm M^\top \bm M+\bm I_{2M}).
    \label{eq:MMSE_finiteGKP_submatrices}
\end{align}
The GKP noise proportional to $\bm M^\top\bm M$ originates from the GKP lattice ancilla whereas the GKP noise proportional to the identity comes from noisy measurements; see Appendix \ref{app_analysis_approximatedGKP} for a derivation. We thus see that all computations from previous sections carry over (e.g., the MMSE formula Eq.~\eqref{eq:fMMSE_main} applies) with the updated covariance matrices to incorporate GKP noise. 
}

\QZ{
In Fig.~\ref{fig:finiteGKP}, we plot the QEC gain $\sigma^2/\bar{\sigma}_{\rm GM}^2$ versus the input AGN $\sigma$ for hexagonal and square GKP with finite GKP squeezing. We assume MMSE estimation. The functional behaviors are similar to that presented in Ref.~\cite{noh2020o2o} for a square GKP and linear estimation. Though, an important distinction here is that we find a \QZ{break-even} squeezing of around $s_{\rm dB}=10.5$dB, below which there is no QEC gain, whereas linear estimation leads to a higher value of 11 dB~\cite{noh2020o2o}. Note that at each squeezing level, when $\sigma\approx \sigma_{\rm GKP}$, the GKP finite-squeezing noise significantly contributes to the noise budget, and the square lattice is slightly better than the hexagonal lattice due to a smaller ${\rm det}(\bm M^\top \bm M+\bm I_{2M})$, with values $2$ for square and $2+4/\sqrt{3}$ for hexagonal. 
}

\section{Discussions and conclusions}
\label{sec:discussion}

In this paper, we derived the optimal form of GKP-stabilizer codes---TMS encoding on GKP ancilla with general lattices---for homogenous input noise. In the case of single-mode data and ancilla, we identified the GKP-TMS code with a two-dimensional hexagonal GKP ancilla to be optimal. For higher dimensions, we found D4 lattices to be superior to lower dimensional lattices. \QZ{We were also able to prove a universal no-threshold theorem for \emph{all} oscillators-to-oscillators codes based on Gaussian encodings and showed that the continuous errors on the data cannot be made arbitrarily small without an infinite amount of squeezing. This is not too surprising because these errors are intrinsically analog and may only be suppressed with continuous variable resources (e.g., squeezing).}

Before closing, a few open questions are worth mentioning. We expect the D4 lattice to perform well and thus picked it for benchmarking, however there might be better lattices in four or higher dimensions. \QZ{Searching for good lattices in higher dimension can in general be challenging, as the number of free parameters grow with the number of modes quadratically.} Although we derived the minimize mean square estimator (MMSE), the optimal estimator for minimizing the geometric mean error is unknown (to our knowledge). We narrowed the range of optimal break-even point of additional noise level down to the range $.607\leq\sigma^\star\leq.707$ from quantum capacity bounds and found MMSE reaching the lower end of this range; the actual optimal break-even point is an open problem, although we expect it to be closer to the lower end. Finally, regarding heterogeneous noise sources, we have conjectured that two- and four-mode interactions are sufficient for general (Gaussian) encodings, however we do not have definitive nor constructive proof of this claim. Ideal code design for heterogeneous noise sources is thus open; dedicated numerical studies may be required to address such.

\begin{acknowledgements}
This project is supported by the Defense Advanced Research Projects Agency (DARPA) under Young Faculty Award (YFA) Grant No. N660012014029. QZ also acknowledges support from NSF CAREER Award CCF-2142882 and National Science Foundation (NSF) Engineering Research Center for Quantum Networks Grant No. 1941583.
\end{acknowledgements}

\bibliographystyle{quantum}

\clearpage

\appendix

\noindent{\huge \bf Appendix}

\section{Notation}

\begin{figure}[b]
    \centering
    \includegraphics[width=\linewidth]{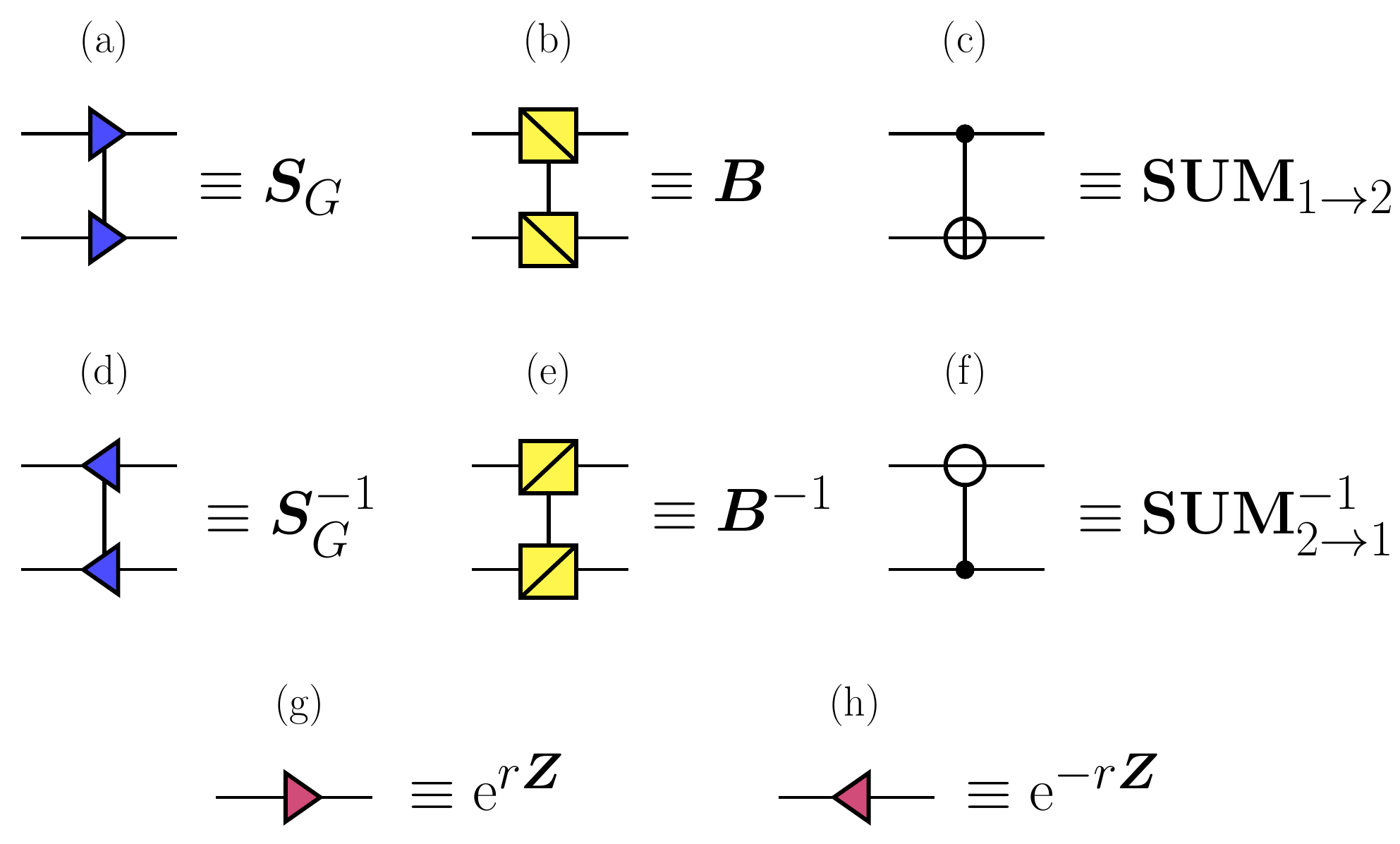}
    \caption{Dictionary of circuit elements for (a and d) two-mode squeezers, (b and e) beamsplitters, (c and f) SUM gates, and (g and h) single-mode squeezers.}
    \label{fig:circuit_dictionary}
\end{figure}

Consider an $N$-mode bosonic Hilbert space $\mathscr{H}^{\otimes N}$, where $\mathscr{H}$ is the Hilbert space of a single bosonic mode. Define the vector of canonical operators for the $N$ modes as,
\begin{equation}
    \hat{\bm r}^\top\equiv\left(\hat{q}_1,\hat{p}_1,\dots,\hat{q}_N,\hat{p}_N\right),
\end{equation}
such that,
\begin{equation}
    [\hat{\bm r}_k,\hat{\bm r}_j]={\rm i} \bm{\Omega}_{kj},
\end{equation}
where $\bm\Omega$ is the $N$-mode symplectic form,
\begin{equation}
    \bm\Omega = \bigoplus_{i=1}^N\bm\Omega_1 \qq{with} 
    \bm\Omega_1=
    \begin{pmatrix}
        0 & 1 \\
        -1 & 0
    \end{pmatrix}.
\end{equation}
The canonical operators $\hat{q}$ and $\hat{p}$ of a single mode are the real and imaginary parts, respectively, of the annihilation operator, $\hat{a}$. Throughout the appendix, we write `hat' on operators only in places where operators can potentially be confused with numbers. For example, in quadrature operators; While we omit `hat' for unitary operators and others that are easily recognized as operators. 

In several figures, we draw quantum circuits with the circuit elements representing symplectic transformations $\bm S$ corresponding to a unitary $U_{\bm S}$. A dictionary of commonly used circuit elements is shown in Fig.~\ref{fig:circuit_dictionary}.

\section{Gaussian Evolution}\label{app:gauss_evol}

\subsection{Gaussian unitaries}\label{app:gauss_unitaries}
Given an $N$-mode symplectic transformation $\bm S\in{\rm Sp}(2N,\mathbb{R})$, where ${\rm Sp}(2N,\mathbb{R})$ is the set of $2N\times2N$ real symplectic matrices (of dimension $\abs{{\rm Sp}(2N,\mathbb{R})}=2N^2+N$), such that $\bm S\bm\Omega\bm S^\top=\bm\Omega$ and $\det\bm S=1$, one can find a unitary representation $U_{\bm S}$, which encodes the symplectic transformation $\bm S$ and acts on the canonical operators as,
\begin{equation}
    U_{\bm S}^\dagger\hat{\bm r}U_{\bm S} = \bm S \hat{\bm r}.\label{eq:r_S}
\end{equation}

A useful fact that holds for an arbitrary symplectic matrix is the so-called \textit{Bloch-Messiah (or Euler) decomposition}~\cite{braunstein05}. The statement is that a symplectic matrix $\bm S\in{\rm Sp}(2N,\mathbb{R})$ decomposes to
\begin{equation}\label{eq:bloch_messiah}
    \bm S = \bm B^\prime\cdot\left(\bigoplus_{i=1}^N{\rm e}^{r_i\bm Z}\right)\cdot\bm B,
\end{equation}
where $\bm Z$ is the Pauli-Z matrix and ${\rm e}^{r_i\bm Z}$ is a single mode squeezing transformation on the $i$th mode with squeezing strength $r_i$. Here $\bm B^\prime,\bm B\in{\rm Sp}(2N,\mathbb{R})\cap{\rm SO}(2N)\simeq U(N)$, where $U(N)$ is the unitary group of dimension $\abs{U(N)}=N^2$. In other words, $\bm B^\prime$ and $\bm B$ are passive (e.g., linear optical) transformations, which admit an even further decomposition in terms of two-mode beamsplitters and phase-shifts~\cite{reck94}. Since the squeezing transformation is a diagonal matrix with $N$ free parameters, it follows that $\abs{{\rm Sp}(2N,\mathbb{R})}=2\abs{U(N)}+N=2N^2+N$, as mentioned previously.

Consider vectors $\bm\mu,\bm\nu\in\mathbb{R}^{2N}$ and define the Weyl (displacement) operator,
\begin{equation}
    {D}_{\bm\mu}\equiv\exp\left({\rm i}{\bm\mu}^\top\bm\Omega\hat{\bm r}\right),\label{eq:weyl_op}
\end{equation}
which form an operator basis for the space of bounded operators $\mathcal{B}(\mathscr{H}^{\otimes N})$ via $\Tr({D}_{\bm\mu}{D}_{-\bm\nu})=(2\pi)^N\delta^{2N}(\bm\mu-\bm\nu)$~\cite{serafini2017book}. Displacement operators satisfy a composition rule, 
\begin{equation}
    {D}_{\bm\mu}{D}_{\bm\nu}={\rm e}^{-{\rm i}\omega(\bm\mu,\bm\nu)}{D}_{\bm\nu}{D}_{\bm\mu},\label{eq:weyl_comp}
\end{equation}
where $\omega(\bm\mu,\bm\nu)\equiv\bm\mu^\top\bm\Omega\bm\nu$. The anti-symmetric bilinear form $\omega:\mathbb{R}^{2N}\times\mathbb{R}^{2N}\rightarrow\mathbb{R}$ takes as input the real vectors $\bm\mu$ and $\bm\nu$ and computes $\omega(\bm\mu,\bm\nu)\in\mathbb{R}$, which is called the symplectic inner product between $\bm\mu$ and $\bm\nu$. The symplectic inner product obeys $\omega(\bm\nu,\bm\mu)=-\omega(\bm\mu,\bm\nu)$ and is invariant under symplectic transformations, i.e. $\omega(\bm S\bm\mu,\bm S\bm\nu)=\omega(\bm\mu,\bm\nu)$. Finally, from Eq.~\eqref{eq:r_S} and the general conjugation formula $U^\dagger f(\hat{g})U=f(U^\dagger\hat{g}U)$, it can be shown that the displacement operators transform under symplectic transformations via
\begin{equation}\label{eq:weyl_S}
    U_{\bm S}{D}_{\bm\mu}U_{\bm S}^\dagger={D}_{\bm S\bm\mu},
\end{equation}
where $\bm S\bm\mu\in\mathbb{R}^{2N}$.

We now provide explicit symplectic matrices for some commonly used Gaussian unitaries in the single- and two-mode cases. For a single-mode, we only have two possible transformations, which are a single-mode squeezer and phase rotation. For two-mode operations, some common elements used in this paper and in the literature are a beamsplitter, TMS, and a SUM gate. 

A single-mode squeezer with squeezing strength $r$ has a symplectic matrix representation,
\begin{equation}
    \textbf{Sq}({\rm e}^r)\equiv {\rm e}^{r\bm Z},
\end{equation}
and a phase rotation at angle $\phi$ has a representation,
\begin{equation}
    \textbf{R}(\phi)\equiv
    \begin{pmatrix}
    \cos\phi & \sin\phi \\
    -\sin\phi & \cos\phi
    \end{pmatrix}.
\end{equation}

The two-mode beamsplitter has a symplectic representation
\begin{equation}\label{eq:bs_matrix}
    \bm B = 
    \begin{pmatrix}
        \cos\theta\bm I & \sin\theta\bm I\\
        -\sin\theta\bm I & \cos\theta\bm I
    \end{pmatrix},
\end{equation}
where, e.g., $\cos^2\theta$ is the transmittance of the beamsplitter. For a 50:50 beamsplitter ($\theta=\pi/4$), we express the symplectic matrix as $\bm B_{1/2}$. [For a beamsplitter with transmittance $T$, we may write $\bm B_{T}$.]

\begin{figure}
    \centering
    \includegraphics[width=.6\linewidth]{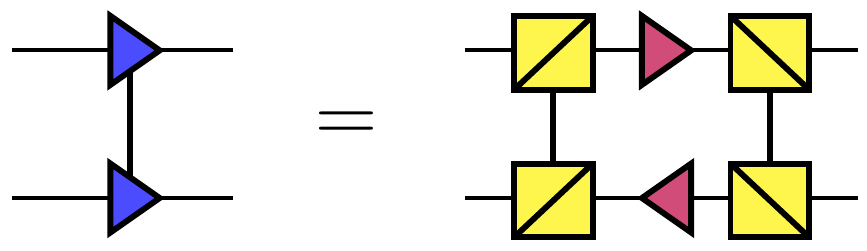}
    \caption{Bloch-Messiah decomposition for a TMS operation in terms of 50:50 beamsplitters and single mode squeezing.}
    \label{fig:tms_blochmessiah}
\end{figure}

The SUM-gate ${\rm SUM}_\delta\equiv{\rm e}^{-{\rm i}\delta \hat{Q}_1\otimes\hat{P}_2}$ operates on the canonical operators as,
\begin{align}
\begin{aligned}
     \hat{Q}_1&\rightarrow\hat{Q}_1,& \qq{} \hat{P}_1&\rightarrow \hat{P}_1-\delta\hat{P}_2,&\\
    \hat{Q}_{2}&\rightarrow\hat{Q}_{2}+\delta\hat{Q}_1,&\qq{}\hat{P}_{2}&\rightarrow\hat{P}_{2}.&
\end{aligned}
\end{align}
where $\delta\in\mathbb{R}$ and has symplectic representation,
\begin{align}
    \textbf{SUM}_\delta&=
    \begin{pmatrix}
    \bm{I}_2 & -\delta\bm\Pi_P\\
    \delta\bm\Pi_Q & \bm{I}_2
    \end{pmatrix}, \label{eq:sum_matrix}\\ \textbf{SUM}^{-1}_\delta&=
    \begin{pmatrix}
    \bm{I}_2 & \delta\bm\Pi_P\\
    -\delta\bm\Pi_Q & \bm{I}_2
    \end{pmatrix},\label{eq:inversesum_matrix}
\end{align}
where $\bm\Pi_Q={\rm diag}(1, 0)$ and $\bm\Pi_P={\rm diag}(0, 1)$ represent projections along the $Q$ quadrature and $P$ quadrature of the respective modes. For $\delta=1$, we use the notation $\textbf{SUM}\equiv\textbf{SUM}_{\delta=1}$, which is the conventional definition of the SUM-gate in the literature. Observe that $\textbf{SUM}_{\delta_1}\cdot\textbf{SUM}_{\delta_2}= \textbf{SUM}_{\delta_1+\delta_2}$. The SUM-gate is the CV analog to the CNOT gate for qubit-into-oscillator codes~\cite{gkp2000} and thus useful for ancilla-assisted stabilizer measurements (see Section~\ref{app:measurements}).

A TMS transformation with gain $G$ has symplectic representation,
\begin{equation}\label{eq:tms_matrix}
    \bm S_{G}=
    \begin{pmatrix}
        \sqrt{G}\bm I & \sqrt{G-1}\bm Z\\
        \sqrt{G-1}\bm Z & \sqrt{G}\bm I
    \end{pmatrix}.
\end{equation}
Parameterizing the gain $G$ in terms of the squeezing strength $r$ via $G=\cosh^2r$, one can provide a Bloch-Messiah decomposition for TMS in terms of two-mode beamsplitters and single-mode squeezers,
\begin{equation}\label{eq:tms_decomp}
    \bm S_{G}= \bm B_{1/2}\cdot\Big(\textbf{Sq}({\rm e}^{r})\oplus\textbf{Sq}({\rm e}^{-r})\Big)\cdot\bm B_{1/2}^\top.
\end{equation}
See Fig.~\ref{fig:tms_blochmessiah} for an illustration.


\subsection{Gaussian channels}\label{app:gauss_channels}
Consider an initial separable quantum state ${\Psi}\otimes{\rho}_E$, where $\Psi\in\mathscr{H}^{\otimes N}$ is a $N$ mode quantum state of the system (not necessarily Gaussian) and ${\rho}_E\in\mathscr{H}^{\otimes M}$ is a $M$ mode Gaussian quantum state of the environment. Given a symplectic transformation $\bm S$ with unitary representation $U_{\bm S}$, we define a Gaussian quantum channel $\mathcal{G}:\mathscr{H}^{\otimes N}\rightarrow\mathscr{H}^{\otimes N}$ which acts on the system modes via
\begin{equation}
    \mathcal{G}({\Psi})\equiv\Tr_E\left[U_{\bm S}\left({\Psi}\otimes{\rho}_E\right)U_{\bm S}^\dagger\right].
\end{equation}
Let the symplectic matrix $\bm S$ be written in the following form, 
\begin{equation}\label{eq:S_partition}
    \bm S =
\begin{pmatrix}
    \bm A & \bm B \\
    \bm C & \bm D
\end{pmatrix},
\end{equation}
where $\bm A$ is a $2N\times 2N$ matrix that dictates internal evolution of the system and $\bm B$ is a $2N\times 2M$ rectangular matrix which encodes the interaction between environment and system. [For a system of $N$ modes evolving under the Gaussian channel $\mathcal{G}$, it is sufficient to choose $M\leq2N$ to fully characterize the channel~\cite{weedbrook2012rmp}.] From thus, one can prove the following theorem, 
\begin{theorem}[Gaussian channel characterization]\label{thm:Gaussian_ch_descript}
Given a Gaussian environment state ${\rho}_E$---characterized by the first and second moments $\bm \mu_E$ and $\bm\sigma_E$---a Gaussian quantum channel $\mathcal{G}:\mathscr{H}^{\otimes N}\rightarrow\mathscr{H}^{\otimes N}$ is completely determined by a displacement-noise vector $\bm d$, a scaling matrix $\bm X$, and a noise matrix $\bm Y$ such that
\begin{align}
    \bm d &= \bm B\bm\mu_E,\\
    \bm X &= \bm A,\\
    \bm Y &= \bm B\bm\sigma_E\bm B^\top,
\end{align}
where $\bm A$ and $\bm B$ are sub-matrices of the symplectic matrix $\bm S$ [Eq.~\eqref{eq:S_partition}] that couples the system to the environment. 
\end{theorem}
The displacement-noise vector $\bm d$ can always be subsumed into a unitary displacement occurring on the system state; thus of primary interest to the dynamics are the scaling and noise matrices, $\bm X$ and $\bm Y$, respectively. 

Another useful fact to know about Gaussian channels is how their noise vectors, scaling matrices, and noise matrices combine when we concatenate Gaussian channels. 
\begin{theorem}[Gaussian channel synthesis]\label{thm:synthesis}
Consider two Gaussian channels $\mathcal{G}_1$ and $\mathcal{G}_2$ with characterizations $(\bm d_1,\bm X_1,\bm Y_1)$ and $(\bm d_2,\bm X_2,\bm Y_2)$, respectively. Then the composite channel ${\mathcal{G}_{12}=\mathcal{G}_2\circ\mathcal{G}_1}$, which is also a Gaussian channel, has a characterization $(\bm d_{12},\bm X_{12},\bm Y_{12})$ that take the form~\cite{holevo2007classification,albert2018GKPcapacity}
\begin{align}
    \bm d_{12}&=\bm d_2 + \bm X_2\bm d_1,\\
    \bm X_{12}&=\bm X_2\bm X_1,\\
    \bm Y_{12}&= \bm X_2\bm Y_1\bm X_2^\top +\bm Y_2.
\end{align}
\end{theorem}


\subsection{Condensing correlations}

We discuss the condensation of correlations---from multimode correlations to elementary pair-like correlations---for general Gaussian transformations.


\subsubsection{Phase-space Schmidt decomposition}
As highlighted in the main text (Theorem~\ref{thm:code_redux}), for the special case of iid AGN, we can significantly reduce the freedom in a generic GKP-stabilizer code by effectively reducing the code to a direct product of TMS codes and a general ancillary GKP lattice $\mathcal{L}$. The workhorse in proving Theorem~\ref{thm:code_redux} is the \textit{modewise entanglement theorem} just below. 

\begin{theorem}[Modewise entanglement~\cite{reznik2003modewise}]\label{thm:modewise}
Consider a subsystem $A$ of $N$ modes and a subsystem $B$ of $M$ modes, such that the joint system $AB$ consists of $K=M+N$ modes, and define the positive definite matrix $\bm{\varrho}_{AB}=\varrho\bm S\bm S^\top$, where $\varrho\in(0,\infty)$ and $\bm S\in{\rm Sp}(2K,\mathbb{R})$. Then, there exists local symplectic matrices $\bm\Lambda_A\in{\rm Sp}(2N,\mathbb{R})$ and $\bm\Lambda_B\in{\rm Sp}(2M,\mathbb{R})$ such that,
\begin{multline}
    \left(\bm\Lambda_A\oplus\bm\Lambda_B\right)\bm{\varrho}_{AB}\left(\bm\Lambda_A^\top\oplus\bm\Lambda_B^\top\right)=\\\varrho\left(\bigoplus_{i=1}^{N}\bm S_{G_i}\bm S_{G_i}^{\top}\right)\oplus\bm I_{K-2N},
\end{multline}
where $\bm S_{G_i}$ is a TMS squeezing operation of gain $G_i$ [see, e.g., Eq.~\eqref{eq:tms_matrix}] between the $i$th mode in $A$ and the $(N+i)$th mode in $B$.
\end{theorem}

Theorem~\ref{thm:modewise} is sometimes referred to as the \textit{phase-space Schmidt decomposition} for Gaussian states~\cite{adesso2007thesis}; see Refs.~\cite{reznik2003modewise} and~\cite{serafini2017book} for proofs and Ref.~\cite{adesso2005localizable,adesso2007thesis} for more discussion and a generalization to `bisymmetric' noise. 

As an illustrative example, let $\bm\varrho_{AB}$ be the covariance matrix for an isotropically mixed $K$ mode Gaussian state $\Psi_{AB}$ ($\varrho\geq1$). Then, by Lemma~\ref{thm:modewise}, the state $\Psi_{AB}$ is locally equivalent (up to a Gaussian unitary $U_{\bm\Lambda_A}^\dagger\otimes U_{\bm\Lambda_B}^\dagger$) to a direct product of $N$ (noisy) TMS vacuum states and $K-2N$ uncorrelated thermal states. In other words, $\Psi_{AB}$ is modewise entangled---i.e., each mode in $A$ is entangled with only one corresponding mode in $B$.

\subsubsection{Not-so-normal mode decomposition}

We briefly discuss a generalization of the normal mode decomposition (see, e.g., Refs.~\cite{weedbrook2012rmp,serafini2017book}) to a ``not-so-normal'' mode decomposition~\cite{wolf2008not,caruso2008multi}. We utilize this decomposition in our argument for the reduction of GKP codes against heterogeneous AGN. It is easier to state the results of Ref.~\cite{wolf2008not} by using the following order of canonical operators $\hat{\bm r}=(\hat{q}_1,\hat{q}_2,\dots,\hat{p}_1,\hat{p_2},\dots)^\top$.
\begin{theorem}[Not-so-normal mode decomposition~\cite{wolf2008not}]\label{thm:not_so_norm}
Consider a $2N\times2N$ invertible, non-defective matrix $\bm X$. Then, there exists symplectic transformations $\bm S_A,\bm S_B\in{\rm Sp}(2N,\mathbb{R})$ such that
\begin{equation}
    \bm S_A\bm X\bm S_B=
    \begin{pmatrix}
    \bm I_N & \bm 0\\
    \bm 0 & \bm J_{\mathbb{R}}
    \end{pmatrix},\label{eq:not_so_norm}
\end{equation}
where $\bm J_{\mathbb{R}}=\bigoplus_{k=1}^N\bm J(\lambda_k)$ is a $N\times N$ block-diagonal matrix consisting of the (two-fold degenerate) eigenvalues $\{\lambda_k\}_{k=1}^N$ of the matrix $\bm X\bm\Omega^\top\bm X^\top\bm \Omega$. For complex $\lambda_k=a_k+ib_k$, the diagonal blocks are $2\times 2$ matrices $\bm J_k(\lambda_k)=\big(\begin{smallmatrix}
a_k & b_k\\
-b_k & a_k
\end{smallmatrix}\big)$. For real $\lambda_k$, $\bm J_k(\lambda_k)=\lambda_k$.
\end{theorem}

As a technical aside, for defective $\bm X$, the above theorem still holds, however in that case, $\bm J_{\mathbb{R}}$ is an $N\times N$ matrix in real Jordan form (hence the subscript), with Jordan blocks containing the degenerate eigenvalues of $\bm X$.

As an example, if $\bm X$ is a positive definite real matrix, then the above is effectively equal to the normal mode decomposition. Recall from the normal mode decomposition that there exists $\bm S_A^\prime,\bm S_B^\prime\in{\rm Sp}(2N,\mathbb{R})$ such that $\bm S_A^\prime\bm X\bm S_B^\prime=\bm{\nu}\oplus\bm{\nu}$  where $\bm\nu={\rm diag}(\nu_1,\nu_2,\dots,\nu_N)$ and $\nu_k$ are the symplectic eigenvalues of $\bm X$~\cite{weedbrook2012rmp,serafini2017book}. The eigenvalues $\lambda_k$ of $\bm X\bm\Omega^\top\bm X^\top\bm \Omega$ are real in this case and are related to the symplectic eigenvalues $\nu_k$ of $\bm X$ via $\lambda_k=\nu_k^2$. Indeed, one can show that the decomposition in Eq.~\eqref{eq:not_so_norm} and the normal mode decomposition just above are equivalent up to local squeezing transformations $\bigoplus_k\textbf{Sq}({\rm e}^{r_k})$, with squeezing strengths ${\rm e}^{r_k}=\nu_k$. 

\begin{figure}
    \centering
    \includegraphics[width=.7\linewidth]{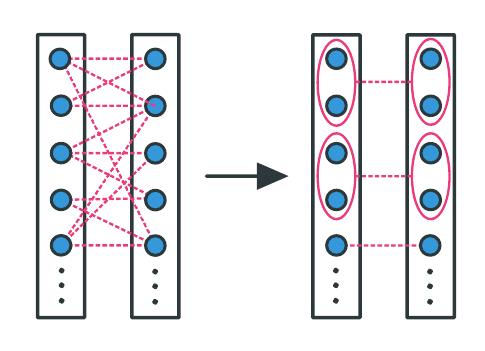}
    \caption{Multimode correlations condense to elementary two- and four-mode units (see Theorem~\ref{thm:not_so_norm} below and cf. Ref.~\cite{wolf2008not}).}
    \label{fig:not_so_norm}
\end{figure}

As another example, consider a $4N\times4N$ covariance matrix $\bm\sigma$ for an $2N$ mode Gaussian quantum state. Let us go back to $(q_1,p_1,q_2,p_2,\dots)$ ordering; this will allow us to write the covariance matrix in block form as below. Then, consider a bipartite cut which partitions a subsystems $A$ of $N$ modes from a subsystem $B$ of $N$ modes. We write the covariance matrix in blocks as
\begin{equation}
    \bm\sigma=\begin{pmatrix}
    \bm\sigma_A & \bm\sigma_{AB} \\
    \bm\sigma_{AB}^\top & \bm \sigma_B
    \end{pmatrix},
\end{equation}
where $\bm\sigma_A,\bm\sigma_B$ are local covariance matrices and $\bm\sigma_{AB}$ is the correlation matrix. Consider local symplectic transformations $\bm S_A,\bm S_B\in\rm{Sp}(2N,\mathbb{R})$ such that $\bm\sigma\rightarrow(\bm S_A\oplus\bm S_B)\bm\sigma(\bm S_A^\top\oplus\bm S_B^\top)$. The correlation matrix transforms as $\bm\sigma_{AB}\rightarrow\bm S_A\bm\sigma_{AB}\bm S_B^\top$. By Theorem~\ref{thm:not_so_norm}, we can choose $\bm S_A$ and $\bm S_B$ to reduce $\bm\sigma_{AB}$ into ``not-so-normal'' form, completely determined by Jordan blocks $\bm J(\lambda_k)$. This transformation condenses the correlations into two- and four-mode units. In particular, each $\lambda_k\in\mathbb{R}$ implies two-mode correlations, whereas $\lambda_k\in\mathbb{C}$ implies four-mode correlations; see Fig.~\ref{fig:not_so_norm} for an illustration.


\section{GKP lattice states}\label{app:gkp_lattice}
We present some details regarding lattice GKP states for $N$ modes and, for concreteness, provide explicit examples for one and two modes; see also Refs.~\cite{baptiste2022multiGKP, eisert2022lattice}.

\subsection{Modular quadratures}
Consider a GKP lattice state $\ket{\mathcal{L}}$ and the stabilizer group $\mathcal{S}_{\mathcal{L}}=\langle S^{\mathcal{L}}_1,\dots,S^{\mathcal{L}}_{2N}\rangle$, such that $\mathcal{S}_{\mathcal{L}}\ket{\mathcal{L}}=\ket{\mathcal{L}}$. Since ${S}^{\mathcal{L}}_J=\exp({\rm i}\hat{g}_J^\mathcal{L})$, where $\hat{g}_J^\mathcal{L}\equiv \bm\lambda_J^{\mathcal{L}\,\top}\bm\Omega\bm{\hat{r}}$, we can equivalently say that $\ket{\mathcal{L}}$ is an eigenstate of the modular quadratures $\{\hat{g}_J^{\mathcal{L}}\}$ (which are the generators of translations in the lattice) with eigenvalue $0\mod{2\pi}$. Given a generator matrix $\bm M$ of $\mathcal{L}$, we can package all the modular quadratures quite nicely into a vector operator,
\begin{equation}\label{eq:g_Mrelation}
    \hat{\bm g}^{\mathcal{L}}=\bm M^\top\bm\Omega\hat{\bm r}=
    \begin{pmatrix}
        \hat{g}_1^\mathcal{L} \\
        \hat{g}_2^\mathcal{L} \\
        \vdots
    \end{pmatrix}.
\end{equation}
If the lattice is a symplectic self-dual lattice~\cite{harrington2001rates} (see below), such that $\bm M^\top\bm \Omega\bm M=2\pi\bm\Omega$, then ${\hat{\bm g}^{\mathcal{L}}=\bm \Omega\bm M^{-1}\hat{\bm r}}$ and $\hat{\bm g}^{\mathcal{L}}$ defines a new set of quadrature operators, up to a factor $\ell=\sqrt{2\pi}$.

\subsection{Canonical GKP state}
The prototypical example of a lattice state is the canonical (square) GKP state $\ket{\square}$, which one can define via the stabilizers,
\begin{align}
    {S}^{\,\square}_{1}&\equiv\exp({\rm i}{\bm\lambda}^{\square\,\top}_1\bm\Omega_1\hat{{\bm r}})={\rm e}^{{\rm i}\ell\hat{q}},\label{eq:square1_stabilizer}\\
    {S}^{\,\square}_{2}&\equiv\exp({\rm i}{\bm\lambda}^{\square\,\top}_2\bm\Omega_1\hat{{\bm r}})={\rm e}^{{\rm i}\ell\hat{p}},\label{eq:square2_stabilizer}
\end{align}
where $\ell\equiv\sqrt{2\pi}$, such that ${S}^{\,\square}_{K}\ket{\square}=\ket{\square}\,\forall\,k\in\{1,2\}$. The vectors $\bm\lambda^{\square}_K$ are given explicitly as
\begin{align}
    {\bm\lambda}^{\square}_1&\equiv\ell\left(0,-1\right)^\top,\\
    {\bm\lambda}^{\square}_2&\equiv\ell\left(1,0\right)^\top.
\end{align}
Observe that $[{S}^{\,\square}_{1},{S}^{\,\square}_{2}]=0$ since $\omega\left({\bm\lambda}^{\square}_1,{\bm\lambda}^{\square}_2\right)=\ell^2=2\pi$.
For reference, a generator matrix built from these vectors is
\begin{equation}
    \bm M^{\square}=\left({\bm\lambda}^{\square}_1\hspace{.6em}{\bm\lambda}^{\square}_2\right)
    =\ell\bm\Omega_1,
\end{equation}
which is just the symplectic form in $\mathbb{R}^2$. The modular quadratures $\hat{\bm g}^{\square}$ are thus equivalent (up to a factor $\ell$) to the canonical operators $\hat{q}$ and $\hat{p}$, as explicitly seen above for the stabilizers~\eqref{eq:square1_stabilizer}-\eqref{eq:square2_stabilizer}. One may therefore refer to the quadrature basis for the square GKP state as the `canonical' or `standard' basis. In terms of the basis states for $\hat{q}$ and $\hat{p}$, the canonical GKP state can be written as
\begin{equation}
    \ket{\square}\propto\sum_{n\in\mathbb{Z}}\ket{q=n\ell}=\sum_{n\in\mathbb{Z}}\ket{p=n\ell},
\end{equation}
which is a translation invariant square lattice in the single-mode phase space with period $\ell=\sqrt{2\pi}$. 

We denote $N$ canonical GKP states as $\ket{\square^N}\equiv\ket{\square}^{\otimes N}$, corresponding to a $2N$ dimensional hypercube. The multimode lattice can be described by the vectors $\bm\lambda^\square_{J}$, where $J=\{1,\dots,2N\}$, such that, e.g.,
\begin{align*}
    \bm\lambda_{1}^\square&=(0,-\ell)^\top\oplus\bm 0_{2(N-1)}^\top,\\ \bm\lambda_{2}^\square&=(\ell,0)^\top\oplus\bm 0_{2(N-1)}^\top,\\
    & \vdotswithin{=} \\
    \bm\lambda_{2N-1}^\square&=\bm 0_{2(N-1)}^\top\oplus(0,-\ell)^\top,\\ 
    \bm\lambda_{2N}^\square&=\bm 0_{2(N-1)}^\top\oplus(\ell,0)^\top,
\end{align*}
where $\bm 0_{2(N-1)}$ is a $2(N-1)$ dimensional zero vector. The generator matrix is the symplectic form on $N$ modes, which is simply a direct sum of the generator matrices for the individual modes,
\begin{equation}
    \bm{M}^{\square^N}=
    \begin{pmatrix}
        \bm\lambda_{1}^\square & \dots & \bm\lambda_{2N}^\square
    \end{pmatrix}=\ell\bm\Omega,
\end{equation}
such that $\mathcal{L}_{\square^N}\simeq\ell\mathbb{Z}^{2N}$. This represents a canonical hypercube in $2N$ dimensions, which is a (symplectic) self-dual lattice.

\subsection{Self-dual GKP states}

One can generate other all other (symplectic) self-dual lattices via symplectic transformations on the canonical hypercube (a consequence of Corollary 1 in Ref.~\cite{eisert2022lattice}). The corresponding lattice states are `symplectically equivalent' to $N$ canonical GKP states but can nonetheless be useful resources for practical tasks.

Consider a symplectic transformation $\bm\Lambda\in{\rm Sp}(2N,\mathbb{R})$ with unitary representation $U_{\bm\Lambda}$ and define new lattice vectors $\bm\lambda^{\Lambda}_J\equiv\bm\Lambda\bm\lambda^\square_{J}$. Since the symplectic inner product is invariant under symplectic transformations, it follows that $\omega(\bm\lambda^{\Lambda}_J,\bm\lambda^{\Lambda}_K)=\omega(\bm\lambda^\square_{J},\bm\lambda^\square_{K})$, and thus one can naturally define a symplectically integral lattice $\mathcal{L}_\Lambda$ via $\{\bm\lambda^{\Lambda}_J\}$. Furthermore, we can associate these vectors with the lattice stabilizers ${S}_J^{\Lambda}\equiv{D}_{\bm\lambda^{\Lambda}_J}$, which generate the stabilizer group $\mathcal{S}_{\mathcal{L}_\Lambda}=\langle{S}_1^{\Lambda}, {S}_2^{\Lambda},\dots,{S}_{2N}^{\Lambda}\rangle$. By relation~\eqref{eq:weyl_S}, the stabilizers for $\mathcal{L}_\Lambda$ are related to canonical stabilizers via conjugation,
\begin{equation}
    {S}_J^{\Lambda} = U_{\bm\Lambda}  {S}_J^{\square}U_{\bm\Lambda}^\dagger.
\end{equation}
It follows immediately that the state $\ket{\mathcal{L}_\Lambda}\equiv U_{\bm\Lambda}\ket{\square^N}$ is a $+1$ eigenstate of any element in $\mathcal{S}_{\mathcal{L}_\Lambda}$, i.e. $\mathcal{S}_{\mathcal{L}_\Lambda}\ket{\mathcal{L}_\Lambda}=\ket{\mathcal{L}_\Lambda}$. 

\begin{figure}
    \centering
    \includegraphics[width=.6\linewidth]{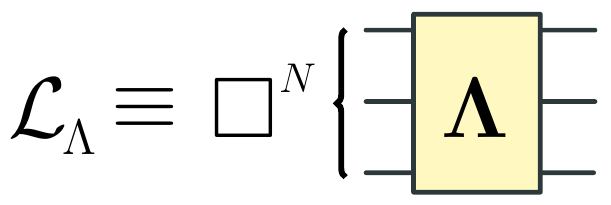}
    \caption{Self-dual GKP lattice state $\mathcal{L}_{\Lambda}$ generated from $N$ canonical square GKP by a symplectic transformation $\bm\Lambda$.}
    \label{fig:selfdual_lattice}
\end{figure}

Consider a generator matrix, $\bm M^\Lambda$, built from the lattice vectors, 
\begin{equation}
\bm M^\Lambda\equiv
    \begin{pmatrix}
        \bm\lambda_1^{\Lambda} & \bm\lambda_2^{\Lambda} & \dots & \bm\lambda_{2N}^{\Lambda}
    \end{pmatrix}.
\end{equation}
A generator matrix and the symplectic transformation $\bm \Lambda$ are related via
\begin{equation}\label{eq:M_lattice}
    \bm M^\Lambda=\ell\bm\Lambda\bm\Omega=\bm\Lambda\bm M^{\square^N},
\end{equation}
i.e., a generator matrix for the new lattice $\bm M^\Lambda$ can be found by transforming a generator matrix of the canonical lattice $\bm{M}^{\square^N}$ by a symplectic transformation $\bm\Lambda$. It immediately follows that $\bm M^{\Lambda\,\top}\bm\Omega\bm M^\Lambda=2\pi\bm\Omega$.

Using Eq.~\eqref{eq:g_Mrelation}, one can show that the modular quadratures $\hat{\bm g}^{\Lambda}$ which generate discrete translations on the lattice are explicitly given as
\begin{equation}
    \hat{\bm g}^{\Lambda}=\ell\bm\Lambda^{-1}\hat{\bm r}.
\end{equation}
Since $\bm\Lambda\in{\rm Sp}(2N,\mathbb{R})$, $\hat{\bm g}^{\Lambda}$ defines a good set of canonical operators which obey the canonical commutation relations up to a factor $\ell^2=2\pi$.

We provide some examples of self-dual lattices below.

\paragraph*{Example: Hexagonal GKP.}  An example of a single-mode ($N=1$) symplectic lattice built from the canonical lattice is the hexagonal lattice $\mathcal{L}_{\hexagon}$, which admits the densest sphere-packing in $\mathbb{R}^2$. Consider the following symplectic transformation,
\begin{equation}\label{eq:lambda_hex}
    \bm\Lambda_{\hexagon}\equiv
    \frac{\ell_{\hexagon}}{\ell}\begin{pmatrix}
        1 & -\frac{1}{2} \\
        0 & \frac{\sqrt{3}}{2}
    \end{pmatrix},
\end{equation}
where $\ell_{\hexagon}/\ell=\sqrt{2}/3^{1/4}\approx1.07$. The hexagonal lattice vectors are defined from the square lattice vectors via $\bm\lambda^{\hexagon}_K=\bm\Lambda_{\hexagon}\bm\lambda^{\square}_K$ and written explicitly as
\begin{align}
    \bm\lambda^{\hexagon}_1&\equiv\ell_{\hexagon}\left(1/2,-\sqrt{3}/{2}\right)^\top,\\
    \bm\lambda^{\hexagon}_2&\equiv\ell_{\hexagon}\left(1,0\right)^\top.
\end{align}
Note that $\norm{\bm\lambda^{\hexagon}_k}=\norm{\bm\lambda^{\hexagon}_2-\bm\lambda^{\hexagon}_1}=\ell_{\hexagon}\approx1.07$, where $\norm{\cdot}$ is the Euclidean norm; this is also the minimal distance between lattice points in $\mathcal{L}_{\hexagon}$. [As an aside, the Euclidean norm of a vector $\bm\mu$ can be defined in symplectic geometry via the symplectic inner product $\omega(\cdot,\cdot)$ without needing to refer to Euclidean geometry. Indeed, consider the dual vector $\tilde{\bm\mu}\equiv\bm\Omega\bm\mu$, then $\omega(\tilde{\bm\mu},\bm\mu)=\norm{\bm\mu}$.]

Furthermore, the hexagonal stabilizer group $\mathcal{S}_{\hexagon}$ can be generated from the following stabilizers,
\begin{align}
    {S}^{\hexagon}_{1}&\equiv\exp({\rm i}{\bm\lambda}^{\hexagon\,\top}_1\bm\Omega_1\hat{{\bm r}})=\exp[{\rm i}\ell_{\hexagon}\left(\frac{\sqrt{3}}{2}\hat{q}+\frac{1}{2}\hat{p}\right)],\\
    {S}^{\hexagon}_{2}&\equiv\exp({\rm i}{\bm\lambda}^{\hexagon\,\top}_2\bm\Omega_1\hat{{\bm r}})=\exp[{\rm i}\ell_{\hexagon}\hat{p}],
\end{align}
which can also be found via conjugation of the square stabilizers via ${S}^{\hexagon}_{K}=U_{{\hexagon}}{S}^{\square}_{K}U_{{\hexagon}}^\dagger$, where $U_{{\hexagon}}$ is a unitary representation of $\bm\Lambda_{\hexagon}$. As in the general case, the hexagonal GKP state is formally given by $\ket{\hexagon}=U_{{\hexagon}}\ket{\square}$, which is necessarily a $+1$ eigenstate of the stabilizer group elements, i.e. $\mathcal{S}_{\hexagon}\ket{\hexagon}=\ket{\hexagon}$. 

\paragraph*{Example: Rectangular GKP.}
Consider the symplectic squeezing matrix with squeezing strength $\eta\equiv{\rm e}^r$,
\begin{equation}
    \textbf{Sq}(\eta)=
    \begin{pmatrix}
        \eta & 0 \\
        0 & 1/\eta
    \end{pmatrix}.
\end{equation}
We define the new lattice vectors $\bm\lambda^{\eta}_K\equiv\textbf{Sq}(\eta)\bm\lambda^{\square}_K$ from which we construct the stabilizers,
\begin{align}
    {S}^{\,\eta}_{1}&\equiv\exp({\rm i}{\bm\lambda}^{\eta\,\top}_1\bm\Omega_1\hat{{\bm r}})={\rm e}^{{\rm i}\hat{q}\ell/\eta},\\
    {S}^{\,\eta}_{2}&\equiv\exp({\rm i}{\bm\lambda}^{\eta\,\top}_2\bm\Omega_1\hat{{\bm r}})={\rm e}^{{\rm i}\hat{p}\ell\eta}.
\end{align}
The state $\ket{\eta}\equiv U_{\textbf{Sq}(\eta)}\ket{\square}$ is a $+1$ eigenstate of these stabilizers, and defines a rectangular GKP state---i.e., squeezed along one quadrature and stretched along the other. This state can be written in the position and momentum bases as
\begin{equation}
    \ket{\eta}\propto\sum_{k\in\mathbb{Z}}\ket{q=k\ell\eta}=\sum_{k\in\mathbb{Z}}\ket{p=k\ell/\eta}.
\end{equation}
The rectangular GKP state has proven useful for bias-enhanced QEC with the GKP surface code~\cite{hanggli2020pra}.

\paragraph*{Example: GKP Bell state.}
An interesting class of two-mode lattice states are entangled GKP states generated via two-mode interactions. A prominent example is a GKP Bell state $\Phi$ often used in DV quantum information processing with bosonic qubits~\cite{walshe2020CVteleport,vanLoock2022gkp}. A GKP Bell state can be formed by interacting two canonical GKP states on a 50:50 beamsplitter~\cite{walshe2020CVteleport}, i.e.
\begin{equation}
    \ket{\Phi}\equiv U_{\bm B_{1/2}}\ket{\square^2}.
\end{equation}
The lattice vectors $\{\bm\lambda_J^{\Phi}\}$ are given by the column vectors of the generator matrix $\bm M^{\Phi}$, which can be directly found via Eq.~\eqref{eq:M_lattice}; explicitly,
\begin{equation}
    \bm M^{\Phi}=\ell\bm B_{1/2}\bm\Omega=\frac{\ell}{\sqrt{2}}
    \begin{pmatrix}
        \bm\Omega_1 & \bm\Omega_1 \\
        -\bm\Omega_1 & \bm\Omega_1
    \end{pmatrix},
\end{equation}
where $\bm\Omega_1$ is the symplectic form on $\mathbb{R}^2$. The GKP Bell state has the same lattice spacing as the canonical GKP state $\square^2$, since the former is simply a rotated version of the latter in $\mathbb{R}^4$ and are thus not useful for QEC. On the other hand, GKP Bell states are useful in other quantum information processing tasks~\cite{walshe2020CVteleport,vanLoock2022gkp}, such as quantum teleportation.

\paragraph*{Example: $D_4$ GKP.}
The $D_4$ lattice admits a densest sphere packing in $\mathbb{R}^4$~\cite{baptiste2022multiGKP}. One can generate a $D_4$ lattice from a direct product of two canonical GKP states by the following symplectic transformation,
\begin{equation}\label{eq:Sd4}
    \bm\Lambda_{D_4}^{\top}=\sqrt[4]{2}
    \begin{pmatrix}
        \frac{1}{2} & -\frac{1}{\sqrt{2}} & \frac{1}{2} & 0\\
        0 & \frac{1}{\sqrt{2}} & 0 & \frac{1}{\sqrt{2}}\\
        0 & \frac{1}{\sqrt{2}} & 0 & -\frac{1}{\sqrt{2}} \\
        -\frac{1}{2} & 0 & \frac{1}{2} & \frac{1}{\sqrt{2}}
    \end{pmatrix},
\end{equation}
such that $\ket{D_4}\equiv U_{\bm\Lambda_{D_4}}\ket{\square^2}$. Since $\bm\Lambda_{D_4}\bm\Lambda_{D_4}^{\top}\neq\bm I$, there is squeezing involved.


\section{Syndrome measurements}\label{app:measurements}

We discuss how to extract error syndromes from a general GKP lattice.

\subsection{Error syndrome}
Consider a lattice state $\mathcal{L}$ and a displacement error $\bm e$ such that
\begin{equation}
    \ket{\mathcal{L}; \bm e}\equiv{D}_{\bm e}\ket{\mathcal{L}}.
\end{equation}
Then observe the following set of equalities,
\begin{equation*}
\begin{split}
    {S}^{\mathcal{L}}_{K}\ket{\mathcal{L};\bm e}&=
    {S}^{\mathcal{L}}_{K}{D}_{\bm e}\ket{\mathcal{L}}\\
    &={\rm e}^{-{\rm i}\omega\left(\bm\lambda_K^{\mathcal{L}},\bm e\right)}{D}_{\bm e} {S}^{\,\mathcal{L}}_{K}\ket{\mathcal{L}}\\
    &={\rm e}^{-{\rm i}\omega\left(\bm\lambda_K^{\mathcal{L}},\bm e\right)}{D}_{\bm e}\ket{\mathcal{L}}\\
    &={\rm e}^{-{\rm i}\omega\left(\bm\lambda_K^{\mathcal{L}},\bm e\right)}\ket{\mathcal{L};\bm e},
\end{split}
\end{equation*}
where $\bm\lambda_K^{\mathcal{L}}$ are the lattice vectors that satisfy $\omega(\bm\lambda_K^{\mathcal{L}},\bm\lambda_J^{\mathcal{L}})=2\pi n_{KJ}$ with $n_{KJ}\in\mathbb{Z}$, and we have used the composition rule~\eqref{eq:weyl_comp} to go from the first equality to the second equality. In other words, the error state $\ket{\mathcal{L};\bm e}$ is an eigenstate of the stabilizer ${S}^{\mathcal{L}}_{K}$ with eigenvalue $\exp[-{\rm i}\omega\left(\bm\lambda_K^{\mathcal{L}},\bm e\right)]$. 

The set of quantities $\{\omega(\bm\lambda_K^{\mathcal{L}},\bm e)\mod{2\pi}\}_{K=1}^{2N}$ are the error syndromes extracted from stabilizer measurements on the GKP lattice $\mathcal{L}$. We can package them nicely into an \textit{error syndrome} vector $\bm s$ with the help of a generator matrix $\bm M$,
\begin{equation}\label{eq:syndrome}
    \bm s \equiv \bm s(\bm e)=\ell^{-1}\bm M^{\top}\bm\Omega\bm e\mod\sqrt{2\pi},
\end{equation}
where the modulo operation acts elementwise and we have factored out the canonical spacing $\ell$. If $\mathcal{L}$ is a symplectic self-dual lattice such that $\bm M=\ell\bm\Lambda\bm \Omega$, where $\bm \Lambda\in{\rm Sp}(2N,\mathbb{R})$ [see Eq.~\eqref{eq:M_lattice}], then we may write the syndrome vector as
\begin{equation}
    \bm s = \bm\Lambda^{-1}\bm e \mod\sqrt{2\pi}.
\end{equation}

\begin{figure}
    \centering
    \includegraphics[width=\linewidth]{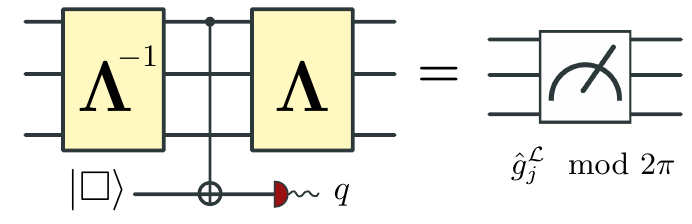}
    \caption{GKP-assisted stabilizer measurement circuit for ${S}_J^{\mathcal{L}}$. Equivalent to measuring the modular quadrature ${\hat{g}^{\mathcal{L}}_J\mod2\pi}$.}
    \label{fig:measurement}
\end{figure}

\subsection{Homodyne and ancilla-assisted measurements}
One way to measure displacements on a lattice state is to perform homodyne measurements along the modular quadratures ($\hat{\bm g}=\bm M^\top\bm\Omega\bm{\hat{r}}\mod{2\pi}$). However, homodyne measurements are destructive, hence an ancilla-assisted measurement scheme is warranted. 

For lattices that can be generated from the canonical lattice via symplectic transformations (i.e., self-dual lattices), we introduce a measurement circuit that is related to the canonical measurement circuit (consisting of standard SUM gates and square GKP~\cite{gkp2000}) via unitary conjugation; see Fig.~\ref{fig:measurement} for an illustration of the stabilizer measure circuit. Our measurement scheme relies on the SUM gate, however we remark that SUM-gate measurement strategies necessarily require inline squeezing, as the SUM-gate is not an orthogonal transformation. An ancilla-assisted measurement strategy which moves squeezing offline has been proposed in Ref.~\cite{vanLoock2022gkp}.

There are $2N$ stabilizers that need to be measured, two for each mode of the GKP lattice $\mathcal{L}$. However, if the GKP lattice $\mathcal{L}$ is an ancillary state in a QEC circuit and can thus be discarded after use, then only $N$ additional GKP ancillae (on top of the $N$ mode GKP lattice $\mathcal{L}$) are required for measurements. This reduces the GKP resources by half. In more detail, we can first perform $N$ nondestructive stabilizer (e.g., $q$ quadrature) measurements via ancilla-assisted measurements, which consumes $N$ GKP measurement ancillae. Following these nondestructive measurements, we subsequently perform $N$ \textit{destructive} homodyne (e.g., $p$ quadrature) measurements on the GKP lattice $\mathcal{L}$. Moreover, this measurement strategy can be performed in parallel on the individual modes.



\section{Symmetry of single-mode lattices}
\label{app:Symmetry_single_mode}
We can have the same lattice $\mathcal{L}$ with different choices of bases. In particular, the generator matrices $\bm M_1$ and $\bm M_2$ generate the same lattice if there exist a unimodular matrix $\bm N$ (i.e. a matrix with integer entries and $\det\bm N=1$) such that
\begin{equation}
    \bm M_1^\top = \bm N \bm M_2^\top.
    \label{eq:lattice-basis-changes}
\end{equation}
Suppose that a single-mode symplectic self-dual lattice is given by (we drop the $\ell=\sqrt{2\pi}$ factor in front of $\bm M$ here for brevity)
\begin{equation}
     \bm M =\bm R(\phi)\text{\bf Sq}(r)\bm R(\theta),
\end{equation}
where $\bm R(\theta)$ and $\bm R(\phi)$ are $2\times2$ rotation matrices and $\text{\bf Sq}(r)$ is the single-mode squeezing. Fixing $r$, it is easy to see that two bases 
\begin{align*}
    &\bm M_1=\text{\bf Sq}(r)\bm R(\theta),\\
    &\bm M_2=\text{\bf Sq}(r)\bm R(\theta+\pi/2),
\end{align*}
are the same lattice since $\bm M_2^\top=\bm R\left(\frac{\pi}{2}\right)\bm M_1^\top$ and $\bm R\left(\frac{\pi}{2}\right)$ is a unimodular matrix.
Furthermore, it is easy to show that 
\begin{align*}
    &\bm M_1=\text{\bf Sq}(r)\bm R(\theta),\\
    &\bm M_2=\text{\bf Sq}(r)\bm R(-\theta),
\end{align*}
are the same lattice
under a reflection about the $x$-axis. Combining the above symmetries, we conclude that two bases
\begin{align*}
    &\bm M_1=\text{\bf Sq}(r)\bm R(\theta),\\
    &\bm M_2=\text{\bf Sq}(r)\bm R(\pi/2-\theta),
\end{align*}
correspond to the same lattice.

\subsection{Hexagonal lattice}

A generator matrix for the hexagonal lattice can be written as
\begin{equation}
    \bm M_{\hexagon}=\frac{\ell_{\hexagon}}{\ell}
    \begin{pmatrix}
    1 & -\frac{1}{2}\\
    0 & \frac{\sqrt{3}}{2}
    \end{pmatrix},
\end{equation}
where $\ell_{\hexagon}/\ell=\sqrt{2}/3^{1/4}\approx1.07$; see Eq.~\eqref{eq:lambda_hex}. By the Bloch-Messiah decomposition,
\begin{equation}
    \bm M_{\hexagon}=\bm R\left(-\frac{\pi}{6}\right)\text{\bf Sq}(3^{1/4})\bm R\left(\frac{\pi}{4}\right),
\end{equation}
where $\text{\bf Sq}(3^{1/4})=\text{diag}(3^{1/4},3^{-1/4})$. Due to rotation symmetry, we have the same lattice if
\begin{align}\label{eq:mhex_unimod}
    \bm M_{\hexagon}^\prime & \equiv \text{\bf Sq}(r)\bm R(\theta)\nonumber\\
    & =  \bm R(\phi) \bm M_{\hexagon}\bm N^\top,
\end{align}
where $\bm N$ is a unimodular matrix. For $\bm N\neq\bm I$, the squeezing value $r\geq3^{1/4}$; see Table~\ref{tab:hexagonal} of the main text.

\section{Proof of Lemma~\ref{lemma:lemma_redux}}
\label{app:proof_lemma_redux}

We prove Lemma~\ref{lemma:lemma_redux} of the main text. The primary techniques that we use are Gaussian channel synthesis and the modewise entanglement theorem~\cite{reznik2003modewise} (Theorem~\ref{thm:synthesis} and Theorem~\ref{thm:modewise} of Appendix~\ref{app:gauss_evol}, respectively). Below, we sometimes write $\mathcal{N}_\sigma$ for the $K$ mode iid AGN channel, as opposed to $\bigotimes_{i=1}^K\mathcal{N}_{\sigma}$, for the sake of brevity.

Denoting $\bm S=\bm S_{\mathsf{enc}}^{-1}$ for simplicity, it is sufficient to show that,
\begin{multline}\label{eq:StoTMS}
    \mathcal{U}_{\bm S}\circ\mathcal{N}_\sigma\circ\mathcal{U}_{\bm S}^{-1}\simeq \\ \left(\mathcal{U}_{\bigoplus\bm S_{G_i}}\otimes{\rm id}_{M-N}\right)\circ\mathcal{N}_\sigma\circ\left(\mathcal{U}_{\bigoplus\bm S_{G_i}}^{-1}\otimes{\rm id}_{M-N}\right),
\end{multline}
where $\bm S\in{\rm Sp}(2K,\mathbb{R})$ is an arbitrary symplectic matrix for $K=M+N$ modes, `$\simeq$' means `equivalent up to local unitaries', ${\rm id}_{M-N}$ is an identity superoperator on $M-N$ modes, and $\bigoplus\bm S_{G_i}$ is shorthand for $\bigoplus_{i=1}^N\bm S_{G_i}$. The channel $\mathcal{G}_1\equiv\mathcal{U}_{\bm S}\circ\mathcal{N}_\sigma\circ\mathcal{U}_{\bm S}^{-1}$ is a Gaussian channel characterized by a displacement-noise vector $\bm d_1=\bm 0$, a scaling matrix $\bm X_1=\bm I$, and a noise matrix $\bm Y_1=\sigma^2\bm S\bm S^\top$ (see Appendix~\ref{app:gauss_evol} for details about characterizing general Gaussian channels). By the modewise entanglement theorem~\cite{reznik2003modewise}, there exists local symplectic matrices $\bm\Lambda_d\in{\rm Sp}(2N,\mathbb{R})$ and $\bm\Lambda_a\in{\rm Sp}(2M,\mathbb{R})$ such that,
\begin{multline}\label{eq:YtoTMS}
    \left(\bm\Lambda_d\oplus\bm\Lambda_a\right)\bm Y_1\left(\bm\Lambda_d^\top\oplus\bm\Lambda_a^\top\right)=\\\sigma^2\left(\bigoplus_{i=1}^{N}\bm S_{G_i}\bm S_{G_i}^{\top}\right)\oplus\bm I_{2(M-N)},
\end{multline}
where $\bm S_{G_i}$ is a TMS operation between the $i$th data mode and the $(N+i)$th ancilla mode. The transformation $\bm\Lambda_d\oplus\bm\Lambda_a$ corresponds to a direct product of local, unitary Gaussian channels $\mathcal{U}_{\bm\Lambda_d\oplus\bm\Lambda_a}=\mathcal{U}_{\bm\Lambda_d}\otimes\mathcal{U}_{\bm\Lambda_a}$. We can pre- and post-process with $\mathcal{U}_{\bm\Lambda_d\oplus\bm\Lambda_a}$ and $\mathcal{U}_{\bm\Lambda_d\oplus\bm\Lambda_a}^{-1}$ to generate a new Gaussian channel,
\begin{align}
    \mathcal{G}_2&=\mathcal{U}_{\bm\Lambda_d\oplus\bm\Lambda_a}\circ\mathcal{G}_1\circ\mathcal{U}_{\bm\Lambda_d\oplus\bm\Lambda_a}^{-1}\nonumber\\
    &=\mathcal{U}_{\bm\Lambda_d\oplus\bm\Lambda_a}\circ\left(\mathcal{U}_{\bm S}\circ\mathcal{N}_\sigma\circ\mathcal{U}_{\bm S}^{-1}\right)\circ\mathcal{U}_{\bm\Lambda_d\oplus\bm\Lambda_a}^{-1}.
\end{align}
Then, by Gaussian channel synthesis, starting from $\bm d_1$, $\bm X_1$ and $\bm Y_1$, $\mathcal{G}_2$ has a characterization $\bm d_2=\bm 0$, $\bm X_2=\bm I$, and
\begin{align}
    \bm Y_2&=\left(\bm\Lambda_d\oplus\bm\Lambda_a\right)\bm Y_1\left(\bm\Lambda_d^\top\oplus\bm\Lambda_a^\top\right)
    \nonumber\\
    &=\sigma^2\left(\bigoplus_{i=1}^{N}\bm S_{G_i}\bm S_{G_i}^{\top}\right)\oplus\bm I_{2(M-N)},
\end{align}
where Eq.~\eqref{eq:YtoTMS} was used for the second equality. Defining the unitary Gaussian channel for the (modewise) TMS operations $\mathcal{U}_{\bigoplus\bm S_{G_i}}\equiv\bigotimes_{i=1}^N\mathcal{U}_{\bm S_{G_i}}$, it follows that,
\begin{equation}
    \mathcal{G}_2=\left(\mathcal{U}_{\bigoplus\bm S_{G_i}}\otimes{\rm id}_{M-N}\right)\circ\mathcal{N}_{\sigma}\circ\left(\mathcal{U}_{\bigoplus\bm S_{G_i}}^{-1}\otimes{\rm id}_{M-N}\right).
\end{equation}
Therefore, up to pre- and post-processing with local Gaussian unitaries, we have that $\mathcal{G}_2\simeq\mathcal{G}_1$, which was to be proved.

\begin{widetext}

\section{Error estimation for general GKP lattices}
\label{app:error_estimation}

Here we derive properties of error estimation---such as the joint PDF $P(\bm x_d,\bm s)$ after a QEC protocol and the minimum mean square estimation (MMSE) $\tilde{\bx}_d = \bm f({\bm s})$ for error syndrome ${\bm s}$---given an encoding $\bm S_{\mathsf{enc}}$ and ancilla GKP state with generator matrix $\bm M$.  
To simplify our derivations, we denote the probability density function (PDF) of the $2n$-dimensional multivariate Gaussian distribution as
\begin{equation}
    g(\bm{\Sigma},\bx - \bm \mu)=\frac{1}{(2\pi)^n\sqrt{\det\bm\Sigma}}\exp{-\frac{1}{2}(\bx-\bm \mu)^\top \bm \Sigma^{-1}(\bx - \bm \mu)},
    \label{g_definition}
\end{equation}
where $\bx=(x_1,..,x_{2n})^\top$, $\bm{\Sigma}$ is the covariance matrix and $\bm \mu$ is the mean. The dimension of the distribution is implicitly given by the dimensions of $\bm{\Sigma}$ and $\bm \mu$.

\subsection{Proof of Theorem~\ref{thm:fMMSE}}
In the QEC protocol, the original noise covariance matrix $\bm V_\xi=\bigoplus_{i=1}^{N+M}\sigma^2_i\bm I_2$ is transformed into  
\begin{align}
    & \bm V_{\bx}^{-1}=\bm S_{\mathsf{enc}}^{\top}\bm V_\xi^{-1} \bm S_{\mathsf{enc}}, 
    \label{eq:AppF-noise-CM-transform}
\end{align}
via the encoding symplectic transform $\bm S_{\mathsf{enc}}$ and decoding symplectic transform $\bm S_{\mathsf{enc}}^{-1}$. Let the additive noise on the data and ancilla be
\begin{align}
    &\bx_d = (x^{(q)}_1,x^{(p)}_1,...,x^{(q)}_N,x^{(p)}_N)^{\top},\\
    & \bx_a = (x^{(q)}_{N+1},x^{(p)}_{N+1},...,x^{(q)}_{N+M},x^{(p)}_{N+M})^{\top},
\end{align}
which are random variables following the joint distribution
\begin{align}
    P(\bm x_d, \bm x_a) = g\left[\bm V_{\bx}, \left(\bx_d,\bx_a\right)\right].\label{eq:joint_dist}
\end{align}
We define the interval $\mathfrak{I}\equiv[-\sqrt{\pi/2},\sqrt{\pi/2}]$, and the error syndrome 
\begin{align} 
{\bm s} = R_{\sqrt{2\pi}} (\bm M^\top \bm \Omega \bx_a)\in \mathfrak{I}^{2M}.
\end{align}
To get the joint distribution of $\bx_d$ and $\bm s$,
we first rewrite Eq.~\eqref{eq:joint_dist} as
\begin{equation}
    P(\bx_d,\bm M^\top \bm \Omega \bx_a)=
    g\left[
    \begin{pmatrix}
    \bm V_d & \bm V_{da}\\
    \bm V_{da}^T & \bm V_a
    \end{pmatrix}^{-1},
    \left(\bx_d,\bm M^\top \bm \Omega \bx_a\right)\right],
\end{equation}
where we have defined the covariance matrix in the block form
\begin{equation}
    \begin{pmatrix}
    \bm V_d & \bm V_{da}\\
    \bm V_{da}^T & \bm V_a
    \end{pmatrix}^{-1}\equiv 
    (\bI_{2N} \oplus {\bm M^\top \bm \Omega}) \bm V_{\bx} (\bI_{2N} \oplus {(\bm M^\top \bm \Omega)} ^{\top})
    .
    \label{eq:AppF-noise-CM-block-form}    
\end{equation}
and used the property $g(\bm\Sigma, \bm x)= g(\bm L\bm\Sigma\bm L^{\top},\bm L\bm x)$ for some invertible matrix $\bm L$. From here, the distribution of the error and syndrome can be solved as
\begin{align}
    P(\bx_d,\bm s) & = \int_{\mathbb{R}^{2M}} \dd\left(\bm M^\top \bm \Omega \bx_a\right) \;\; P(\bx_d,\bm M^\top \bm \Omega \bx_a) \sum_{\bm k}\delta\left(\bm s - \bm M^\top \bm \Omega \bx_a -\bm k \sqrt{2\pi} \right) 
    \label{eq:xd-s-joint-distribution–step1}
    \\
     & = \sum_{\bm k}  g\left((\bI_{2N} \oplus {\bm M^\top \bm \Omega}) \bm V_{\bx} (\bI_{2N} \oplus {(\bm M^\top \bm \Omega)} ^{\top}),(\bx_d,\bm s - \bm k \sqrt{2\pi})\right)
     \label{eq:xd-s-joint-distribution-step2}
     \\
    & =  \sum_{\bm k} g(\bm V_d^{-1},\bm x_d+ \bm V^{-1}_d \bm V_{da}(\bm s-\bm k\sqrt{2\pi})) g(\bm V_{d|a}^{-1}, \bm s -\bm k\sqrt{2\pi}).
    \label{eq:xd-s-joint-distribution}
\end{align}
where we sum over all vector of integers $\bm k\in \mathbb{Z}^{2M}$, $\delta(\cdot)$ is Dirac delta distribution and $\bm V_{d|a}=\bm V_a-\bm V_{da}^T \bm V_d^{-1}\bm V_{da}$. From \eqref{eq:xd-s-joint-distribution–step1} to \eqref{eq:xd-s-joint-distribution-step2}, we  integrate over ${\bm x}_a$. From \eqref{eq:xd-s-joint-distribution-step2} to \eqref{eq:xd-s-joint-distribution}, we adopt the block form of Eq.~\eqref{eq:AppF-noise-CM-block-form} and separate the joint distribution into two parts.

The PDF of the syndrome measurement result can be obtained by integrating over ${\bm x}_d$,
\begin{align}
    & P(\bm s) =\int_{\mathbb{R}^{2N}}\dd{\bm x_d}P(\bm x_d,\bm s)= \sum_{\bm k} g(\bm V_{d|a}^{-1}, \bm s -\bm k\sqrt{2\pi}).
    \label{eq:PDF-s}
\end{align}
The PDF of the conditional distribution is therefore
\begin{align}
    P(\bm x_d|\bm s)  & = P(\bm x_d, \bm s)/ P(\bm s)\nonumber
    \\
    & =  \sum_{\bm n} g\left(\bm V_d^{-1},\bm x_d+ \bm V^{-1}_d \bm V_{da}\left(\bm s-\bm n\sqrt{2\pi}\right)\right) g\left(\bm V_{d|a}^{-1}, \bm s -\bm n\sqrt{2\pi}\right)/ \left(\sum_{\bm m}g(\bm V_{d|a}^{-1}, \bm s -\bm k\sqrt{2\pi}) \right).
\end{align}
The MMSE estimator is obtained by evaluating the mean of the conditional distribution
\begin{align} 
    \tilde{\bx}_d\left(\bm s\right)\equiv \bm f_{\rm MMSE}(\bm s) & = \int_{\mathbb{R}^{2N}}\dd{\bm x_d}\;\; \bm x_d P(\bm x_d|\bm s) \nonumber\\
    & = -\sum_{\bm n} \bm V_d^{-1}\bm V_{da}\left(\bm s-\bm n\sqrt{2\pi}\right) g\left(\bm V_{d|a}^{-1},\bm s -\bm n \sqrt{2\pi}\right) / \left(\sum_{\bm m} g(\bm V_{d|a}^{-1},\bm s -\bm m \sqrt{2\pi}) \right) \nonumber\\
    & = \frac{-\sum_{\bm n} \bm V_d^{-1}\bm V_{da}\left(\bm s-\bm n\sqrt{2\pi}\right)\exp{\sqrt{2\pi} \bm n^{\top} \bm V_{d|a}\bm s-\pi \bm n^{\top} \bm  V_{d|a}\bm n} }{  \sum_{\bm m} \exp{\sqrt{2\pi} {\bm m}^{\top} \bm V_{d|a}\bm s-\pi {\bm m}^{\top} \bm V_{d|a}\bm m}
    },
    \label{eq:App-MMSE-estimator}
\end{align}
where $\bm n,\bm m\in\mathbb{Z}^{2M}$. We have thus proven Eq.~\eqref{eq:fMMSE_main} of Theorem~\ref{thm:fMMSE}.

\subsection{Output covariance matrix}
Given an error syndrome $\bm s$ extracted from a general GKP ancilla in a QEC protocol, the element at the $i$th row and $j$th column of the output covariance matrix $\bm V_{\rm{out}}$ for $N$ data modes is given as
\begin{align}
    [\bm V_{\rm out}]_{ij} &= \int_{\mathbb{R}^{2N}}\dd{\bm x_d} \int_{\mathfrak{I}^{2M}}\dd{\bm s}\;\; (\bx_d-\tilde{\bx}_d)_{i \times j} P(\bm x_d,\bm  s) \nonumber\\
    & = \sum_{\bm k} \int_{\mathfrak{I}^{2M}}\dd{\bm s} \int_{\mathbb{R}^{2N}}\dd{\bx_d}\; {(\bx_d-\tilde{\bx}_d)}_{i \times j}g(\bm V_d^{-1},\bm x_d+ \bm V^{-1}_d \bm V_{da}(\bm s-\bm k\sqrt{2\pi})) g(\bm V_{d|a}^{-1}, \bm s -\bm k\sqrt{2\pi}) \nonumber\\
    &= \sum_{\bm k} \int_{\mathfrak{I}^{2M}}\dd{\bm s}\;\left( [\bm V_d^{-1}]_{ij} g(\bm V_{d|a}^{-1}, \bm s -\bm k\sqrt{2\pi}) +  g(\bm V_{d|a}^{-1}, \bm s -\bm k\sqrt{2\pi}) \left\{\tilde{\bx}_d + \left[\bm V_d^{-1}\bm V_{da}(\bm s-\bm k\sqrt{2\pi})\right]\right\}_{i\times j}\right) \nonumber\\
    & = [\bm V_d^{-1}]_{ij} + \sum_{\bm k}\int_{\mathfrak{I}^{2M}}\diff{\bm  s}\;\;g(\bm V_{d|a}^{-1},\bm s - \bm k \sqrt{2\pi}) \left\{\tilde{\bx}_d + \left[\bm V_d^{-1}\bm V_{da}(\bm s-\bm k\sqrt{2\pi})\right]\right\}_{i\times j} \nonumber\\
    & =  [\bm V_d^{-1}]_{ij} + \sum_{\bm k}2\pi \int_{\mathfrak{I}^{2M}}\diff{\bm  s}\;\;g(\bm V_{d|a}^{-1},\bm s - \bm k \sqrt{2\pi})\left[\frac{ \sum_{\bm n} \bm V_d^{-1} \bm V_{da} (\bm n -\bm k)g(\bm V_{d|a}^{-1},\bm s - \bm n \sqrt{2\pi})}{\sum_{\bm m} g(\bm V_{d|a}^{-1},\bm s - \bm m \sqrt{2\pi})}\right]_{i \times j}.
    \label{eq:Vout_ij}
\end{align}
where we use the notation $\bm x _{i\times j} = \bm x_i \bm x_j$, as the product vector components $i$ and $j$ and have expanded and simplified by using the fact that 
\begin{equation}
    \int_{\mathbb{R}^{2M}}\dd{\bx_d} \; (\bm x_i-\bm\mu_i)(\bm x_j-\bm\mu_j)\frac{\det(\bm V_{d})^\frac{1}{2}}{(2\pi)^N }\exp{-\frac{1}{2}(\bm x-\bm \tau)^T \bm V_{d}(\bm x-\bm \tau)} = {[\bm V_d^{-1}]}_{ij}+(\bm\mu_i-\bm\tau_i)(\bm\mu_j-\bm\tau_j).
    \label{eq:app_integral_simplification}
\end{equation}

\subsection{\QZ{Explicit analyses of the decoding strategies for GKP-TMS code}}
\label{app:explicit_applying_theorems_estimator}

\QZ{
Here we show explicit calculation to obtain the estimators for GKP-TMS codes with GKP square lattice.
According to Eq.~\eqref{eq:Noise-CM-transform} and the encoding operation of the GKP-TMS code in Eq.~\eqref{eq:tms_matrix_main}, the random displacement $\boldsymbol{z} = (z_q^{(1)},z_p^{(1)},z_q^{(2)},z_p^{(2)})$ of the output AGN channel now has the covariance matrix
\begin{align}
    \label{covMatrix2}
    \begin{split}
    &\bm V_x = \bm{S}_{\rm{enc}}^{-1}\bm{V}_\xi\bm{S}_{\rm{enc}}^{-\top}=
    \begin{pmatrix}
        [G\sigma_1^2+(G-1)\sigma_2^2]\boldsymbol{I} && -\sqrt{G(G-1)}(\sigma_1^2+\sigma_2^2)\boldsymbol{Z}\\
        -\sqrt{G(G-1)}(\sigma_1^2+\sigma_2^2)\boldsymbol{Z} && [G\sigma_2^2+(G-1)\sigma_1^2]\boldsymbol{I}
    \end{pmatrix}.
    \end{split}
\end{align}
The GKP square lattice has generator matrix $\bm M = \bm I$. From eq.~\eqref{eq:MMSE_submatrices}, we have 
\begin{align}
    &\bm V_d=\frac{(G-1)\sigma_1^2+G \sigma_2^2}{\sigma_1^2 \sigma_2^2} \bI, \quad \bm V_{da}=\frac{\sqrt{G(G-1)}(\sigma_1^2 + \sigma_2^2)}{\sigma_1^2 \sigma_2^2} \boldsymbol{\Omega} \boldsymbol{Z}, \quad \bm V_{d|a}=\frac{1}{(G-1)\sigma_1^2+G \sigma_2^2} \bI.
\end{align}
The linear estimation becomes
\begin{align}
    \bm f_{\rm{Linear}}(\bs) = -\bm V_d^{-1} \bm V_{da} \bs =  \tilde{\mu}\begin{pmatrix}
        0 & 1\\
        1 & 0
    \end{pmatrix} \bs,
\end{align}
where
\begin{align}
    & \tilde{\mu} = \frac{\sqrt{G(G-1)}(\sigma_1^2+\sigma_2^2)}{(G-1)\sigma_1^2+G\sigma_2^2}.
\end{align}
Further calculations on the output variance can be found in Refs.~\cite{noh2020o2o,wu2021continuous}.
}

\QZ{
The MMSE estimator of Eq.~\eqref{eq:fMMSE_main} becomes
\begin{align}
    \bm f_{\rm MMSE}(\bm s)  = \frac{\sum_{\bm n} \tilde{\mu} \begin{pmatrix}
        0 & 1\\
        1 & 0
    \end{pmatrix} (\bm s-\bm n\sqrt{2\pi})g(\sigma_G^2\bI,\bs-\bm n\sqrt{2\pi}) }
    {\sum_{\bm m} g(\sigma_G^2\bI,\bs-\bm m\sqrt{2\pi}) }.
\end{align}
The output covariance matrix of Eq.~\eqref{eq:Vout_ij} leads to
\begin{align}
    [\bm V_{\rm out}]_{ij} = \frac{\sigma_1^2 \sigma_2^2}{(G-1)\sigma_1^2+G \sigma_2^2} [\bI]_{ij}+ \sum_{\bm k}2\pi \tilde{\mu}^2 \int_{\mathfrak{I}^{2}}\diff{\bm  s}\;\;g(\sigma_G^2\bI,\bm s - \bm k \sqrt{2\pi})\left[\frac{ \sum_{\bm n} (\bm n -\bm k)g(\sigma_G^2\bI,\bm s - \bm n \sqrt{2\pi})}{\sum_{\bm m} g(\sigma_G^2\bI,\bm s - \bm m \sqrt{2\pi})}\right]_{i \times j}.
\end{align}
}

\subsection{Lattice basis transformations and estimators}
\label{app:lattice_transform} 

We show that the output joint PDF of the QEC protocol is invariant under a change of lattice basis. We then show that the MMSE estimator of Theorem~\ref{thm:fMMSE} is likewise invariant but that linear estimation depends upon the choice of the lattice basis.

\begin{theorem}[Invariance of the joint PDF]\label{thm:invariant_pdf}
    Consider a generator matrix $\bm M$ and a change of lattice basis by a unimodular matrix $\bm N$, which defines another generator matrix $\bm M^\prime=\bm M\bm N^\top$. Let $\bm s=R_{\sqrt{2\pi}} (\bm M^\top \bm \Omega \bx_a)$ be the error syndrome and consider the joint PDF $P(\bm x_d,\bm s)$. Likewise, let ${\bm s}^\prime = R_{\sqrt{2\pi}} ( \bm M^{\prime\,\top}\bm \Omega \bx_a)=R_{\sqrt{2\pi}} (\bm N \bm M^\top \bm \Omega \bx_a)$ be the error syndrome in the new basis and the corresponding joint PDF $P^\prime(\bx_d,\bm s^\prime)$. Then, $P^\prime(\bx_d,\bm s^\prime)=P(\bm x_d,\bm s)$.
\end{theorem}
\begin{proof}
From properties of modulo operations, we have that $R_{\sqrt{2\pi}} (\bm N^{-1} {\bm s}^\prime)=R_{\sqrt{2\pi}} (\bm N^{-1}\bm N \bm M^\top \bm \Omega \bx_a)=\bm s$, as both $\bm N$ and $\bm N^{-1}$ are a matrices of integers. Therefore $ \bm N^{-1} {\bm s}^\prime=\bm s+\bm \ell_{\bm N^{-1} {\bm s}^\prime} \sqrt{2\pi}$ for some vector of integers $\bm \ell_{\bm N^{-1} {\bm s}^\prime}$ determined by $\bm N^{-1} {\bm s}^\prime$. The joint PDF for the data $\bm x_d$ and the error syndrome $\bm s^\prime$ in the new basis is then
\begin{align}
    P^\prime(\bx_d,\bm s^\prime) & = \sum_{\bm k}  g\left((\bI_{2N} \oplus {(\bm N \bm M^\top \bm \Omega)}) \bm V_{\bx} (\bI_{2N} \oplus {(\bm N \bm M^\top \bm \Omega)} ^{\top}),(\bx_d,\bm s^\prime - \bm k \sqrt{2\pi})\right) 
    \label{eq:app_PDF_xd_s_prime-step1}
    \\
    & = \sum_{\bm k}  g\left((\bI_{2N} \oplus {(\bm M^\top \bm \Omega)}) \bm V_{\bx} (\bI_{2N} \oplus {(\bm M^\top \bm \Omega)} ^{\top}),(\bx_d,\bm N^{-1}(\bm s^\prime - \bm k \sqrt{2\pi}))\right)
    \label{eq:app_PDF_xd_s_prime-step2}
    \\
    & = \sum_{\bm k^{\prime}}  g\left((\bI_{2N} \oplus {(\bm M^\top \bm \Omega)}) \bm V_{\bx} (\bI_{2N} \oplus {(\bm M^\top \bm \Omega)} ^{\top}),(\bx_d,\bm N^{-1} \bm s^\prime - \bm k^{\prime} \sqrt{2\pi})\right)
    \label{eq:app_PDF_xd_s_prime-step3}
    \\
     & = \sum_{\bm k^{\prime}}  g\left((\bI_{2N} \oplus {(\bm M^\top \bm \Omega)}) \bm V_{\bx} (\bI_{2N} \oplus {(\bm M^\top \bm \Omega)} ^{\top}),(\bx_d, \bm s - (\bm k^{\prime}-\bm \ell_{\bm N^{-1} {\bm s}^\prime}) \sqrt{2\pi})\right)
     \label{eq:app_PDF_xd_s_prime-step4}
     \\
     &=\sum_{\bm k^{\prime\prime}} g\left((\bI_{2N} \oplus {(\bm M^\top \bm \Omega)}) \bm V_{\bx} (\bI_{2N} \oplus {(\bm M^\top \bm \Omega)} ^{\top}),(\bx_d, \bm s - \bm k^{\prime\prime}\sqrt{2\pi})\right)
     \label{eq:app_PDF_xd_s_prime-step5}
     \\
     &=P\left(\bx_d,{\bm s}\right).
\end{align}
\eqref{eq:app_PDF_xd_s_prime-step1} follows directly from the Eq.~\eqref{eq:xd-s-joint-distribution}. 
To move from \eqref{eq:app_PDF_xd_s_prime-step2} to \eqref{eq:app_PDF_xd_s_prime-step3}, we use the fact that $\bm N^{-1}$ is a unimodular matrix and acting with $\bm N^{-1}$ on an integer vector simply changes the summation index, $\bm N^{-1}\bm k \rightarrow \bm k^{\prime}$. In \eqref{eq:app_PDF_xd_s_prime-step4}, we use $ \bm N^{-1} {\bm s}^\prime=\bm s+\bm \ell_{\bm N^{-1} {\bm s}^\prime} \sqrt{2\pi}$. And in \eqref{eq:app_PDF_xd_s_prime-step5}, we change the summation index, which is over all integer vectors in $\mathbb{Z}^{2M}$. 
\end{proof}

We next show that the MMSE estimator $\tilde{\bx}_d\left(\bm s\right)= \bm f_{\rm MMSE}(\bm s)$ of Eq.~\eqref{eq:App-MMSE-estimator} is invariant under a basis transformation, but before doing so, we prove a useful lemma.
\begin{lemma}\label{lemma:invariant_integral}
Given a vector $\bm s\in\mathbb{R}^{2M}$, a unimodular matrix $\bm N$, and a function $J(\bm s)$, it follows that
\begin{equation}
\int_{\mathfrak{I}^{2M}}J(\bs)\dd{\left(R_{\sqrt{2\pi}}(\bm N \bs)\right)}=\int_{\mathfrak{I}^{2M}}J(\bs)\dd{\bs}.
\end{equation}
\end{lemma}
\begin{proof}
For brevity, we use the shorthand $R()$ for the modulo function $R_{\sqrt{2\pi}}()$. Given a unimodular matrix $\bm N$, the function $f:\bs \rightarrow R(\bm N \bs)$ is a bijection from $\mathfrak{I}^{2M}$ to $\mathfrak{I}^{2M}$, as we can find the inverse function $f^{-1}:R(\bm N \bs) \rightarrow \bs$ by $\bs = R\left(\bm N^{-1} R(\bm N \bs)\right)$.
We can then divide the region $\mathfrak{I}^{2M}$ into sub regions $A_{\bk}= \left\{\bx \in \mathfrak{I}^{2M}| \bm N \bx -\bk\sqrt{2\pi} \in \mathfrak{I}^{2M}\right\}$. There are finite number of these regions since $\bm N$ and $\bx$ are finite. Let $B_{\bk} = \left\{ f(\bx)| \bx \in A_{\bk} \right\}$ be the image of $A_{\bk}$. Then $\bx \in A_{\bk}$ is equivalent to $f(\bx)=\bm N \bx -\bk \sqrt{2\pi} \in B_{\bk}$ by the map $f$ and thus
\begin{align}
    \int_{\mathfrak{I}^{2M}}J(\bs)\dd{\left(f( \bs)\right)}&=\sum_{\bk} \int_{B_{\bk}}J(\bs)\dd{\left(f( \bs)\right)} \label{eq:lemma10_step1}
    \\
    &= \sum_{\bk} \int_{B_{\bk}}J(\bs)\dd{\left(\bm N \bs - \bk \sqrt{2\pi} \right)} \label{eq:lemma10_step2} \\
    & = \sum_{\bk} \int_{A_{\bk}}J(\bs)\dd{\bs} \label{eq:lemma10_step3}
    \\ 
    &= \int_{\mathfrak{I}^{2M}}J(\bs)\dd{\bs}, \label{eq:lemma10_step4}
\end{align}
where we use $|\bm N| = 1$ from \eqref{eq:lemma10_step2} to \eqref{eq:lemma10_step3}.
\end{proof}

\begin{theorem}[Invariance of MMSE]\label{thm:invariant_mmse}
    Consider a generator matrix $\bm M$ and change of basis by a unimodular matrix $\bm N$, which defines another generator matrix $\bm M^\prime=\bm M\bm N^\top$. Let $\bm s=R_{\sqrt{2\pi}} (\bm M^\top \bm \Omega \bx_a)$ be the error syndrome and consider the MMSE estimator $\tilde{\bx}_d\left(\bm s\right)= \bm f_{\rm MMSE}(\bm s)$. Likewise, let ${\bm s}^\prime = R_{\sqrt{2\pi}} (\bm N \bm M^\top \bm \Omega \bx_a)$ be the error syndrome in the new basis and the corresponding MMSE estimator $\tilde{\bx}_d^\prime\left(\bm s^\prime\right)= \bm f_{\rm MMSE}(\bm s^\prime)$. Then, $\tilde{\bx}_d^\prime\left(\bm s^\prime\right)=\tilde{\bx}_d\left(\bm s\right)$.
\end{theorem}

\begin{proof}
From Theorem~\ref{thm:invariant_pdf}, we have that $P\left(\bx_d,\bm s\right)=P^\prime(\bx_d,\bm s^\prime)$ and thus $P\left(\bx_d|\bm s\right) =P^\prime(\bx_d|\bm s^\prime)$ for the conditional distribution, such that
\begin{align*}
\tilde{\bx}_d^\prime(\bm s^\prime)&=\int_{\mathbb{R}^{2N}}\dd{\bm x_d}\; \bm x_d P^\prime(\bm x_d|\bm s^\prime)
\\ &= \int_{\mathbb{R}^{2N}}\dd{\bm x_d}\; \bm x_d P\left(\bm x_d|R_{\sqrt{2\pi}} \bm s\right)
\\&=\tilde{\bx}_d \left(\bm s\right),
\end{align*}
which was to be proved.
\end{proof}

\begin{corollary}
    The output covariance matrix $\bm V_{\rm out}$ describing the noise on the data modes (i.e., the second moments of data error) is invariant under a basis transformation on the ancillary lattice.
\end{corollary}
This corollary follows immediately from Theorem~\ref{thm:invariant_mmse}, but we explicitly derive such for clarity.
\begin{proof}
Consider the output covariance matrix in the new basis $\bm V_{\rm out}^\prime$ [defined in a similar fashion as Eq.~\eqref{eq:Vout_ij}] and the following set of equalities,
\begin{align}
[\bm V_{\rm out}^\prime]_{ij}& =\int_{\mathbb{R}^{2N}}\dd{\bm x_d} \int_{\mathfrak{I}^{2M}}\dd{\bm s^\prime}\; \left[\bx_d-\tilde{\bx}_d^\prime\left(\bm s^\prime\right)\right]_{i \times j} P^\prime(\bm x_d,\bm  s^\prime)\nonumber
\\
&=\int_{\mathbb{R}^{2N}}\dd{\bm x_d} \int_{\mathfrak{I}^{2M}}\dd{\bm s^\prime}\; \left[\bx_d-\tilde{\bx}_d \left(\bm s\right)\right]_{i \times j} P(\bm x_d,\bm s)\nonumber\\
&  = \int_{\mathbb{R}^{2N}}\dd{\bm x_d} \int_{\mathfrak{I}^{2M}}\dd{\left(R_{\sqrt{2\pi}} (\bm N \bm s)\right)}\; \left(\bx_d-\tilde{\bx}_d\left(\bm s\right)\right)_{i \times j} P(\bm x_d,\bm  s) \nonumber\\
& = \int_{\mathbb{R}^{2N}}\dd{\bm x_d} \int_{\mathfrak{I}^{2M}}\dd{\bm s}\; \left(\bx_d-\tilde{\bx}_d\left(\bm s\right)\right)_{i \times j} P(\bm x_d,\bm  s)\nonumber\\
&= [\bm V_{\rm out}]_{ij}. \nonumber
\end{align}
We have used Theorems~\ref{thm:invariant_pdf} and \ref{thm:invariant_mmse} to go from the first equality to the second equality and used the relation $\bm s^\prime = R_{\sqrt{2\pi}} \left(\bm N{\bm s}\right)$ and Lemma~\ref{lemma:invariant_integral} to go from the third equality to the fourth equality. The final equality follows by definition of the covariance matrix.
\end{proof}

We now discuss linear estimation and show that such is not invariant under a lattice basis transformation.
\begin{theorem}
    Linear estimation is not invariant under a lattice basis transformation.
\end{theorem}
\begin{proof} Considering only the $\bm{n}=\bm{0}$ element in the sum of Eq.~\eqref{eq:App-MMSE-estimator}, we get the linear estimator 
\begin{equation}
    \bm{f}_{\rm{Linear}}(\bm s) = -\bm V_d^{-1}\bm V_{da}\bm s,
    \label{eq:App-linear-estimator}
\end{equation}
which is obviously not invariant under a basis transformation. Indeed, consider a basis transformation such that $ {\bm s}^\prime= R_{\sqrt{2\pi}}(\bm N \bm s)$. The estimator in Eq.~\eqref{eq:App-linear-estimator} then changes to
\begin{align}
    \bm{f}_{\rm{Linear}}(\bm s^\prime) & = -\bm V_d^{-1}\bm V_{da}\bm s^\prime \nonumber\\
    & = -\bm V_d^{-1}\bm V_{da}R_{\sqrt{2\pi}}(\bm N \bm s).
    \label{eq:App-linear-estimator-change-bases}
\end{align}
and $\bm{f}_{\rm{Linear}}(\bm s)\neq\bm{f}_{\rm{Linear}}(\bm s^\prime)$. 
\end{proof}
As a consequence, the output covariance matrix for linear estimation also changes under a change of lattice basis. To show this explicitly, let us first write the output covariance matrix $\bm V_{\rm out}$ in full,
\begin{align}
    [\bm V_{\rm out}]_{ij} &= \int_{\mathbb{R}^{2N}}\dd{\bm x_d} \int_{\mathfrak{I}^{2M}}\dd{\bm s}\; {(\bx_d-\bm{f}_{\rm{Linear}}(\bm s))}_{i \times j} P(\bm x_d,\bm  s) \nonumber\\
    &= \sum_{\bm k} \int_{\mathfrak{I}^{2M}}\dd{\bm s} \int_{\mathbb{R}^{2N}}\dd{\bm x_d}\; {(\bx_d-\bm{f}_{\rm{Linear}}(\bm s))}_{i \times j}g(\bm V_d^{-1},\bm x_d+ \bm V^{-1}_d \bm V_{da}(\bm s-\bm k\sqrt{2\pi})) g(\bm V_{d|a}^{-1}, \bm s -\bm k\sqrt{2\pi})\nonumber\\
    &= \sum_{\bm k} \int_{\mathfrak{I}^{2M}}\dd{\bm s}\;  {[\bm V_d^{-1}]}_{ij}g(\bm V_{d|a}^{-1}, \bm s -\bm k\sqrt{2\pi})  +\left[-\bm V^{-1}_d \bm V_{da} \bs +\bm V^{-1}_d \bm V_{da}(\bm s-\bm k\sqrt{2\pi}) \right]_{i \times j}g(\bm V_{d|a}^{-1}, \bm s -\bm k\sqrt{2\pi})\nonumber\\
    & =  [\bm V_d^{-1}]_{ij} + \sum_{\bm k} 2\pi \left[\bm V_d^{-1} \bm V_{da} \bm k \right]_{i \times j} \int_{\mathfrak{I}^{2M}}\diff{\bm  s}\;g(\bm V_{d|a}^{-1},\bm s - \bm k \sqrt{2\pi}),
    \label{eq:linear-Vout_ij}
\end{align}
where we have used the expression of $P(\bx_d,\bm s)$ in~\eqref{eq:xd-s-joint-distribution} and~\eqref{eq:app_integral_simplification} in the first and second steps, respectively. Under change of basis, the new elements of the covariance matrix $\bm V_{\rm out}^\prime$ are
\begin{align}
    [\bm V_{\rm out}^\prime]_{ij} 
    &  =\int_{\mathbb{R}^{2N}}\dd{\bm x_d} \int_{\mathfrak{I}^{2M}}\dd{\bm s^\prime}\; \left[\bx_d-\bm{f}_{\rm{Linear}}(\bm s^\prime)\right]_{i \times j} P^\prime(\bm x_d,\bm  s^\prime)
    \label{VV_step1}
    \\
    &= \int_{\mathbb{R}^{2N}}\dd{\bm x_d} \int_{\mathfrak{I}^{2M}}\dd{\left(R_{\sqrt{2\pi}} \left(\bm N \bm s\right) \right)} \;[\bx_d+ \bm{V}_d^{-1}\bm V_{da} R_{\sqrt{2\pi}}(\bm N \bs) ]_{i \times j} P\left(\bx_d, \bm s\right)
    \label{VV_step2}
    \\
    &= \int_{\mathbb{R}^{2N}}\dd{\bm x_d} \int_{\mathfrak{I}^{2M}}\dd{\bm s} \;[\bx_d+ \bm{V}_d^{-1}\bm V_{da} R_{\sqrt{2\pi}}(\bm N \bs) ]_{i \times j} \; P\left(\bx_d, \bm s\right)
    \label{VV_step3}
    \\
    &= \sum_{\bm k } \int_{\mathbb{R}^{2N}}\dd{\bm x_d} \int_{\mathfrak{I}^{2M}}\dd{\bm s} \;[\bx_d+ \bm{V}_d^{-1}\bm V_{da} R_{\sqrt{2\pi}}(\bm N \bs) ]_{i \times j} \;g(\bm V_d^{-1},\bm x_d+ \bm V^{-1}_d \bm V_{da}(\bm s-\bm k\sqrt{2\pi})) g(\bm V_{d|a}^{-1}, \bm s -\bm k\sqrt{2\pi}) 
    \label{VV_step4}
    \\
    & = \sum_{\bm k} \int_{\mathfrak{I}^{2M}}\dd{\bm s}\; g(\bm V_{d|a}^{-1}, \bs -\bm k\sqrt{2\pi})\left\{[\bm V_d^{-1}]_{ij} + \left[\bm V_d^{-1}\bm V_{da}R_{\sqrt{2\pi}}(\bm N \bm s)- \bm V^{-1}_d \bm V_{da}(\bm s-\bk \sqrt{2\pi})\right]_{i \times j }\right\}
    \label{VV_step5}
    \\
    & = [\bm V_d^{-1}]_{ij} + \sum_{\bm k} \int_{\mathfrak{I}^{2M}}\dd{\bm s}\;g(\bm V_{d|a}^{-1}, \bs -\bm k\sqrt{2\pi}) \left[\bm V_d^{-1}\bm V_{da}R_{\sqrt{2\pi}}(\bm N \bm s)- \bm V^{-1}_d \bm V_{da}(\bm s-\bk \sqrt{2\pi})\right]_{i \times j }
    \label{VV_step6}
    \\
    & = [\bm V_d^{-1}]_{ij} + \sum_{\bm k} \int_{\mathfrak{I}^{2M}}\dd{\bm s}\;g(\bm V_{d|a}^{-1}, \bs -\bm k\sqrt{2\pi}) \left[
    \bm V^{-1}_d \bm V_{da}\bk \sqrt{2\pi}+
    \bm V_d^{-1}\bm V_{da}\left(R_{\sqrt{2\pi}}(\bm N \bm s)-\bm s\right)
    \right]_{i \times j }
    \label{VV_step7}
    \\
    & = \underbrace{[\bm V_d^{-1}]_{ij} + \sum_{\bm k} 2\pi \left[\bm V_d^{-1} \bm V_{da} \bm k \right]_{i \times j} \int_{\mathfrak{I}^{2M}}\diff{\bm  s}\;g(\bm V_{d|a}^{-1},\bm s - \bm k \sqrt{2\pi})}_{=[\bm V_{\rm out}]_{ij}} \nonumber\\
    &\qquad \qquad \;\; + \sum_{\bm k} \int_{\mathfrak{I}^{2M}}\diff{\bm  s}\;g(\bm V_{d|a}^{-1},\bm s - \bm k \sqrt{2\pi})\left[\bm V_d^{-1} \bm V_{da} \left(R_{\sqrt{2\pi}}(\bm N \bs)-\bs \right)  \right]_{i \times j} \nonumber\\
    &\qquad \qquad \;\; + \sum_{\bm k} \sqrt{2\pi} \left[\bm V_d^{-1} \bm V_{da} \bm k \right]_i \int_{\mathfrak{I}^{2M}}\diff{\bm  s}\;g(\bm V_{d|a}^{-1},\bm s - \bm k \sqrt{2\pi})\left[\bm V_d^{-1} \bm V_{da} \left(R_{\sqrt{2\pi}}(\bm N \bs)-\bs \right)  \right]_{j} \nonumber \\
    &\qquad \qquad \;\; + \sum_{\bm k} \sqrt{2\pi} \left[\bm V_d^{-1} \bm V_{da} \bm k \right]_j \int_{\mathfrak{I}^{2M}}\diff{\bm  s}\;g(\bm V_{d|a}^{-1},\bm s - \bm k \sqrt{2\pi})\left[\bm V_d^{-1} \bm V_{da} \left(R_{\sqrt{2\pi}}(\bm N \bs)-\bs \right)  \right]_{i}.
    \label{eq:linear-Vout_ij-change-base}
\end{align}
From \eqref{VV_step1} to \eqref{VV_step2}, we use $\bm s^\prime =R_{\sqrt{2\pi}} \left(\bm N \bm s\right)$, $P^\prime(\bx_d,\bs^\prime)=P(\bx_d,\bs)$ and the expression in~\eqref{eq:App-linear-estimator-change-bases}. From \eqref{VV_step2} to \eqref{VV_step3}, we apply Lemma~\ref{lemma:invariant_integral}. To obtain \eqref{VV_step4}, we input the explicit equation of \eqref{eq:xd-s-joint-distribution}. The integral of $\bx_d$ is simplified with~\eqref{eq:app_integral_simplification} from \eqref{VV_step4} to \eqref{VV_step5}. From \eqref{VV_step5} to \eqref{VV_step7}, we first recognize that the $\bm V_d^{-1}$ part does not depend on $\bm s$ and then reorganized the integrand. To obtain \eqref{VV_step7}, we make use of the definition $\bm x _{i\times j} = \bm x_i \bm x_j$, where $\bm x_i$ is the $i$th component of $\bm x$, so that $(\bm{a}+\bm{b})_{i\times j}=(\bm{a}+\bm{b})_i (\bm{a}+\bm{b})_j = \bm{a}_{i\times j}+ \bm{b}_{i\times j} + \bm a_i \bm b_j +\bm a_j \bm b_i$. Comparing equations~\eqref{eq:linear-Vout_ij} and~\eqref{eq:linear-Vout_ij-change-base}, we see that the linear estimation is not invariant under a lattice basis transformation.

\end{widetext}


\section{Lower bound on output variance for multimode codes}\label{app:bounds}
We find a lower bound on the output variance for multimode GKP codes (with $N$ data modes and $M\geq N$ GKP ancilla modes), in a similar vein to the single-mode lower bound found in Ref.~\cite{noh2020o2o}. We first provide an effective reduction from multimode additive noise to independent single-mode additive noises which we utilize to prove the result.

Given a $N$ mode input data state $\rho$, the output data state $\rho^\prime$ after a round of error correction with a multimode GKP stabilizer code (which protects against AGN) is related to the input by an additive non-Gaussian noise channel $\widetilde{\mathcal{N}}$ via
\begin{equation}\label{eq:ANGN_channel}
\rho^\prime=\widetilde{\mathcal{N}}\left(\rho\right)=\int_{\mathbb{R}^{2N}}\dd{\bm e_d}P(\bm e_d)D_{\bm e_d}\rho D_{\bm e_d}^\dagger,
\end{equation}
where $P(\bm e_d)$ is a multivariate non-Gaussian pdf that describes the residual data noise $\bm e_d$---e.g., $\expval{\bm e_d}=0$ and $\expval{\acomm{(\bm e_d)_i}{(\bm e_d)_j}}=(\bm V_{\rm out})_{ij}$; for brevity in what follows, we let $\bm V\equiv\bm V_{\rm out}$. Consider the Gaussian approximations $\rho_G$ and $\rho^\prime_G$ for the input and output states $\rho$ and $\rho^\prime$, which agree with $\rho$ and $\rho^\prime$ at the level of first and second moments. We have that
\begin{equation}
\rho^\prime_G=\mathcal{N}_{\bm V}\left(\rho_G\right)=\int_{\mathbb{R}^{2N}}\dd{\bm e_d}g(\bm V,\bm e_d)D_{\bm e_d}\rho_G D_{\bm e_d}^\dagger,
\end{equation}
where $g(\bm V,\bm e_d)$ is a Gaussian distribution according to Eq.~\eqref{g_definition} with the same first and second moments as the non-Gaussian pdf $P(\bm e_d)$ and $\mathcal{N}_{\bm V}$ is a multimode AGN channel with noise matrix $\bm V$. 

Let $\bm S$ be the symplectic transformation that diagonalizes $\bm V$ such that $\bm S\bm V\bm S^\top=\bigoplus_{i=1}^N\nu_i\bm I_2$, where $\nu_i\geq0$ are the symplectic eigenvalues (single-mode variances) of $\bm V$. We append Gaussian unitaries $U_{\bm S}^\dagger$ and $U_{\bm S}$ before and after error correction, respectively, and define the resulting channel $\widetilde{\mathcal{N}}^{(\bm S)}\equiv\mathcal{U}_{\bm S}\circ\widetilde{\mathcal{N}}\circ\mathcal{U}_{\bm S}^{-1}$, where $\mathcal{U}_{\bm S}$ is the unitary channel of $U_{\bm S}$. Then, 
\begin{align}
    \rho^\prime(\bm S)&=\widetilde{\mathcal{N}}^{(\bm S)}(\rho )
    \nonumber
    \\
    &=\int_{\mathbb{R}^{2N}}\dd{\bm e_d}P(\bm e_d)\left(U_{\bm S}D_{\bm e_d}U_{\bm S}^\dagger\right)\rho \left(U_{\bm S} D_{\bm e_d}^\dagger U_{\bm S}^\dagger\right)\nonumber\\
    &=\int_{\mathbb{R}^{2N}}\dd{\bm e_d}P(\bm e_d)D_{\bm S\bm e_d}\rho D_{\bm S\bm e_d}^\dagger\nonumber\\
    &=\int_{\mathbb{R}^{2N}}\dd{\bm e_d}P(\bm S^{-1}\bm e_d)D_{\bm e_d}\rho D_{\bm e_d}^\dagger,\label{eq:NtildeS_rho}
\end{align}
where Eq.~\eqref{eq:weyl_S} was used in the second line and a change of variables was used in the third line, along with the fact that $\dd(\bm S^{-1}\bm e_d)=\dd{\bm e_d}$. The Gaussian approximation for $\rho^\prime(\bm S)$ is then simply,
\begin{equation}\label{eq:gauss_approx_redux}
    \rho^\prime_G(\bm S)=\mathcal{N}_{\bm S\bm V\bm S^\top}\left(\rho_G\right)=\bigotimes_{i=1}^N\mathcal{N}_{\nu_i}(\rho_G).
\end{equation}
This is due to the fact that displacements are always local, $D_{\bm e_d}=\bigotimes_i D_{\bm e_{d_i}}$, and that $P(\bm S^{-1}\bm e_d)$ has a corresponding Gaussian distribution satisfying
\begin{align}
    g(\bm V,\bm S^{-1}\bm e_d)&=g(\bm S\bm V\bm S^\top,\bm e_d)\nonumber\\
    &=g\Big({\bigoplus}_i\nu_i\bm I_2,\bm e_d\Big)\nonumber\\
    &=\prod_i g(\nu_i\bm I_2, \bm e_{d_i}),
\end{align}
where $\bm e_{d_i}\in\mathbb{R}^2$. In other words, at the level of first and second moments, we have decorrelated the residual noises on the data modes via $\bm S$ and reduced the channel to a direct product of independent AGN channels. We will use this simplification in what follows to prove a lower bound on the quantum capacity of a general non-Gaussian additive noise channel $\widetilde{\mathcal{N}}$, which includes the single-mode case proven in Ref.~\cite{noh2020o2o}.

\begin{lemma}\label{lemma:multimode_lb}
    Consider a $N$ mode additive non-Gaussian noise channel $\widetilde{\mathcal{N}}$ [see Eq.~\eqref{eq:ANGN_channel}]. Let the geometric mean error of the channel be defined as ${\bar{\sigma}_{\rm GM}^2=\sqrt[2N]{\det\bm V}}$, where $\bm V$ is the $2N\times2N$ covariance matrix of $\widetilde{\mathcal{N}}$. Then the quantum capacity $C_{\mathcal{Q}}$ of $\widetilde{\mathcal{N}}$ has the following lower bound,
    \begin{equation}\label{eq:multimode_lb}
        C_{\mathcal{Q}}(\widetilde{\mathcal{N}})\geq\max\left[0,N\log_2\left(\frac{1}{e\bar{\sigma}^{2}_{\rm GM}}\right)\right].
    \end{equation}
\end{lemma}

\begin{proof}
\QZ{
To prove this result, we first consider the channel $\widetilde{\mathcal{N}}^{(\bm S)}$ of Eq.~\eqref{eq:NtildeS_rho}. Since $\widetilde{\mathcal{N}}^{(\bm S)}$ is related to $\widetilde{\mathcal{N}}$ by unitary pre- and post-processing, the quantum capacities of the two channels are equivalent, $C_{\mathcal{Q}}(\widetilde{\mathcal{N}})=C_{\mathcal{Q}}(\widetilde{\mathcal{N}}^{(\bm S)})$. We focus on the channel $\widetilde{\mathcal{N}}^{(\bm S)}$ from hereon. 
}

\QZ{
Observe that the single-shot coherent information places a lower bound on the quantum capacity~\cite{schumacher1996qec,lloyd1997capacity,devetak2005capacity}
    \begin{equation}\label{eq:capacity_coherentinfo}
   C_{\mathcal{Q}}(\widetilde{\mathcal{N}}^{(\bm S)})\geq\max_{\rho}I_c(\rho,\widetilde{\mathcal{N}}^{(\bm S)})
    \end{equation}
    where $I_c(\rho,\widetilde{\mathcal{N}}^{(\bm S)})=S(\widetilde{\mathcal{N}}^{(\bm S)}(\rho))-S(\widetilde{\mathcal{N}}^{(\bm S)\,c}(\rho))$ is the coherent information, $S$ is the von Neumann entropy and $\widetilde{\mathcal{N}}^{(\bm S)\,c}$ is the complementary channel of $\widetilde{\mathcal{N}}^{(\bm S)}$. By Gaussian extremality~\cite{wolf2006extremality}, one can show that $I_c(\rho,\widetilde{\mathcal{N}}^{(\bm S)})\geq S(\widetilde{\mathcal{N}}^{(\bm S)}(\rho)_G)-S(\widetilde{\mathcal{N}}^{(\bm S)\,c}(\rho)_G)$ where, e.g., $\widetilde{\mathcal{N}}^{(\bm S)}(\rho)_G$ is a Gaussian state with the same first and second moments as $\widetilde{\mathcal{N}}^{(\bm S)}(\rho)$. In fact, $\widetilde{\mathcal{N}}^{(\bm S)}(\rho)_G=\bigotimes_{i=1}^N\mathcal{N}_{\nu_i}(\rho_G)$ from Eq.~\eqref{eq:gauss_approx_redux}, which is just a direct product of AGN channels acting on an input Gaussian state $\rho_G$; similarly, $\widetilde{\mathcal{N}}^{(\bm S)\,c}(\rho)_G=\bigotimes_{i=1}^N\mathcal{N}_{\nu_i}^c(\rho_G)$. Therefore, $I_c(\rho,\widetilde{\mathcal{N}}^{(\bm S)})\geq I_c(\rho_G,\bigotimes_{i=1}^N\mathcal{N}_{\nu_i})$. In other words, we have reduced the problem to analyzing a set of independent AGN channels with Gaussian input states; this has been done in previous works for a single mode with an input thermal state $\Theta_{\bar{n}}$~\cite{holevo2001bosonic_channels,albert2018GKPcapacity}. In the energy unconstrained setting ($\bar{n}\rightarrow\infty$), the single mode result is $\max_{\rho_G}I_c(\rho_G,\mathcal{N}_{\nu_i})=\log(1/e\nu_i)$. Hence, for the multimode channel $\bigotimes\mathcal{N}_{\nu_i}$,
    \begin{align}
    \max_{\rho_G}I_c\left(\rho_G,\bigotimes_{i=1}^N\mathcal{N}_{\nu_i}\right)&\geq\sum_{i=1}^N\log(\frac{1}{e\nu_i})\\
    &=\log(\frac{1}{e^N\prod_{i=1}^N\nu_i}).
    \end{align}
    Since $\nu_i$ are the symplectic eigenvalues of the channel covariance matrix $\bm V$, $\det\bm V=\prod_{i=1}^N\nu_i^2$, and the result is proved.
    }
\end{proof}

Now we consider the upper bound of the quantum capacity. We begin by quoting a lemma in Refs.~\cite{albert2018GKPcapacity,wu2021continuous}.
\begin{lemma}\label{lemma:capacity_ub}
    Given an $N+M$ mode AGN channel, $\mathcal{N}_{\bm Y}$, with noise covariance matrix $\bm Y\geq0$, the quantum capacity $C_{\mathcal{Q}}$ then has the following upper bound~\cite{albert2018GKPcapacity,wu2021continuous},
    \begin{equation}\label{eq:multimode_AGN_ub}
C_{\mathcal{Q}}\left(\mathcal{N}_{\bm Y}\right)\leq \sum_{i=1}^{N+M}\log_2\left(\frac{1-\sigma_i^2}{\sigma_i^2}\right),
    \end{equation}
    where $\sigma^2_i$ are the symplectic eigenvalues of $\bm Y$ (i.e., the single-mode variances) . 
\end{lemma}

From the above results of lower and upper bounds, we can have the following theorem.
\begin{theorem}[Lower bound for output noise of multimode codes]\label{thm:lb_logicalnoise}
    Consider a multimode GKP stabilizer code which consumes $M\geq N$ ancillary GKP modes to protect $N$ data modes against independent AGN. The joint noise channel across all modes is given by $\bigotimes_{i=1}^{N+M}\mathcal{N}_{\sigma_i^2}$ with variances $\sigma_i^2$. Let the $2N\times 2N$ output noise matrix of the data be $\bm V_{\rm out}$, with RMS and GM errors given as $\bar{\sigma}^2_{\rm RMS}=\Tr{\bm V_{\rm out}}/2N$ and $\bar{\sigma}_{\rm GM}^2=\sqrt[2N]{\det\bm V_{\rm out}}$, respectively. Then,
    \begin{equation}\label{eq:lb_logicalnoise}
        \bar{\sigma}_{\rm RMS}\geq\bar{\sigma}_{\rm GM}\geq\frac{1}{\sqrt{e}}\sqrt[2N]{\left(\prod_{i=1}^{N+M}\frac{\sigma_i^2}{1-\sigma_i^2}\right)}.
    \end{equation}
\end{theorem}
\begin{proof}
    The proof for the geometric mean error follows by concurrently considering the lower bound and upper bound of Lemmas~\ref{lemma:multimode_lb} and~\ref{lemma:capacity_ub}, respectively. In detail, we can use a multimode GKP code---and thus the resultant additive non-Gaussian channel $\widetilde{\mathcal{N}}$---to transmit information across the AGN channels $\bigotimes_{i=1}^{N+M}\mathcal{N}_{\sigma_i^2}$. The transmission rate is bounded from below by the quantum capacity $C_{\mathcal{Q}}(\widetilde{\mathcal{N}})$ since $C_{\mathcal{Q}}(\widetilde{\mathcal{N}})$ is achievable, which in turn is bounded from below by Lemma~\ref{lemma:multimode_lb}. Likewise, the transmission rate is bounded from above by the quantum capacity of the original channel, $C_{\mathcal{Q}}(\bigotimes_{i=1}^{N+M}\mathcal{N}_{\sigma_i^2})$, which in turn is bounded from above by Lemma~\ref{lemma:capacity_ub}. This establishes the second inequality in Eq.~\eqref{eq:lb_logicalnoise}. The bound on the RMS error [the first inequality in Eq.~\eqref{eq:lb_logicalnoise}] follows from the fact that the arithmetic mean is an upper bound on the geometric mean.
\end{proof}
\begin{corollary}\label{cor:iid_bound}
    For iid AGN, such that $\sigma_i=\sigma\,\forall\,i\in\{1,\dots,M+N\}$, $\bar{\sigma}_{\rm RMS}\geq\bar{\sigma}_{\rm GM}\geq\frac{1}{\sqrt{e}}\left(\frac{\sigma^{2}}{1-\sigma^{2}}\right)^{\frac{N+M}{2N}}$.
\end{corollary}

Thus, for general multimode codes, if the number of ancilla modes is equal to the number of data modes ($M=N$), then error suppression is at most quadratic in the initial noise $\sigma$. The performance can be further enhanced with concatenated codes ($M> N$), i.e. $\Bar{\sigma}_{\rm GM}\sim \sigma^{1+\frac{M}{N}}$. We find an upper bound on the \QZ{break-even point}, $\sigma^\star\leq1/\sqrt{2}$; see Fig.~\ref{fig:lowerbound} of the main text. For $\sigma>1/\sqrt{2}$, there is no gain to be had from QEC. This agrees with the fact that the upper bound for the quantum capacity of the AGN channel [Eq.~\eqref{eq:multimode_AGN_ub}] vanishes as $\sigma\rightarrow1/\sqrt{2}$.

\QZ{
\section{Classical entropies for the iid Gaussian channel}\label{app:classical_it}
The differential entropy of multivariate normal distribution is~\cite{cover2006elements}
\begin{equation}
    S(\bx) = \frac{1}{2}\ln\left[(2\pi e)^K\det\bm V\right],
\end{equation}
where $K$ is the dimension of the random variable $\bx$. In the case of $N$ data modes and $M$ ancilla modes [$K=2(M+N)$], the correlated data and ancilla noises (just prior to corrective displacements on the data) has a covariance matrix
\begin{widetext}
\begin{align}
    \begin{split}
    \bm V_{\bx} &= \bm{S}_{\rm{enc}}^{-1}\bm{V}_\xi\bm{S}_{\rm{enc}}^{-\top}\\
    &= (\bm \Lambda_d^{-1} \oplus \bm \Lambda_a^{-1})
\left[\sigma^2\left(\bigoplus_{i=1}^{N}\bm S_{G_i}\bm S_{G_i}^{\top}\right)\oplus\bm I_{2(M-N)}
    \right]
    (\bm \Lambda_d^{-\top} \oplus \bm \Lambda_a^{-\top}),
    \end{split}
\end{align}
\end{widetext}
where $\bm\Lambda_d$ and $\bm\Lambda_a$ are local symplectic transformations on the data and ancilla, respectively. The general form in terms of two-mode squeezing blocks, for an arbitrary Gaussian encoding $\bm S_{\rm enc}$, is a consquence of Lemma~\ref{lemma:lemma_redux}.
}

\QZ{
Since the local symplectic transforms do not change the determinants of the subsystems, we can easily calculate all entropies of interest, 
\begin{align}
&S(\bx_d)=\sum_{i=1}^N \ln{[2\pi e (2G_i-1)\sigma^2]},\label{eq:S_xd}\\
&S(\bx_a)=\sum_{i=1}^N \ln{[2\pi e (2G_i-1)\sigma^2]}+(M-N) \ln(2\pi e \sigma^2),\label{eq:S_xa}\\
&S(\bx_d,\bx_a)=(N+M)\ln(2\pi e \sigma^2).\label{eq:S_xdxa}
\end{align}
The mutual information of the joint distribution is then
\begin{align}
I(\bx_d;\bx_a)&=S(\bx_d)+S(\bx_a)-S(\bx_d,\bx_a)\nonumber\\
&= 2 \sum_{i=1}^{N}\ln{(2G_i-1)},
\end{align}
and the conditional differential entroy is 
\begin{align}
S(\bx_d|\bx_a)&=S(\bx_d,\bx_a)-S(\bx_a)\nonumber\\
&=\sum_{i=1}^{N}\ln{\left(\frac{2\pi e \sigma^2}{2G_i-1}\right)}.
\end{align}
}

\QZ{
In obtaining these results, we have utilized the fact that $\bm\Lambda_d\oplus\bm\Lambda_a$ does not change any of the above quantities. So in the calculations, we can actually consider the random variables $\bm\Lambda_d\bm x_d$ and $\bm\Lambda_a\bm x_a$ for simplicity. In terms of these variables, the data modes (and ancilla modes) are uncorrelated amongst themselves, and there are only two-mode correlations for each data-ancilla mode pair. Moreover, due to the structure of two-mode squeezing, there are only $qq$ and $pp$ correlations and no $qp$ correlations. Therefore, the total entropy, mutual information, and conditional entropy originate from: (1) two-mode pairs (one data and one ancilla) due to two-mode squeezing and (2) individual quadratures ($qq$ or $pp$ correlations but no $qp$ correlations). As an example, $S(\bm x_d|\bm x_a)=S(\bm\Lambda_d\bm x_d|\bm\Lambda_a\bm x_a)=\sum_{i=1}^{2N} S(\bm\Lambda_d\bm x_d|\bm\Lambda_a\bm x_a)_i$, where $S(\bm\Lambda_d\bm x_d|\bm\Lambda_a\bm x_a)_i=\ln{\left(\frac{2\pi e \sigma^2}{2G_i-1}\right)}/2$ is the conditional entropy for the $qq$ ($pp$) quadrature correlations of the $i$th data-ancilla mode pair. Everything considered, we see that we can treat each data quadrature independently and sum their individual contributions to calculate global properties.
}

\begin{widetext}
\QZ{
\section{Error from approximated states}
\label{app_analysis_approximatedGKP}
Given the noise model for finite GKP squeezing discussed in Section~\ref{sec:approx_gkp}, it is not hard to show that the error syndrome changes to
\begin{align}
    \bs &  =  R_{\sqrt{2\pi}}\left(\bm M^\top \bm \Omega (\bm\xi^{(1)}_{\rm {GKP}}+\bx_a) + \bm\xi^{(2)}_{\rm {GKP}}\right)\\
    &= R_{\sqrt{2\pi}}\left(\bm M^\top \bm \Omega \bx_a + \bm \delta\right),
\end{align}
where $\bm\xi^{(1)}_{\rm{GKP}}\sim\mathcal{N}(0,\sigma_{\rm GKP}^2)$ is the GKP noise from the ancilla modes, $\bm\xi^{(2)}_{\rm{GKP}}\sim\mathcal{N}(0,\sigma_{\rm GKP}^2)$ is the GKP noise from the noisy stabilizer measurements, and $\bm\delta=\bm M^\top\bm\Omega\bm\xi^{(1)}_{\rm{GKP}}+\bm\xi^{(2)}_{\rm{GKP}}$. The PDF of $\bm \delta$ is
\begin{equation}
    P(\bm \delta) = g\left(\sigma_{\rm {GKP}}^2(\bm M^\top \bm M+\bm I_{2M}),\bm \delta\right).
\end{equation}
}

\QZ{
We can show that the joint distribution of $\bx_d$ and the error syndrome $\bs$ is 
\begin{align}
    P(\bx_d,\bm s) & = \int_{\mathbb{R}^{4M}} \dd\left(\bm M^\top \bm \Omega \bx_a\right)  \dd \bm \delta \;\; P(\bx_d,\bm M^\top \bm \Omega \bx_a) P(\bm \delta) \sum_{\bm k}\delta\left(\bm s - \bm M^\top \bm \Omega \bx_a-\bm \delta -\bm k \sqrt{2\pi} \right) 
    \label{eq:xd-s-joint-distribution–step1_appx}
    \\
     & = \sum_{\bm k} \int_{\mathbb{R}^{2M}} \dd \bm \delta \;\; g\left((\bI_{2N} \oplus {\bm M^\top \bm \Omega}) \bm V_{\bx} (\bI_{2N} \oplus {(\bm M^\top \bm \Omega)} ^{\top}),(\bx_d,\bm s -\bm \delta- \bm k \sqrt{2\pi})\right) g\left(\sigma_{\rm {GKP}}^2(\bm M^\top \bm M+\bm I_{2M}),\bm \delta\right)
     \label{eq:xd-s-joint-distribution-step2_appx}
     \\
    & = \sum_{\bm k} g\left((\bI_{2N} \oplus {\bm M^\top \bm \Omega}) \bm V_{\bx} (\bI_{2N} \oplus {(\bm M^\top \bm \Omega)} ^{\top}+(\bm{0}_{2N}\oplus \sigma_{\rm {GKP}}^2(\bm M^\top \bm M+\bm I_{2M}) ) ),(\bx_d,\bm s - \bm k \sqrt{2\pi})\right)
     \label{eq:xd-s-joint-distribution-step3_appx}
     \\
    & \equiv  \sum_{\bm k} g(\bm 
              V_d^{-1},\bm x_d+ \bm V^{-1}_d \bm V_{da}(\bm s-\bm k\sqrt{2\pi})) g(\bm V_{d|a}^{-1}, \bm s -\bm k\sqrt{2\pi}),
    \label{eq:xd-s-joint-distribution_appx}
\end{align}
where we have redefined the notation $\bm V_d, \bm V_{da}, \bm V_{d|a}$ to incorporate finite-squeezing GKP noise.
}

\QZ{
We use a moment-generating function to prove the identities~\eqref{eq:xd-s-joint-distribution-step2_appx} and \eqref{eq:xd-s-joint-distribution-step3_appx}. Let $\bm X$ and $\bm Y$ be the vectors of independent random Gaussian variables, then 
\begin{align}
    &P\left(\bm X = (\bx_d, \bm s - \bm k \sqrt{2\pi})\right)=g\left((\bI_{2N} \oplus {\bm M^\top \bm \Omega}) \bm V_{\bx} (\bI_{2N} \oplus {(\bm M^\top \bm \Omega)} ^{\top}),(\bx_d,\bm s - \bm k \sqrt{2\pi})\right),\\
    &P\left(\bm Y = (\bm{0}_{2N}, \bm \delta )\right)=g\left((\bm{0}_{2N}\oplus \sigma_{\rm {GKP}}^2(\bm M^\top \bm M+\bm I_{2M}) ),(\bm{0}_{2N},\bm \delta)\right).
\end{align}
The corresponding moment generating functions are
\begin{align}
    &M_{\bm X}(\bm t)=\exp{-\frac{1}{2}\bm t^\top (\bI_{2N} \oplus {\bm M^\top \bm \Omega}) \bm V_{\bx} (\bI_{2N} \oplus {(\bm M^\top \bm \Omega)} ^{\top})  \bm t},\\
    &M_{\bm Y}(\bm t)=\exp{-\frac{1}{2}\bm t^\top (\bm{0}_{2N}\oplus \sigma_{\rm {GKP}}^2(\bm M^\top \bm M+\bm I_{2M}) )  \bm t}.
\end{align}
Let $\bm Z=\bm X+\bm Y$. Since $\bm X$ and $\bm Y$ are independent, then
\begin{align}
    M_{\bm Z}(\bm t)&=M_{\bm X}(\bm t)M_{\bm Y}(\bm t)\\
    & = \exp{-\frac{1}{2}\bm t^\top \left[(\bI_{2N} \oplus {\bm M^\top \bm \Omega}) \bm V_{\bx} (\bI_{2N} \oplus {(\bm M^\top \bm \Omega)} ^{\top})+(\bm{0}_{2N}\oplus \sigma_{\rm {GKP}}^2(\bm M^\top \bm M+\bm I_{2M}) )\right]  \bm t},
\end{align}
and we have 
\begin{align}
    \eqref{eq:xd-s-joint-distribution-step2}=P(\bm Z =\bm X+\bm Y) =\eqref{eq:xd-s-joint-distribution-step3_appx}.
\end{align}
}
\end{widetext}
\QZ{
This complete the proof.}

\end{document}